\RequirePackage{fix-cm}

\documentclass[a4paper,UKenglish,cleveref, autoref, thm-restate]{lipics-v2021}
\usepackage{todonotes}

\newtheorem{theorem1}{Theorem}
\newtheorem{RR}[theorem1]{Reduction Rule}
\newtheorem{question}[theorem]{Question}

\usepackage{verbatim}
\usepackage[T1]{fontenc}

\usepackage{bold-extra}
\usepackage{lmodern}
\usepackage{enumitem}
\DeclareFontShape{T1}{lmr}{b}{sc}{<->ssub*cmr/bx/sc}{}
\DeclareFontShape{T1}{lmr}{bx}{sc}{<->ssub*cmr/bx/sc}{}

\usepackage{bm}
\usepackage{caption}
\usepackage[ruled,vlined, linesnumbered]{algorithm2e}

\newcommand{\FPT}{$\mathsf{FPT}$}
    \usepackage{tabularx}
    \newcommand{\VDP}{$\textsc{VDP}$}
\newcommand{\EDP}{$\textsc{EDP}$}
\newcommand{\NPoly}{\textsf{NP} $\subseteq$ \textsf{coNP}$/$\textsf{poly}\xspace}
\nolinenumbers
\newcommand{\todom}[1]{\todo[linecolor=blue,backgroundcolor=blue!25,bordercolor=blue]{#1}}

\title{Kernels for the Disjoint Paths Problem on Subclasses of Chordal Graphs}

\titlerunning{Kernels for the Disjoint Paths Problem on Subclasses of Chordal Graphs} 
\author{Juhi Chaudhary}{Ben-Gurion University of the Negev, Beersheba, Israel  \and \url{https://sites.google.com/view/juhichaudhary/home}}{juhic@post.bgu.ac.il}{https://orcid.org/0000-0001-5560-9129}{}

\author{Harmender Gahlawat}{Ben-Gurion University of the Negev, Beersheba, Israel \and \url{https://sites.google.com/view/harmendergahlawat/} }{harmendergahlawat@gmail.com}{https://orcid.org/0000-0001-7663-6265}{}

\author{Michal Włodarczyk}{University of Warsaw, Warsaw, Poland \and \url{https://www.mimuw.edu.pl/~mw277619/}}{michal.wloda@gmail.com}{https://orcid.org/0000-0003-0968-8414}{}

\author{Meirav Zehavi}{Ben-Gurion University of the Negev, Beersheba, Israel \and \url{https://sites.google.com/site/zehavimeirav/}}{zehavimeirav@gmail.com}{https://orcid.org/0000-0002-3636-5322}{}

\authorrunning{J. Chaudhary et al.}
\ccsdesc[100]{Theory of Computation $\rightarrow$ Parameterized Complexity and Exact Algorithms; Mathematics of Computing $\rightarrow$ Graph Algorithms
} 

\keywords{Kernelization, Parameterized Complexity, Vertex-Disjoint Paths Problem, Edge-Disjoint Paths Problem.} 

\funding{All the authors are supported by the European Research Council (ERC) project titled PARAPATH.}
\relatedversion{} 
\Copyright{J. Chaudhary and H. Gahlawat and M. Włodarczyk and M. Zehavi}
\begin{document}

\maketitle

\begin{abstract}
Given an undirected graph $G$ and a multiset of $k$ terminal pairs $\mathcal{X}$, the \textsc{Vertex-Disjoint Paths} (\VDP) and \textsc{Edge-Disjoint Paths} (\EDP) problems ask whether $G$ has $k$ pairwise internally vertex-disjoint paths and $k$ pairwise edge-disjoint paths, respectively, connecting every terminal pair in~$\mathcal{X}$. In this paper, we study the kernelization complexity of \VDP~and~\EDP~on subclasses of chordal graphs. For \VDP, we design a $4k$ vertex kernel on split graphs and an $\mathcal{O}(k^2)$ vertex kernel on well-partitioned chordal graphs. We also show that the problem becomes polynomial-time solvable on threshold graphs.
For \textsc{EDP}, we first prove that the problem is $\mathsf{NP}$-complete on complete graphs. Then, we design an $\mathcal{O}(k^{2.75})$ vertex kernel for \EDP~on split graphs, and improve it to a $7k+1$ vertex kernel on threshold graphs. Lastly, we provide an $\mathcal{O}(k^2)$ vertex kernel for \EDP~on block graphs and a $2k+1$ vertex kernel for clique paths. Our contributions improve upon several results in the literature, as well as resolve an open question by Heggernes et al.~[Theory Comput. Syst., 2015].

\end{abstract}

\section{Introduction}\label{sec:intro}
The \textsc{Vertex-Disjoint Paths} (\VDP) and \textsc{Edge-Disjoint Paths} (\EDP) problems are fundamental routing problems, having applications in VLSI design and virtual circuit routing \cite{frank1990packing,EDPHardTreewidth,schrijver2003combinatorial,srinivas2005finding}. Notably, they have been a cornerstone of the groundbreaking Graph Minors project of Robertson and Seymour~\cite{DPMinor}, and several important techniques, including the \textit{irrelevant vertex technique}, originated in the process of solving disjoint paths~\cite{DPMinor}. In \VDP~(respectively, \EDP), the input is an undirected graph $G$ and a multiset of terminal pairs $\mathcal{X} = \{(s_1,t_1), \ldots, (s_k,t_k)\}$, and the goal is to find $k$ pairwise internally vertex-disjoint (respectively, edge-disjoint) paths $P_1,\ldots, P_k$ such that $P_i$ is a path with endpoints $s_i$ and $t_i$.

Both \VDP~and \EDP~are extensively studied in the literature, and have been at the center of numerous results in algorithmic graph theory~\cite{DPApproxHard,DPApproxTW,DPApprox1,DPApprox2,DPDirected,DPApprox3,VDPSTOCPlanar,DPPartialTrees}. Karp~\cite{VDPHard} proved that \VDP~is $\mathsf{NP}$-complete (attributing the result to Knuth), and a year later, Even, Itai, and Shamir~\cite{EDPHard} proved the same for \EDP. When $k$ is \emph{fixed} (i.e., treated as a constant), Robertson and Seymour~\cite{lokshtanov2020efficient, DPMinor} gave an $\mathcal{O}(|V(G)|^3)$ time algorithm as a part of their famous Graph Minors project. This algorithm is a \textit{fixed-parameter tractable} (\FPT) algorithm parameterized by $k$. Later, the power $3$ was reduced to $2$ by Kawarabayashi, Kobayashi, and Reed~\cite{fastestVDP}. 

In Parameterized Complexity, each problem instance is associated with an integer parameter $k$. We study both \textsc{VDP} and \textsc{EDP} through the lens of \textit{kernelization} under the parameterization by $k$. A \textit{kernelization algorithm} is a polynomial-time algorithm that takes as input an instance $(I,k)$ of a problem and outputs an \textit{equivalent instance} $(I',k')$ of the same problem such that the size of $(I',k')$ is bounded by some computable function $f(k)$. The problem is said to admit an $f(k)$ sized kernel, and if $f(k)$ is polynomial, then the problem is said to admit a polynomial kernel. It is well known that a problem is \FPT~if and only if it admits a kernel. Due to the profound impact of preprocessing, kernelization has been termed ``\emph{the lost continent of polynomial time}''~\cite{kernelApplication}.  For more details on kernelization, we refer to  books~\cite{bookParameterized, bookKernelization}. 

Bodlaender et al.~\cite{bodlaender} proved that, unless \NPoly, \VDP~does not admit a polynomial kernel (on general graphs). On the positive side, Heggernes et al.~\cite{hegger} extended this study to show that \VDP~and \EDP~admit polynomial kernels on split graphs with $\mathcal{O}(k^2)$ and $\mathcal{O}(k^3)$ vertices, respectively. Yang et al.~\cite{yang} further showed that a restricted version of \VDP, where each vertex can appear in at most one terminal pair, admits a $4k$ vertex kernel. Recently, Ahn et al.~\cite{ahn1} introduced so-called \textit{well-partitioned chordal graphs} (being a generalization of split graphs), and showed that \VDP~on these graphs admits an $\mathcal{O}(k^3)$ vertex kernel. 
In this paper, we extend the study of kernelization of \EDP~and~\VDP~on these (and other) subclasses of chordal graphs.

\begin{table}[t]
\scalebox{0.82}		{
	\begin{tabular}{|l|l|l|}
	\hline
\Large{\textbf{\underline{Graph Class}}} & \Large{\textbf{\underline{\VDP}}} & \Large{\textbf{\underline{\EDP}}} \\
\textbf{Well-partitioned} \textbf{Chordal} & $\bm{\mathcal{O}(k^{2})}$ \textbf{vertex kernel}  \textbf{[Theorem} \textbf{\ref{labelwpcvdp}]}  & \textcolor{red}{\textbf{OPEN}} \\ 

\textbf{Split} & $\bm{4k}$ \textbf{vertex kernel} \textbf{[Theorem} \textbf{\ref{labelsplitvdp}]} &  $\bm{\mathcal{O}(k^{2.75})}$ \textbf{vertex kernel} \textbf{[Theorem} \textbf{\ref{labelsplitedp}]}\\

\textbf{Threshold} &  \textcolor{blue}{\textbf{P}} \textbf{[{Theorem}} \textbf{\ref{thm:threshold-vdp-poly}]} & $\bm{7k+1}$ \textbf{vertex kernel} \textbf{[Theorem} \textbf{\ref{T:Threshold}]}\\

\textbf{Block} & \textcolor{blue}{\textbf{P}} \textbf{[{Observation}} \textbf{\ref{blockP}]} & $\bm{4k^{2}-2k}$ \textbf{vertex kernel} \textbf{[Theorem} \textbf{\ref{T:Block}]} \\

\textbf{Clique Path} & \textcolor{blue}{\textbf{P}} \textbf{[{Observation}} \textbf{\ref{blockP}]}  &  $\bm{2k+1}$ \textbf{vertex kernel} \textbf{[Theorem} \textbf{\ref{T:CP}]} \\ 

		\hline

	\end{tabular}
	}
		\caption{Summary of the kernelization results of \textsc{VDP} and \textsc{EDP} parameterized by the number of occurrences of terminal pairs ($k$) on the subclasses of chordal graphs studied in this paper. }
		
			\label{tablepig}
			
\end{table}

\subsection{Our Contribution}\label{sec:contri}

An overview of our results is given in Table~\ref{tablepig}.
We begin by discussing the results about \textsc{EDP}.
First, we observe that the problem remains \textsf{NP}-hard even on inputs with a trivial graph structure given by a clique, unlike \textsc{VDP}.
This extends the known hardness results for split graphs~\cite{hegger} and graphs of cliquewidth at most~6~\cite{DPCliqueWidth}.

 \begin{restatable}{theorem}{complete} \label{npc}
\textsc{EDP} is $\mathsf{NP}$-hard on complete graphs.
\end{restatable}

Every graph class treated in this paper includes cliques, so \textsc{EDP} is \textsf{NP}-hard on each of them.
This motivates the study of kernelization algorithms.
From now on, we always use $k$ to denote the number of occurrences of terminal pairs in an instance.

We present an $\mathcal{O}(k^{2.75})$ vertex kernel for \textsc{EDP} on split graphs, improving upon the $\mathcal{O}(k^{3})$ vertex kernel given by Heggernes et al.~\cite{hegger}. 
Our main technical contribution is a lemma stating that the length of each path in a minimum-size solution is bounded by $\mathcal{O}(\sqrt{k})$.
This allows us to obtain the following.

 \begin{restatable}{theorem}{splitedp} \label{labelsplitedp}
\textsc{EDP} on split graphs admits a kernel with at most $\mathcal{O}(k^{2.75})$ vertices. 
\end{restatable}

In the quest to achieve better bounds, we consider a subclass of split graphs. Specifically, we prove that \textsc{EDP} on threshold graphs admits a kernel with at most $7k+1$ vertices. Here, we exploit the known vertex ordering of threshold graphs that exhibits an inclusion relation concerning the neighborhoods of the vertices. 

 \begin{restatable}{theorem}{threshold}\label{T:Threshold}
\textsc{EDP} on threshold graphs admits a kernel with at most $7k+1$ vertices. 
\end{restatable}

Another important subclass of chordal graphs is the class of block graphs.
For this case, we present a kernel with at most $4k^2-2k$ vertices. 
Our kernelization algorithm constructs an equivalent 
instance where the number of blocks can be at most $4k-2$, and each block contains at most $k$ vertices. Thus, we have the following theorem.

\begin{restatable}{theorem}{block}\label{T:Block}
\EDP~on block graphs admits a kernel with at most $4k^2-2k$ vertices. 
\end{restatable}

Whenever a block has more than two cut vertices, decreasing the size of that block below $\mathcal{O}(k)$ becomes trickier. However, if we restrict our block graph to have at most two cut vertices per block---i.e., if we deal with \emph{clique paths}---then this can be done. The key point in designing our linear kernel in clique paths is that, in the reduced instance, for each block $B$, the number of vertices in $B$ is dictated by a linear function of the number of terminal pairs having at least one terminal vertex in $B$. 
So,
we obtain a $2k+1$ vertex kernel for this class.


\begin{restatable}{theorem}{cliquepath}\label{T:CP}
\EDP~on clique paths admits a kernel with at most $2k+1$ vertices. 
\end{restatable}


Now, we switch our attention to kernelization algorithms for \textsc{VDP}.
First, 
we give a $4k$ vertex kernel for \textsc{VDP} on split graphs. This resolves an open question by Heggernes et al.~\cite{hegger}, who asked whether this problem admits a linear vertex kernel.
For this purpose, we use the result by Yang et al.~\cite{yang}, who gave a $4k$ vertex kernel for a restricted variant of \textsc{VDP}, 
called \textsc{VDP-Unique} by us, 
where each vertex can participate in at most one terminal~pair. 
			
In order to obtain a linear vertex kernel for \textsc{VDP}, we give a parameter-preserving reduction to \textsc{VDP-Unique}.
Our reduction relies on a non-trivial matching-based argument.

In this way, we improve upon the $4k^{2}$ vertex kernel given by Heggernes et al.~\cite{hegger} as well as generalize the result given by Yang et al.~\cite{yang}. Specifically, we have the following theorem.

 \begin{restatable}{theorem}{splitvdp} \label{labelsplitvdp}
\textsc{VDP} on split graphs admits a kernel with at most $4k$ vertices. 
\end{restatable}

Next, we give an $\mathcal{O}(k^{2})$ vertex kernel for \textsc{VDP} on well-partitioned chordal graphs (see Definition \ref{defwpc}).  Ahn et al.~\cite{ahn1} showed that \textsc{VDP} admits an $\mathcal{O}(k^{3})$ vertex kernel on this class. 
 We improve their bound by giving a marking procedure
 that marks a set of at most $\mathcal{O}(k^{2})$ vertices in $G$,
 which ``covers'' some solution (if it exists). As a result, we arrive at the following theorem.
 
 \begin{restatable}{theorem}{wpcvdp} \label{labelwpcvdp}
\textsc{VDP} on well-partitioned chordal graphs admits a kernel with $\mathcal{O}(k^{2})$ vertices. 
\end{restatable}

Unlike \textsc{EDP}, the \textsc{VDP} problem turns out easier on the remaining graph classes. In block graphs,
for every terminal pair with terminals in different blocks, there is a unique induced path connecting these terminals (all internal vertices of this path are the cut vertices).
After adding these paths to the solution, we end up with a union of clique instances, where \textsc{VDP} is solvable in polynomial time.
This leads to the following observation about block graphs and its subclass, clique paths.

\begin{restatable}{observation}{blockvdp}\label{blockP}
\textsc{VDP} on block graphs (and, in particular, on clique paths) is solvable in polynomial time.
\end{restatable}

Finally, we identify a less restricted graph class on which \textsc{VDP} is polynomial-time solvable, namely the class of threshold graphs.
This yields a sharp separation between split graphs and its subclass---threshold graphs---in terms of \textsc{VDP}.

 \begin{restatable}{theorem}{thresholdPoly}\label{thm:threshold-vdp-poly}
 \textsc{VDP} on thresholds graphs is solvable in polynomial time.
\end{restatable}

\subsection{Brief Survey of Related Works}\label{sec:LS}

Both \VDP~and~\EDP~are known to be $\mathsf{NP}$-complete for planar graphs~\cite{kramer}, line graphs~\cite{hegger}, and split graphs~\cite{hegger}. Moreover, \VDP~is known to be $\mathsf{NP}$-complete on interval graphs~\cite{VDPIntervalHard} and grids~\cite{kramer} as well. Both \VDP~and \EDP~were studied from the viewpoint of structural parameterization. While \VDP~is \FPT~parameterized by treewidth~\cite{VDPTreewidth}, \EDP~remains $\mathsf{NP}$-complete even on graphs with treewidth at most 2~\cite{EDPHardTreewidth}. Gurski and Wange~\cite{DPCliqueWidth} showed that \VDP~is solvable in linear time on co-graphs but becomes $\mathsf{NP}$-complete for graphs with clique-width at most 6. As noted by Heggernes et al.~\cite{hegger}, there is a reduction from \VDP~on general graphs to \EDP~on line graphs, and from \EDP~on general graphs to \VDP~on line graphs. Using this reduction and the fact that a graph with treewidth $\ell$ has clique-width at most $2\ell +2$~\cite{DPCliqueWidth}, \EDP~can also be shown to be $\mathsf{NP}$-complete on graphs with clique-width at most 6~\cite{hegger}. 

The Graph Minors theory of Robertson and Seymour provides some of the most important algorithmic results of the modern graph theory. Unfortunately, these algorithms, along with the $\mathcal{O}(n^3)$ and $\mathcal{O}(n^2)$ algorithms for \VDP~and \EDP (when $k$ is fixed),~respectively~\cite{DPMinor,fastestVDP}, hide such big constants that they have earned a name for themselves: ``\textit{galactic algorithms}''. Since then, finding efficient \FPT~algorithms for \VDP~and \EDP~ has been a tantalizing question for researchers. Several improvements have been made for the class of planar graphs~\cite{AdlerDP,VDPSTOCPlanar,reed1995rooted,robertson2012graph}, chordal graphs~\cite{DPChordal}, and bounded-genus graphs~\cite{reed1995rooted,dvovrak2009coloring,kobayashi2009algorithms}. 

Concerning kernelization (in addition to the works surveyed earlier in the introduction), we note that Ganian and Ordyniak~\cite{cutEDP} proved that \EDP~admits a linear kernel parameterized by the feedback edge set number. Recently, Golovach et al.~\cite{SetRestricted} proved that \textsc{Set-Restricted Disjoint Paths}, a variant of \VDP~where each terminal pair has to find its path from a predefined set of vertices, does not admit a polynomial compression on interval graphs unless \NPoly.

The optimization variants of \VDP~and~\EDP---\textsc{MaxVDP} and \textsc{MaxEDP}--- are well-studied in the realm of approximation algorithms~\cite{DPApproxHard,DPApproxTW,DPApprox1,DPApprox2,DPApprox3}. Chekuri et al.~\cite{DPApproxHard} gave an $\mathcal{O}(\sqrt{n})$-approximation algorithm for \textsc{MaxEDP} on general graphs, matching the $\Omega(\sqrt{n})$ lower bound provided by Garg et al.~\cite{DPApprox2}.  Recently, \textsc{MaxVDP} has gained much attention on planar graphs~\cite{planarApprox1,planarApprox2,planarApprox4,planarApprox3,planarApprox5}. Highlights of this line of works (\textsc{MaxVDP} on planar graphs) include an approximation algorithm with approximation ratio $\mathcal{O}(n^{\frac{9}{19}}\log^{\mathcal{O}(1)}n)$~\cite{planarApprox2}, and hardness of approximating within a factor of $n^{\Omega( {1}/{(\log \log n)^2} )}$~\cite{planarApprox4}. 

\subsection{Organization of the paper}

We begin with formal preliminaries, where we gather information about the studied graph classes and the basic algorithmic tools. 
In \cref{sec:edp}, we prove the kernelization theorems for \textsc{EDP}, which are followed by the \textsf{NP}-hardness proof for \textsc{EDP} on cliques in \cref{npc:complete}.
Next, we cover the kernelization results for \textsc{VDP} in \cref{sec:vdp} and present a polynomial-time algorithm for threshold graphs in \cref{sec:threshold-poly}.
We conclude in \cref{sec:conclusion}.

\section{Preliminaries}
For a positive integer $\ell$, let $[\ell]$ denote the set $\{1, \ldots, \ell \}$.
\subsection{Graph Notations} \label{sec:prelim}
All graphs considered in this paper are simple, undirected, and connected unless stated otherwise. 	Standard graph-theoretic terms not explicitly defined here can be found in~\cite{diestel}. For a graph $G$, let $V(G)$ denote its vertex set, and $E(G)$ denote its edge set. For a graph $G$, the subgraph of $G$ induced by $S\subseteq V(G)$ is denoted by $G[S]$, where $G[S]=(S,E_{S})$ and $E_{S}=\{xy\in E(G)\mid x,y \in S\}.$  For two sets $X,Y \subseteq V(G)$, we denote by $G[X,Y]$ the subgraph of $G$ with vertex set $X\cup Y$ and edge set $\{xy\in E(G) : x \in X, y \in Y \}$. The \emph{open neighborhood} of a vertex $v$ in $G$ is $N_{G}(v)=\{u\in V(G): uv\in E(G)\}$. The \emph{degree} of a vertex $v$ is $|N_{G}(v)|$, and it is denoted by $d_{G}(v)$. When there is no ambiguity, we do not use the subscript $G$ in $N_{G}(v)$ and $d_{G}(v)$. A vertex $v$ with $d(v)=1$ is a \emph{pendant vertex}. The \emph{distance} between two vertices in a graph $G$ is the number of edges in the shortest path between them. We use the notation $d(u,v)$ to represent the distance between two vertices $u$ and $v$ in a graph $G$ (when $G$ is clear from the context). For a graph $G$ and a set $X\subseteq V(G)$, we use $G-X$ to denote $G[V(G)\setminus X]$, that is, the graph obtained from $G$ by deleting $X$. In a graph $G$, two vertices $u$ and $v$ are \emph{twins} if $N_{G}[u] = N_{G}[v]$.

An \emph{independent set} of a graph $G$ is a subset of $V(G)$ such that no two vertices in the subset have an edge between them in $G$. A \emph{clique} is a subset of $V(G)$ such that every two distinct vertices in the subset are adjacent in $G$. Given a graph $G$, a \emph{matching} $M$ is a subset of edges of $G$ that do not share an endpoint. The edges in $M$ are called \emph{matched edges}, and the remaining edges are called \emph{unmatched edges}. Given a matching $M$, a vertex
$v\in V(G)$ is \textit{saturated} by $M$ if $v$ is incident on an edge of $M$, that is, $v$ is an end
vertex of some edge of $M$. Given a graph $G$, \textsc{Max Matching} is to find a matching of maximum cardinality in $G$. 
\begin{proposition} [\cite{hopcroft}] \label{propmatching}
For a bipartite graph $G$, \textsc{Max Matching} can be solved in $\mathcal{O}(\sqrt{|V(G)|}\cdot|E(G)|)$ time. 
\end{proposition}

A path $P = (v_1,\ldots,v_n)$ is an \emph{$M$-alternating path} if the edges in $P$ are matched and unmatched alternatively with respect to $M$. If both the end vertices of an alternating path are unsaturated, then it is an \emph{$M$-augmenting path}. 

\begin{proposition} [\cite{berge}] \label{maxprop}
A matching $M$ is maximum if and only if there
is no $M$-augmenting path in $G$.
\end{proposition}

A path $P = (v_1, v_2, \ldots, v_n)$ on $n$ vertices is a \emph{$(v_{1},v_{n})$-path} and $\{v_2, \ldots, v_{n-1}\}$ are the \emph{internal vertices} of $P$. Moreover, for a path $P = (v_1, v_2, \ldots, v_n)$, we say that $P$ \emph{visits} the vertices $\{v_1, v_2, \ldots, v_n\}$. Throughout this paper, let $P_{uv}$ denote the path containing only the edge $uv$. Let $P_1$ be an $(s_1, t_1)$-path and $P_2$ be an $(s_2, t_2)$-path. Then, $P_1$ and $P_2$ are \emph{vertex-disjoint} if $V(P_1) \cap V(P_2)= \emptyset$. Moreover, $P_1$ and $P_2$ are \emph{internally vertex-disjoint} if $(V (P_1) \setminus \{s_1, t_1\}) \cap V(P_{2}) = \emptyset$ and $(V (P_2) \setminus \{s_2, t_2\}) \cap V(P_{1}) = \emptyset$, that is, no internal vertex of one path is used as a vertex on
the other path, and vice versa. Two paths are said to be \emph{edge-disjoint} if they do not have any edge in common. Note that two internally vertex-disjoint paths are edge-disjoint, but the converse may not be true. A path $P$ is \emph{induced} if $G[V(P)]$ is the same as $P$. For a path $P = (v_1,\ldots,v_n)$ on $n$ vertices and vertices $v_i,v_j \in V(P)$, let $(v_i,v_j)$-\textit{subpath of $P$} denote the  subpath of $P$ with endpoints $v_i$ and $v_j$. 

\subsection{Problem Statements}\label{sec:PS}
  Given a graph $G$ and a set (or, more generally, an ordered multiset) $\mathcal{X}$ of pairs of distinct vertices in $G$, we refer to the pairs in $\mathcal{X}$ as \emph{terminal pairs}. A vertex in $G$ is a \emph{terminal vertex} if it appears in at least one terminal pair in $\mathcal{X}$ (when $\mathcal{X}$ is clear from context); else, it is a \emph{non-terminal vertex}. For example, if $G$ is a graph with $V(G)=\{v_{1},v_{2},\ldots, v_{6}\}$, and  $\mathcal{X}=\{(v_1,v_{3}),(v_{2},v_{3}),(v_{3},v_{6}),(v_{1},v_{6})\}$ is a set of terminal pairs, then $\{v_{1},v_{2},v_{3},v_{6}\}$ are terminal vertices in $G$ and $\{v_{4},v_{5}\}$ are non-terminal vertices in $G$.

The \textsc{Vertex-Disjoint Paths}  problem takes as input a graph
$G$ and a set of $k$ terminal pairs in $G$, and the task is to decide whether
there exists a collection of $k$ pairwise internally vertex-disjoint paths in $G$ such that the vertices
in each terminal pair are connected to each other by one of the paths. More formally,
\bigskip

		\noindent\fbox{ \parbox{137mm}{
		\noindent \underline{\textsc{Vertex-Disjoint Paths (VDP)}:}\\
		\textbf{Input:} A graph $G$ and an ordered multiset $\mathcal{X} = \{(s_1, t_1), \ldots,(s_k, t_k)\}$ of $k$ pairs of terminals.\\
		\textbf{Question:} Does $G$ contain $k$ distinct and pairwise internally vertex-disjoint paths $P_1, \ldots, P_k$ such
that for all $i\in[k]$, $P_{i}$
is an $(s_i, t_i)$-path?}}
			
			\medskip
			
 The \textsc{Edge-Disjoint Paths}  problem takes as input a graph
$G$ and a set of $k$ terminal pairs in $G$, and the task is to decide whether
there exists a collection of $k$ pairwise edge-disjoint paths in $G$ such that the vertices
in each terminal pair are connected to each other by one of the paths. More formally,

			\bigskip

		\noindent\fbox{ \parbox{137mm}{
		\noindent \underline{\textsc{Edge-Disjoint Paths (EDP)}:}\\
		\textbf{Input:} A graph $G$ and an ordered multiset $\mathcal{X} = \{(s_1, t_1), \ldots,(s_k, t_k)\}$ of $k$ pairs of terminals.\\
		\textbf{Question:} Does $G$ contain $k$ pairwise edge-disjoint paths $P_1, \ldots, P_k$ such
that for all $i\in[k]$, $P_{i}$
is an $(s_i, t_i)$-path?}}
\medskip

\begin{remark} \label{remark1}
Note that in
both problems (\VDP~and \EDP), we allow different terminal pairs to intersect, that is, it may happen that for $i\neq j, \{s_i,t_i \} \cap \{ s_j,t_j \}  \neq \emptyset$.
\end{remark}

If there are two identical pairs $\{ s_i,t_i\} = \{s_j,t_j \} = \{ x,y\}$ in $\mathcal{X}$
and the  edge $xy$ is present in $G$ then only one of the paths $P_i, P_j$ can use the edge $xy$ if we require them to be edge-disjoint.
However, setting $P_i = P_j = (x,y)$ does not violate the condition of being internally vertex-disjoint.
It is natural though and also consistent with the existing literature to
impose the additional condition that all paths in a solution have to be pairwise distinct. 

\begin{remark} \label{remark2}
Throughout this paper, we assume that the degree of every terminal vertex is at
least the number of terminal pairs in which it appears. Else, it is trivially a No-instance.
\end{remark}

Following the notation introduced in~\cite{ahn1} and~\cite{hegger}, we have the following definition.
\begin{definition} \label{def:heavy}
An edge $xy \in E(G)$ is \emph{heavy} if for some $w\geq 2$,
there exist pairwise distinct indices $i_1, \ldots, i_w$ such that for each $j\in [w]$, $\{x, y\} = \{s_{i_j} , t_{i_j} \}$. We call a terminal pair $(s_{i},t_{i})$ 
\emph{heavy} if $s_it_i$ is a heavy edge; else, we call it \emph{light}. Note that calling a terminal pair heavy or light only makes sense when the terminals in the pair have an edge between them. 
\end{definition}

Next, consider the following definition.

\begin{definition} [Minimum Solution]
Let $(G,\mathcal{X},k)$ be a Yes-instance of \textsc{VDP} or \textsc{EDP}.
A solution $\mathcal{P} = \{P_1, \ldots, P_k\}$ for the instance $(G,\mathcal{X},k)$ is \emph{minimum} if there is no
solution $\mathcal{Q} = \{Q_1, \ldots, Q_k\}$ for $(G,\mathcal{X},k)$ such that $\sum_{i=1}^{k}|E(Q_{i})|<\sum_{i=1}^{k}|E(P_{i})|$.
\end{definition}

Since we deal with subclasses of chordal graphs (see Section \ref{sec:GC}), the two following propositions are crucial for us. 

\begin{proposition} [\cite{hegger}] \label{prop1}
Let $(G,\mathcal{X},k)$ be a Yes-instance of \textsc{VDP} such that $G$ is a chordal graph, and let $\mathcal{P}=\{P_1, \ldots, P_k\}$ be a minimum
solution of $(G,\mathcal{X},k)$. Then, every path $P_{i}\in \mathcal{P}$ satisfies exactly one of the following two statements:
\begin{enumerate}
    \item [$(i)$] $P_i$ is an induced path.
    \item [$(ii)$] $P_i$ is a path of length $2$, and there exists a path $P_{j}\in \mathcal{P}$ of length $1$ whose endpoints are the same as the endpoints of $P_i$.
 
\end{enumerate}
\end{proposition}

The next observation follows from Proposition \ref{prop1}.

\begin{observation} \label{obsminimum}
Let $(G,\mathcal{X},k)$ be a Yes-instance of \textsc{VDP}. If there is a terminal pair $(s,t)\in \mathcal{X}$ such that $st\in E(G)$, then $P_{st}$ belongs to every minimum solution of $(G,\mathcal{X},k)$.
\end{observation}



\subsection{Parameterized Complexity}

 Standard notions in Parameterized Complexity not explicitly defined here can be found in~\cite{bookParameterized}. Let $\Pi$ be an $\mathsf{NP}$-hard problem. In the framework of Parameterized Complexity, each instance of $\Pi$ is associated with a \textit{parameter} $k$. We say that $\Pi$ is \textit{fixed-parameter tractable} ($\mathsf{FPT}$) if any instance $(I,k)$ of $\Pi$ is solvable in time $f(k)\cdot |I|^{\mathcal{O}(1)}$, where
$f$ is some computable function of $k$. 

\begin{definition} [Equivalent Instances] \label{equivalent}
Let $\Pi$ and $\Pi'$ be two parameterized problems. Two instances $(I, k)\in \Pi$ and $(I', k')\in \Pi'$ are \emph{equivalent} if: $(I, k)$ is a Yes-instance of $\Pi$ if and only if $(I', k')$ is a Yes-instance of $\Pi'$.
\end{definition}

A parameterized (decision) problem $\Pi$ admits a \emph{kernel} of size $f(k)$ for some computable function $f$ that depends only on $k$ if the following is true: There exists an algorithm (called a \emph{kernelization algorithm}) that runs in $(|I|+k)^{\mathcal{O}(1)}$ time and translates any input instance $(I,k)$ of $\Pi$ into an equivalent instance $(I',k')$ of $\Pi$ such that the size of $(I',k')$ is bounded by $f(k)$. If the function $f$ is polynomial (resp., linear) in $k$, then the problem is said to admit a \emph{polynomial kernel} (resp., \emph{linear kernel}). It is well known that a decidable parameterized problem is \FPT~if and only if it admits a kernel~\cite{bookParameterized}.

To design kernelization algorithms, we rely on the notion of \emph{reduction rule}, defined below.

\begin{definition} [Reduction Rule]
A \emph{reduction rule} is a polynomial-time procedure that consists of a condition and an operation, and its input is an instance $(I,k)$ of a parameterized problem $\Pi$. If the condition is true, then the rule outputs a new instance $(I',k')$ of $\Pi$ such that $k'\leq k$. Usually, also $|I'|<|I|$.
\end{definition}

A reduction rule is \emph{safe}, when the condition is true, $(I,k)$ and $(I',k')$ are equivalent. Throughout this paper, the reduction rules will be numbered, and the reduction rules will be applied exhaustively in the increasing order of their indices. So, if reduction rules $i$ and $j$, where $i<j,$ are defined for a problem, then $i$ will be applied exhaustively before $j$. Notice that after the application of rule $j$, the condition of rule $i$ might become true. In this situation, we will apply rule $i$ again (exhaustively). In other words, when we apply rule $j$, we always assume that the condition of rule $i$ is false. 

\subsection{Graph Classes}\label{sec:GC}
A graph $G$ is a \emph{chordal graph} if every cycle in $G$ of length at least four
has a \emph{chord}, that is, an edge joining two non-consecutive vertices of the cycle. In what follows, we define several subclasses of the class of chordal graphs, namely, \emph{complete graphs}, \emph{block graphs}, \emph{split graphs}, \emph{threshold graphs}, and \emph{well-partitioned chordal graphs}. A graph whose vertex set is a clique is a \emph{complete graph}.  A vertex is a \emph{cut vertex} in a graph $G$ if removing it increases the total number of connected components in $G$.
A \emph{block} of a graph $G$ is a maximal connected subgraph of $G$ that does not contain any
cut vertex. A graph $G$ is a \emph{block graph} if the vertex set of every block in $G$ is a clique. A block of a block graph $G$ is an \emph{end block} if it contains exactly one cut vertex. Note that a block graph that is not a complete graph has at least two end blocks. A graph $G$ is a \emph{split graph} if there is a partition $(C, I)$ of $V(G)$ such that $C$ is a clique and
$I$ is an independent set. A split graph $G$ is a \textit{threshold graph} if there exists a linear ordering of the vertices in $I$, say, $(v_{1},v_{2},\ldots, v_{|I|})$, such that $N(v_{1})\subseteq N(v_{2}) \subseteq \cdots \subseteq N(v_{|I|})$.

An undirected graph in which any two vertices are connected by exactly one path is a \emph{tree}. A tree with at most two vertices or exactly one non-pendant vertex is a \emph{star}. The vertices of degree one in a tree are called \emph{leaves}. Note that a split graph admits a partition of its vertex
set into cliques that can be arranged in a star structure, where the leaves are cliques of size one. Motivated by this definition of split graphs, Ahn et al.~\cite{ahn1} introduced \emph{well-partitioned chordal graphs}, which are defined by relaxing the definition of split graphs in the following two ways: (i) by allowing
the parts of the partition to be arranged in a tree structure instead of a star structure, and (ii) by allowing the
cliques in each part to have arbitrary size instead of one. A more formal definition of a well-partitioned chordal graph is given below.

\begin{definition} [Well-Partitioned Chordal Graph] \label{defwpc}
A connected graph $G$ is a \emph{well-partitioned chordal graph} if there exists a partition $\mathcal{B}$ of $V(G)$ and a
tree $\mathcal{T}$ having $\mathcal{B}$ as its vertex set such that the following hold.

\begin{enumerate}
    \item [$(i)$] Each part $X\in \mathcal{B}$ is a clique in $G$.
    
    \item [$(ii)$] For each edge $XY\in E(\mathcal{T})$, there are subsets $X'\subseteq X$ and $Y'\subseteq Y$ such that $E(G[X,Y])=X'\times Y'$.
    
\item [$(iii)$] For each pair of distinct $X,Y\in \mathcal{B}$ with $XY\notin E(\mathcal{T})$, $E(G[X,Y])=\emptyset$.
\end{enumerate}
\end{definition}

The tree $\mathcal{T}$ in Definition \ref{defwpc} is called a \emph{partition tree} of $G$, and the elements of $\mathcal{B}$ are called its \emph{bags}. Notice that a well-partitioned chordal graph can have multiple partition trees.

\begin{proposition} [\cite{ahn1}] \label{ahn:prop}
Given a well-partitioned chordal graph $G$, a partition tree of $G$ can be found in polynomial time.
\end{proposition} 

  \begin{figure}[t]
 \centering
    \includegraphics[scale=0.55]{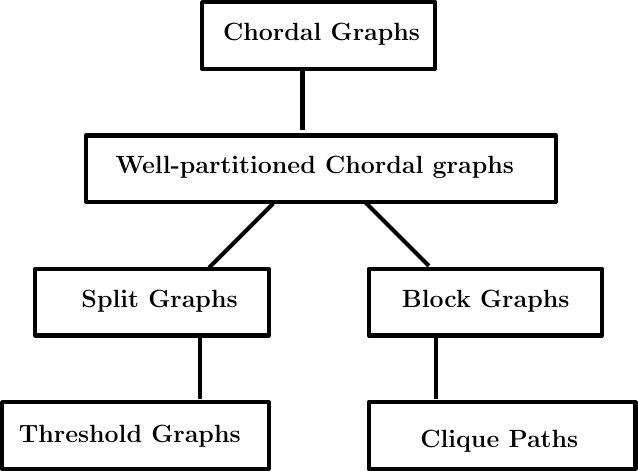}
    \caption{The inclusive relationship among the subclasses of chordal graphs studied in this paper.}
    \label{fig0}
\end{figure}

\begin{definition} [Boundary of a Well-Partitioned Chordal Graph]
Let $\mathcal{T}$ be a partition tree of a well-partitioned chordal graph $G$ and let $XY \in E (\mathcal{T})$. The \emph{boundary} of $X$ with respect to $Y$, denoted
by $\mathsf{bd}(X,Y)$, is the set of vertices of $X$ that have a neighbor in $Y$, i.e., $\mathsf{bd}(X,Y)= \{x \in X : N_{G}(x) \cap Y \neq \emptyset \}$.
\end{definition}

\begin{remark} \label{obsdel}
Note that all the graph classes considered in this paper are closed under vertex deletion. Therefore, throughout this paper, if $G$ belongs to class $\mathcal{Z}$ and $G'$ is obtained from $G$ after deleting some vertices, then we assume that $G'$ also belongs to $\mathcal{Z}$ without mentioning it explicitly. 
\end{remark}

The inclusion relationship among various subclasses of chordal graphs discussed in this paper is shown in Figure \ref{fig0}.

\section{Kernelization Results on \textsc{EDP}}\label{sec:edp}

We begin with the analysis of the simplest scenario where the input graph is a clique.
In this setting, \textsc{EDP} is still \textsf{NP}-hard (see \cref{npc:complete}), but we show below that whenever the size of the clique is larger than the parameter $k$, then we always obtain a Yes-instance.
This improves the bound in~\cite[Lemma 7]{hegger} by a factor of 2, which will play a role in optimizing the constants in our kernels (particularly, the linear ones).

\begin{lemma}\label{L:cliqueBetter}
    Let $(G,\mathcal{X}, k)$ be an instance of \textsc{EDP} such that $G$ is
a clique. If $|V(G)| > k$, then $(G,\mathcal{X},k)$ is a Yes-instance.
\end{lemma}
\begin{proof}
    We give a proof by induction on $k$.
    For $k \in \{0, 1\}$ the lemma clearly holds.
    Consider $k>1$ and define $H$ as a graph on the vertex set $V(G)$ with edges given by pairs of vertices $x,y$ for which $\{x,y\}$ appears at least twice in $\mathcal{X}$.
    Then $H$ has at most $\frac{k}{2}$ edges and thus at most $k$ non-isolated vertices.
    Since $|V(H)| = |V(G)| > k$, there exists an isolated vertex $v$ in $H$.
    Let $\mathcal{X}_v \subseteq \mathcal{X}$ denote the set of pairs containing $v$ (it is not a multiset by the choice of $v$) and $k_v = |\mathcal{X}_v|$.
    We distinguish two cases.

    First, suppose that $k_v > 0$. 
    Then $|\mathcal{X} \setminus \mathcal{X}_v| = k - k_v \le k - 1 < |V(G-v)|$.
    Hence, the \textsc{EDP} instance $(G-v,\mathcal{X} \setminus \mathcal{X}_v, k - k_v)$ satisfies the conditions of the lemma and hence,  by the inductive assumption, it is a Yes-instance.
    Let $\mathcal{P}$ denote a solution to this instance and
    $\mathcal{P}_v$ denotes the set of single-edge paths corresponding to the pairs in $\mathcal{X}_v$.
    Clearly, all edges used in $\mathcal{P}_v$ are incident to $v$, and so these paths are edge-disjoint with $\mathcal{P}$.
    Consequently, $\mathcal{P} \cup \mathcal{P}_v$ is a solution to the instance $(G,\mathcal{X}, k)$.

    Suppose now that $k_v = 0$, that is, $v$ does not appear in $\mathcal{X}$.
    Let $e = \{u,w\}$ be an arbitrary pair from $\mathcal{X}$.
    Again, by the inductive assumption, the instance $(G-v, \mathcal{X} \setminus \{e\}, k-1)$ admits some solution $\mathcal{P}$.
    Then the path $(u,v,w)$ is edge-disjoint with $\mathcal{P}$, and so    
    $\mathcal{P} \cup \{(u,v,w)\}$ forms a~solution to the instance $(G,\mathcal{X}, k)$.

    In both cases, we were able to construct a solution to $(G,\mathcal{X}, k)$, which concludes the inductive argument.
\end{proof}

The bound above is tight as one can construct a No-instance $(G,\mathcal{X}, k)$ where $G$ is a clique and $|V(G)| = k$.
Consider $\mathcal{X}$ comprising just $k$ copies of some pair $\{u,v\}$.
Since the degree of $u$ is $k-1$, there cannot be $k$ edge-disjoint paths having $u$ as their common endpoint.

If $G$ is a split graph with more than $k$ vertices in the clique and the degree of each terminal
vertex is at least the number of terminals on it, then we can reduce such an instance to the setting of \cref{L:cliqueBetter} by replacing each terminal $v$ in the independent set with an arbitrary neighbor of $v$.
As a consequence, we obtain the following corollary, being a quantitative improvement over~\cite[Lemma 8]{hegger}.

\begin{corollary}\label{prop2}
Let $(G,\mathcal{X}, k)$ be an instance of \textsc{EDP} such that $G$ is
a split graph with split partition $(C, I)$. If $|C| > k$ and the degree of each terminal
vertex is at least the number of terminals on it, then $(G,\mathcal{X},k)$ is a Yes-instance.
\end{corollary}

\subsection{A Subcubic Vertex Kernel for Split Graphs} \label{subcubic:edp:split}

In this section, we show that $\textsc{EDP}$ on split graphs admits a kernel with $\mathcal{O}(k^{2.75})$ vertices. Let $(G,\mathcal{X},k)$ be an instance of \textsc{EDP} where $G$ is a split graph. Note that given a split graph $G$, we can compute (in linear time) a partition $(C, I)$ of $V(G)$ such that $C$ is a clique and $I$ is an independent set~\cite{hammer}. We partition the set $I$ into two sets, say, $I_T$ and $I_N$, where $I_T$ and $I_N$ denote the set of terminal vertices and the set of non-terminal vertices in $I$, respectively. 

To ease the presentation of mathematical calculations, for this section (Section~\ref{subcubic:edp:split}), we assume that $k^{\frac{1}{4}}$ is a natural number. If this is not the case, then we can easily get a new equivalent instance that satisfies this condition in the following manner. Let $d= (\lceil k^{\frac{1}{4}}\rceil)^4-k$ and $v\in C$. Now, we add $d$ terminal pairs $\{(s_{i_1},t_{i_1}), \ldots, (s_{i_d},t_{i_d})\}$ and attach each of these terminals to $v$. Observe that this does not affect the size of our kernel ($\mathcal{O}(k^{2.75})$ vertices) since $(\lceil k^{\frac{1}{4}}\rceil)^4 = \mathcal{O}(k)$. Moreover, we assume that $k>8$, as otherwise, we can use the \FPT algorithm for \EDP~\cite{DPMinor} to solve it in polynomial time. 

Before proceeding further, let us first discuss the overall idea
leading us to Theorem~\ref{labelsplitedp}.

\medskip
\noindent\textbf{Overview.}
Heggernes et al.~\cite{hegger} gave an $\mathcal{O}(k^3)$ vertex kernel for \EDP~on split graphs. In our kernelization algorithm (in this section), we use their algorithm as a preprocessing step. After the prepossessing step, the size of $C$ and $I_T$ gets bounded by $2k$ each, and the size of $I_N$ gets bounded by $\mathcal{O}(k^3)$. Therefore, we know that the real challenge in designing an improved kernel for \textsc{EDP} on split graphs lies in giving a better upper bound on $|I_{N}|$.

Our kernelization algorithm makes a non-trivial use of a lemma (Lemma \ref{L:SPathLength}), which establishes that the length of each path in any minimum solution (of \EDP~on a split graph $G$) is bounded by $\mathcal{O}(\sqrt{k})$. This, in turn, implies that a minimum solution of \EDP~for split graphs contains $\mathcal{O}(k^{1.5})$ edges. Note that during the preprocessing step (i.e., the kernelization algorithm by Heggernes et al.~\cite{hegger}), for every pair of vertices in $C$, at most $4k+1$ vertices are reserved in $I_{N}$, giving a cubic vertex kernel. In our algorithm, we characterized those vertices (called \emph{rich} by us) in $C$ for which we need to reserve only $\mathcal{O}(k^{1.5})$ vertices in $I_{N}$.
Informally speaking, a vertex $v\in C$ is \emph{rich} if there are $\Omega(k^{0.75})$ vertices in $C$ that are ``reachable'' from $v$, even if we delete all the edges used by a ``small'' solution (containing  $\mathcal{O}(k^{1.5})$ edges). Then, we show that if two vertices are rich, then even if they do not have any common neighbors in $I_N$, there exist ``many'' ($\Omega(k^{1.5})$)  edge-disjoint paths between them even after removing any $\mathcal{O}(k^{1.5})$ edges of $G$. Hence, for every rich vertex, we keep only those vertices in $I_N$ that are necessary to make the vertex rich, that is, we keep $\mathcal{O}(k^{1.5})$ vertices in $I_N$ for every rich vertex. Thus, all rich vertices in $C$ contribute a total of $\mathcal{O}(k^{2.5})$ vertices in $I_N$. The vertices in $C$ that are not rich are termed as \emph{poor}. Finally, we establish that a poor vertex cannot have too many neighbors in $I_N$. More specifically, a poor vertex can have only $\mathcal{O}(k^{1.75})$ neighbors in $I_N$. So, even if we keep all their neighbors in $I_{N}$, we store a total of $\mathcal{O}(k^{2.75})$ vertices in $I_N$ for the poor vertices. This leads us to an $\mathcal{O}(k^{2.75})$ vertex kernel for \EDP~on split graphs. 
\subsubsection{A Bound on the Length of the Paths in a Minimum Solution} \label{SS:boundSplit}
In this section, we prove that for a minimum solution $\mathcal{P}$ of an instance $(G,\mathcal{X},k)$ of \EDP~where $G$ is a split graph, each path $P \in \mathcal{P}$ has length at most $4\sqrt{k}+3$. We prove this bound by establishing that if there is a path of length  $4\sqrt{k}+4$ in $\mathcal{P}$, then $\mathcal{P}$ contains at least $k+1$ paths, a contradiction. To this end, we need the concept of \emph{intersecting edges} (see Definition \ref{def2}) and \emph{non-compatible edges} (see Definition \ref{def1}).

Now, consider the following remark.

\begin{remark}
For ease of exposition, throughout this section (Section~\ref{SS:boundSplit}), we assume (without mentioning explicitly it every time) that $(G,\mathcal{X},k)$ is a Yes-instance of \EDP,~where $G$ is a split graph. Moreover, $\mathcal{P}$ denotes a (arbitrary but fixed) minimum solution of $(G,\mathcal{X},k)$, and $P\in \mathcal{P}$ is a path such that $P$ contains $\ell$ vertices, say, $v_1,\ldots, v_\ell$, from clique $C$. Moreover, without loss of generality, let $v_1, \ldots, v_\ell$ be the order in which these vertices appear in the path $P$ from some terminal to the other. See Figure~\ref{Fig:splitDemo} for an illustration. Note that if a path, say, $P'$, in a split graph has length $p$ (i.e. $|E(P')| = p$), then it contains at least $\lceil \frac{p}{2} \rceil$ vertices from $C$. Therefore, to bound the length of $P$ by $\mathcal{O} (\sqrt{k})$, it suffices to bound the number of vertices of $C$ in $P$ by $\mathcal{O}(\sqrt{k})$. 
\end{remark}

\begin{figure}
    \centering
    \includegraphics[scale=0.8]{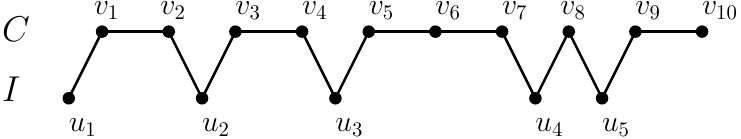}
    \caption{Here, $P$ is a path with endpoints $u_1$ and $v_{10}$. The vertices $v_1, \ldots, v_{10}$ are clique vertices of the path $P$ (here, $\ell =10$) and $u_1, \ldots, u_5$ are independent set vertices of $P$. Observe that $v_iv_{i+1}$ is not necessarily an edge in $P$.}
    \label{Fig:splitDemo}
\end{figure}

 Assuming the ordering $v_1,\ldots, v_\ell$ of the vertices in $V(P) \cap C$ along the path $P$, we have the following definitions.

\begin{definition}[Intersecting Edges] \label{def2}
Consider two edges $e_i = v_iv_{i'}$ and $e_j = v_jv_{j'}$ such that $i,i',j,j' \in [\ell]$, and without loss of generality, assume that $i<i',j<j'$, and $i \leq j$. Then, $e_i$ and $e_j$ are {\em non-intersecting} if $j\geq i'$; otherwise, they are {\em intersecting}. See Figure~\ref{Fig:SplitIntersect} for an illustration.  
\end{definition}

\begin{definition}[Non-compatible Edges] \label{def1}
 Two edges $e_1,e_2 \in E(G)$ are \emph{non-compatible} if there does not exist a path $P'\in \mathcal{P} \setminus \{P\}$ (given $P \in \mathcal{P}$) such that $\{e_1,e_2\} \subseteq E(P')$. Moreover, a set of edges $\mathcal{S} = \{ e_1, \ldots, e_p \}\subseteq E(G)$ is \emph{non-compatible} if every distinct $e_i,e_j \in \mathcal{S}$ are non-compatible.
\end{definition}

\begin{figure}
    \centering
    \includegraphics[scale=0.8]{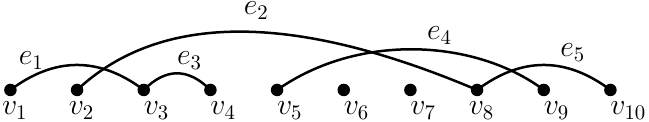}
    \caption{Here, edge $e_1$ intersects with $e_2$, edge $e_2$ intersects with edges $e_1,e_3,e_4$, edge $e_3$ intersects with $e_2$, edge $e_4$ intersects with $e_2$ and $e_5$, and edge $e_5$ intersects with $e_4$. Note that although edges $e_2$ and $e_5$ share an endpoint; they are non-intersecting. }
    \label{Fig:SplitIntersect}
\end{figure}

Next, we show that, since $P$ is a path in a minimum solution $\mathcal{P}$, most of the edges with both endpoints in $\{v_1,\ldots, v_\ell\}$ are used by paths in $\mathcal{P}$ (otherwise, we get a contradiction to the fact that $\mathcal{P}$ is a minimum solution). In particular, we have the following lemma.

\begin{lemma}\label{C:splitShortcut}
Each edge of the form $v_iv_j$, where $i, j\in [\ell]$ and $j\geq i+2$, is used by some path in $ \mathcal{P} \setminus \{P\}$.
\end{lemma}
\begin{proof}
Targeting a contradiction, consider $v_iv_j$, where $i, j\in [\ell]$ and $j\geq i+2$, that is not used by any path in $\mathcal{P}$. Then, observe that the path we get, say, $P'$, after replacing the $(v_i,v_j)$-subpath of $P$, which contains at least two edges, with the edge $v_iv_j$, has fewer edges than $P$. Moreover, $(\mathcal{P} \setminus \{P \}) \cup \{P' \}$ is also a solution of  $(G,\mathcal{X},k)$. This contradicts the fact that $\mathcal{P}$ is a minimum solution.  
\end{proof}

Now, we show that if two edges are intersecting edges, then they are non-compatible as well. In particular, we have the following lemma.
\begin{lemma}\label{L:splitIntersect}
Let $e_i = v_iv_{i'}$ and $e_j = v_jv_{j'}$ be two (distinct) intersecting edges. Then, $e_i$ and $e_j$ are non-compatible.
\end{lemma}

\begin{proof}
Without loss of generality, assume that $i<i',j<j'$, and $i \leq j$. 
Targeting a contradiction, let $P' \in \mathcal{P}\setminus \{P\}$ be the path such that $e_i,e_j \in E(P')$. Moreover, let $s,t \in \{v_i, v_{i'}, v_j, v_{j'} \}$ be the two vertices such that $v_i, v_{i'}, v_j, v_{j'}$ appear in the $(s,t)$-subpath of $P'$.
 First, we prove the following claim.
\begin{claim}\label{O:splitPath}
$|\{s,t\}\cap \{v_i,v_{i'} \}|  \leq 1$. Similarly, $|\{s,t\}\cap \{v_j,v_{j'}\}| \leq 1$.
\end{claim}
\begin{proof}
We will show that if $s \in \{v_i,v_{i'} \}$, then $t \notin \{v_i,v_{i'} \}$. (The other cases are symmetric.) Without loss of generality, assume that $s= v_i$ and $t  = v_{i'}$. Since the $(s,t)$-subpath of $P'$ contains at least two edges ($e_i$ and $e_j$), we can replace the $(s,t)$-subpath with the edge $e_i$ to get a path $P''$ such that $E(P'') \subset E(P')$ and the endpoints of $P'$ and $P''$ are the same. Thus, we can replace $P'$ with $P''$ in $\mathcal{P}$ to get a solution with fewer edges, contradicting that $\mathcal{P}$ is a minimum solution.
\end{proof}

Next,  we will argue that we can reconfigure paths $P$ to $\widehat{P}$ and $P'$ to $\widehat{P'}$ such that $|E(\widehat{P})|+|E(\widehat{P'})| < |E(P)|+|E(P')|$ and $E(\widehat{P}) \cup E(\widehat{P'}) \subseteq E(P) \cup E(P')$. This will complete the proof, since then $\widehat{\mathcal{P}} = (\mathcal{P} \setminus \{P,P'\}) \cup \{\widehat{P},\widehat{P'} \}$ is a solution of $(G,\mathcal{X},k)$ having fewer edges than $\mathcal{P}$, contradicting the fact that $\mathcal{P}$ is a minimum solution.

\begin{figure}
\centering
\begin{subfigure}{.5\textwidth}
  \centering
  \includegraphics[width=0.99\linewidth]{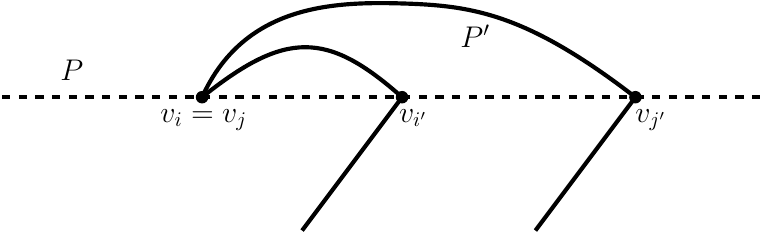}
\end{subfigure}%
\begin{subfigure}{.5\textwidth}
  \centering
  \includegraphics[width=0.99\linewidth]{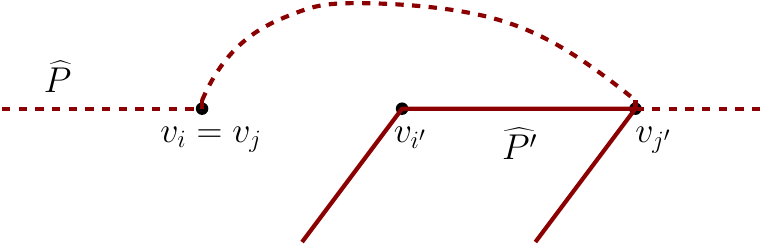}
\end{subfigure}
\caption{Here, $P$ and $P'$ are represented in dotted and solid black, respectively, in the left figure. Similarly, $\widehat{P}$ and $\widehat{P'}$ are represented in dotted and solid red, respectively, in the right figure.}
\label{fig:Case1}
\end{figure}

\begin{figure}
\centering
\begin{subfigure}{.5\textwidth}
  \centering
  \includegraphics[width=0.99\linewidth]{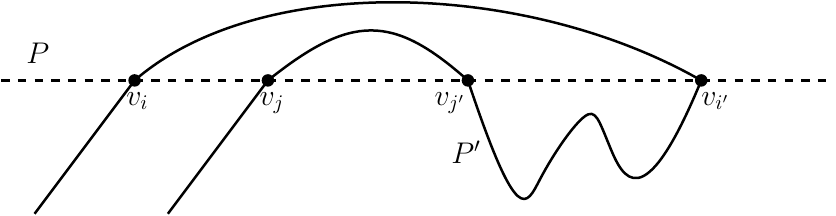}
\end{subfigure}%
\begin{subfigure}{.5\textwidth}
  \centering
  \includegraphics[width=0.99\linewidth]{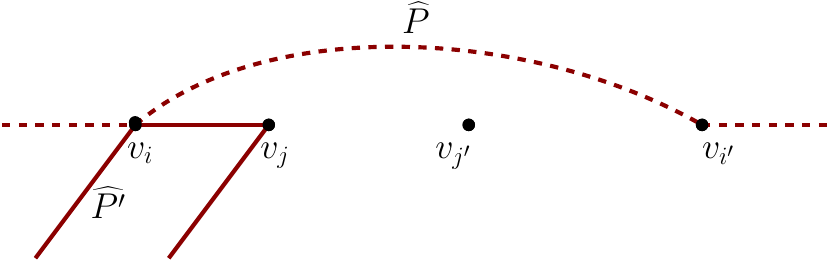}
\end{subfigure}
\caption{Here, $P$ and $P'$ are represented in dotted and solid black, respectively, in the left figure. Similarly, $\widehat{P}$ and $\widehat{P'}$ are represented in dotted and solid red, respectively, in the right figure.}
\label{fig:Case2}
\end{figure}

\begin{figure}
\centering
\begin{subfigure}{.5\textwidth}
  \centering
  \includegraphics[width=0.99\linewidth]{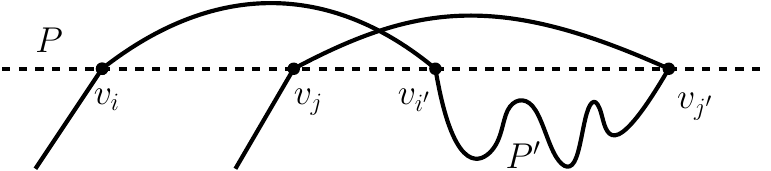}
\end{subfigure}%
\begin{subfigure}{.5\textwidth}
  \centering
  \includegraphics[width=0.99\linewidth]{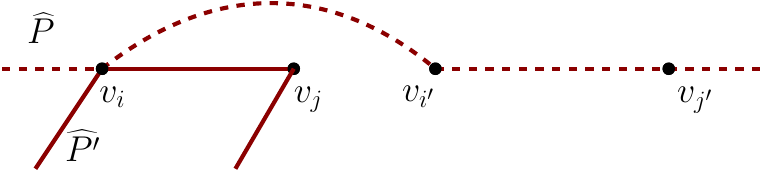}
\end{subfigure}
\caption{Here, $P$ and $P'$ are represented in dotted and solid black, respectively, in the left figure. Similarly, $\widehat{P}$ and $\widehat{P'}$ are represented in dotted and solid red, respectively, in the right figure.}
\label{fig:Case3}
\end{figure}

\begin{figure}
\centering
\begin{subfigure}{.5\textwidth}
  \centering
  \includegraphics[width=0.99\linewidth]{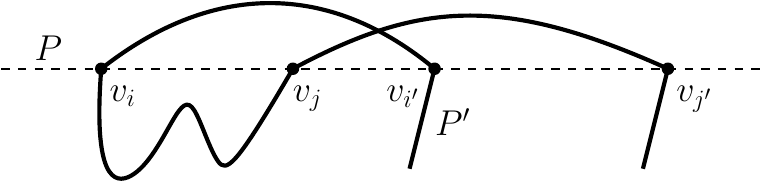}
\end{subfigure}%
\begin{subfigure}{.5\textwidth}
  \centering
  \includegraphics[width=0.99\linewidth]{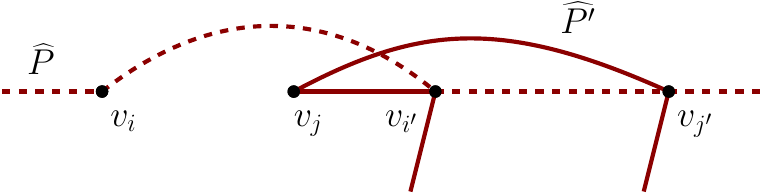}
\end{subfigure}
\caption{Here, $P$ and $P'$ are represented in dotted and solid black, respectively, in the left figure. Similarly, $\widehat{P}$ and $\widehat{P'}$ are represented in dotted and solid red, respectively, in the right figure.}
\label{fig:Case4}
\end{figure}

To this end, consider the following cases depending on the positions of $i,i',j,j'$. 
\smallskip

\noindent \textbf{Case 1:} $\bm{ i = j.}$ Since $e_i$ and $e_j$ are distinct edges, note that $i' \neq j'$. Hence, either $i' < j'$ or $i'> j'$. These two cases are symmetric, and therefore, we consider only the case when $i' <j'$.  See Figure~\ref{fig:Case1} for an illustration. Let $v = v_i = v_j$. Now, due to Claim~\ref{O:splitPath}, we have that $v \notin \{ s,t\}$. Therefore, either $s=v_{j'}$ and $t=v_{i'}$, or $s= v_{i'}$ and $t = v_{j'}$. Since these cases are symmetric, we assume $s=v_{j'}$ and $t=v_{i'}$. 
    In this case, we get $\widehat{P}$ by replacing the $(v_i,v_j')$-subpath of $P$ with the edge $v_jv_{j'}$, and we get $\widehat{P'}$ by replacing the $(v_{i'},v_{j'})$-subpath of $P'$ by the $(v_{i'},v_{j'})$-subpath of $P$. Observe that $E(\widehat{P}) \cup E(\widehat{P'}) \subseteq E(P) \cup E(P')$, and also $|E(\widehat{P})|+|E(\widehat{P'})| < |E(P)|+|E(P')|$ since the $(v_i,v_{i'})$-subpath of $P$ is removed. 
    \smallskip
    
 \noindent \textbf{Case 2:} $\bm{i' =j'.}$ This case is symmetric to Case 1.
    \smallskip
    
   \noindent \textbf{Case 3:} $\bm{i <j <j' < i'.}$ Due to Claim~\ref{O:splitPath}, either $s \in \{v_i,v_{i'}\}$ and $t\in \{v_j, v_{j'}\}$, or $s \in \{v_j, v_{j'}\}$ and $t\in \{v_i,v_{i'}\}$. Since these cases are symmetric, we assume that $s \in \{v_i,v_{i'}\}$ and $t\in \{v_j, v_{j'}\}$. Moreover, the case when $s=v_i$ and $t\in \{v_j, v_{j'} \}$ and the case when $s=v_{i'}$ and $t\in \{v_j, v_{j'} \}$ are symmetric (since we can just reverse the ordering of the vertices $v_1, \ldots, v_\ell$ to get the other case). Therefore, without loss of generality, we assume that $s=v_i$ and $t\in \{v_j, v_{j'} \}$. See Figure~\ref{fig:Case2} for an illustration.
    We get $\widehat{P}$ by replacing the $(v_i,v_{i'})$-subpath of $P$ with the edge $v_iv_{i'}$, and we get $\widehat{P'}$ by replacing the $(v_i,t)$-subpath of $P'$ with the $(v_i,t)$-subpath of $P$.  Observe that $E(\widehat{P}) \cup E(\widehat{P'}) \subseteq E(P) \cup E(P')$, and also $|E(\widehat{P})|+|E(\widehat{P'})| < |E(P)|+|E(P')|$ since the $(t,v_{i'})$-subpath of $P$ is removed. 
    \smallskip
    
    \noindent \textbf{Case 4:} {$\bm {i < j < i' < j'.}$}  Due to Claim~\ref{O:splitPath}, either $s \in \{v_i,v_{i'}\}$ and $t\in \{v_j, v_{j'}\}$, or $s \in \{v_j, v_{j'}\}$ and $t\in \{v_i,v_{i'}\} $. Since these cases are symmetric, we assume that $s \in \{v_i,v_{i'}\}$ and $t\in \{v_j, v_{j'}\}$. Here, we consider the following two cases:
   \smallskip
   
        \noindent \textbf{Subcase 4.1:} $\bm{s=v_i}$ \textbf{and} $\bm{t\in \{v_j,v_{j'}\}.}$ See Figure~\ref{fig:Case3} for an illustration. In this case, we obtain $\widehat{P}$ by replacing the $(v_i,v_{i'})$-subpath of $P$ with the edge $v_iv_{i'}$, and we obtain $\widehat{P'}$ by replacing the $(v_i,v_j)$-subpath of $P'$ by the $(v_i,v_j)$-subpath of $P$. Observe that $E(\widehat{P}) \cup E(\widehat{P'}) \subseteq E(P) \cup E(P')$, and also $|E(\widehat{P})|+|E(\widehat{P'})| < |E(P)|+|E(P')|$ since the $(v_j,v_{i'})$-subpath of $P$ is removed.
        
        \smallskip
        
        \noindent \textbf{Subcase 4.2:} $\bm{s= v_{i'}}$ \textbf{and} $\bm{t \in \{v_j,v_{j'}\}.}$ See Figure~\ref{fig:Case4} for an illustration. In this case, we obtain $\widehat{P}$ by replacing the $(v_i,v_{i'})$-subpath of $P$ with the edge $v_iv_{i'}$, and we obtain $\widehat{P'}$ by replacing the $(v_{i'},v_j)$-subpath of $P'$ by the $(v_{i'},v_j)$-subpath of $P$. Observe that $E(\widehat{P}) \cup E(\widehat{P'}) \subseteq E(P) \cup E(P')$, and also $|E(\widehat{P})|+|E(\widehat{P'})| < |E(P)|+|E(P')|$ since the $(v_j,v_{i})$-subpath of $P$ is removed.

This completes our proof.
\end{proof}

Now, we present the main lemma of this section.
\begin{lemma}\label{L:SPathLength}
Let $(G,\mathcal{X},k)$ be a Yes-instance of \EDP~where $G$ is a split graph. Moreover, let $\mathcal{P}$ be a minimum solution of $(G,\mathcal{X},k)$. Then, for every path $P \in \mathcal{P}$, $|E(P)| < 4  \sqrt{k} +4$.
\end{lemma}
\begin{proof}

Let $P$ be a path in $\mathcal{P}$ such that $P$ contains $\ell$ vertices from the clique $C$, and let $v_1,\ldots, v_\ell$ be an ordering of vertices of $C$ in $P$, along $P$ from one terminal to the other. First, we show that there are $\Omega (\ell^2)$ intersecting edges with both endpoints in $\{ v_1,\ldots, v_\ell\}$.  Let $\mathcal{S} = \{v_iv_j : 1\leq i \leq \big\lceil \frac{\ell}{2} \big\rceil-1, \ \big\lceil \frac{\ell}{2} \big\rceil +1 \leq j \leq \ell \}$. Observe that $\mathcal{S}$ is a set of pairwise intersecting edges, and hence, due to Lemma~\ref{L:splitIntersect}, $\mathcal{S}$ is a set of non-compatible edges. Moreover, it is easy to see that $|\mathcal{S}| = (\big\lceil \frac{\ell}{2} \big\rceil -1)(\big\lfloor\frac{\ell}{2} \big\rfloor)$. Furthermore, since each edge in $\mathcal{S}$ is of the form $v_iv_j$ such that $i,j \in [\ell]$ and $j\geq i+2$, due to Lemma~\ref{C:splitShortcut}, each edge in $\mathcal{S}$ is used by some path in $\mathcal{P} \setminus \{P\}$. Therefore, $|\mathcal{P}| > |\mathcal{S}|$.

Now, if  $|E(P)| \geq 4\sqrt{k}+4$, then $\ell \geq 2\sqrt{k}+2$ (since any three consecutive edges in $P$ require at least two vertices from $C$). In this case, we have that $|\mathcal{P}| > |\mathcal{S}| =(\Big\lceil \frac{2 \sqrt{k}  +2}{2} \Big\rceil -1)(\Big\lfloor\frac{2 \sqrt{k}  +2}{2} \Big\rfloor) = ( \sqrt{k}) ( \sqrt{k} +1) > k$, a contradiction (since $|\mathcal{P}| = k$).
\end{proof}

We have the following corollary as a consequence of Lemma~\ref{L:SPathLength} (since $k\geq 9$).
\begin{corollary}\label{C:SplitEdge}
Let $\mathcal{P}$ be a minimum solution of an instance $(G,\mathcal{X},k)$ of \EDP~where $G$ is a split graph. Then, $\sum_{P\in \mathcal{P}}|E(P)| \leq 5  k^{1.5} $.
\end{corollary}

\subsubsection{An \texorpdfstring{$\mathcal{O}(k^{2.75})$}{} Vertex Kernel for Split Graphs}\label{SS:Split}
In this section, we use Corollary~\ref{C:SplitEdge} stating that there can be at most $5 k^{1.5} $ edges in any minimum solution to design a subcubic ($\mathcal{O}(k^{2.75})$) vertex kernel for \EDP~on split graphs. 

We start with the following preprocessing step, which we apply only once to our input instance. 
\smallskip

\noindent \textbf{\textsc{Preprocessing Step}:} First, we use the kernelization for \EDP~on split graphs provided by Heggernes et al.~\cite{hegger} as a preprocessing step. In their kernel, if $|C| \geq 2k$, then they report a Yes-instance (due to~\cite[Lemma 8]{hegger}), and hence, assume that $|C| <2k$. Due to \cref{prop2}, if $|C| >k$, then we have a Yes-instance, and hence we assume that $|C|\leq k$. Moreover, in their kernel, for any two vertices $u,w \in C$, $|N(u)\cap N(w)\cap I_N| \leq 4k+1$ (i.e., $u$ and $w$ have at most $4k+1$ common neighbors in $I_N$). Furthermore, there are no pendant vertices in $I_N$.

Next, we define a \textsc{Marking Procedure}, where we label the vertices in $C$ as \textit{rich} or \textit{poor}. Furthermore, we partition the vertices in $I_N$ into two sets, denoted $U$ (read {\em unmarked}) and $M$ (read {\em marked}), in the following manner. 
\smallskip

 \noindent \textbf{\textsc{Marking Procedure}:} Let $(G,\mathcal{X},k)$ be an instance of \textsc{EDP} where $G$ is a split graph.
\begin{enumerate}
    \item  $M \Leftarrow \emptyset$ and $U \Leftarrow I_N$. (Initially, all vertices in $I_N$ is unmarked.) Moreover, fix an ordering $v_1, \ldots, v_{|C|}$ of the vertices of $C$.
    \item For $1\leq i \leq |C|$:
    \begin{enumerate}
        \item [2.1] $A_{v_i} \Leftarrow \emptyset$, $M_{v_i} \Leftarrow \emptyset$ (read {\em marked for $v_i$}), and $U_T = U$ (read {\em unmarked temporary}).
         \item [2.2] For $1\leq j \leq |C|$ such that $i \neq j$ and $|A_{v_i}| <100 k^{0.75}$:
         \begin{enumerate}
             \item [2.2.1] If $|N(v_i) \cap N(v_j) \cap {U_T}| \geq k^{0.75}$, then $A_{v_i} \Leftarrow A_{v_i} \cup \{v_j\}$. Moreover, select some (arbitrary) subset $M_{v_i,v_j} \subseteq N(v_i) \cap N(v_j) \cap {U_T}$ such that $|M_{v_i,v_j}| = k^{0.75}$. Then, $M_{v_i} \Leftarrow M_{v_i} \cup M_{v_i,v_j}$ and $U_T \Leftarrow U_T \setminus M_{v_i,v_j}$. 
         \end{enumerate}

 \item [2.3] If $|A_{v_i}|= 100 k^{0.75}$, then label $v_i$ as {\em rich}. Moreover, $M \Leftarrow M \cup M_{v_i}$ and $U \Leftarrow U_T$.
 
 \item [2.4] If $|A_{v_i}|< 100 k^{0.75}$, then label $v_i$ as {\em poor}.
 \end{enumerate}
\end{enumerate}
This completes our \textsc{Marking Procedure}. 

\begin{remark}
Note that the definition of {\em rich} and {\em poor} depends on the order in which our \textsc{Marking Procedure} picks and marks the vertices (i.e., being rich or poor is not an intrinsic property of the vertex itself). A different execution of the above procedure can label different vertices as rich and poor. Moreover, note that if $M_{v,x}$ exists (i.e., $v$ is rich and $x\in A_{v}$) and $M_{x,v}$ exists (i.e., $x$ is rich and $v\in A_{x}$), then $M_{x,v}\cap M_{v,x} = \emptyset$.
\end{remark}

We have the following observation regarding the vertices marked rich by an execution of \textsc{Marking Procedure} on an instance $(G,\mathcal{X},k)$ of \EDP~where $G$ is a split graph.
\begin{observation}\label{O:rich}
Consider an execution of \textsc{Marking Procedure} on an instance $(G,\mathcal{X},k)$ of \EDP~where $G$ is a split graph. Then for a rich vertex $v$, $|M_v| = 100 k^{1.5}$ (i.e., the number of vertices marked in $I_N$ for $v$ are $100k^{1.5}$).
\end{observation}
\begin{proof}
Notice that a vertex $v\in C$ is rich if $|A_v| = 100k^{0.75}$. Moreover, for each vertex $x \in A_v$, we mark exactly $k^{0.75}$ previously unmarked vertices in $I_N$ (denoted $M_{v,x}$). Hence, for any two distinct vertices $x,y \in A_v$, $M_{v,x} \cap M_{v,y} =\emptyset$. Therefore, $|M_v| = \sum_{w\in A_v} |M_{v,w}| = 100k^{0.75}\times k^{0.75} = 100k^{1.5}$.
\end{proof}

\begin{definition}[Reachable Vertices]\label{D:reachable}
Consider an execution of \textsc{Marking Procedure} on an instance $(G,\mathcal{X},k)$ of \EDP~where $G$ is a split graph. Moreover, let $\mathcal{P}$ be a solution of $(G,\mathcal{X},k)$. Then, for a rich vertex $v \in C$, let $R_v\subseteq A_v$ (read {\em reachable from $v$}) denote the set of vertices such that for each vertex $x\in R_v$, there is a vertex $u \in M_{v,x}$ such that $u$ is not used by any path in $\mathcal{P}$. 
\end{definition}

Notice that, in Definition~\ref{D:reachable}, a path of the form $(v,u,x)$ is edge-disjoint from every path in $\mathcal{P}$. Furthermore, we briefly remark that $R_v$ is defined with respect to the execution of \textsc{Marking Procedure} and a solution $\mathcal{P}$ of $(G,\mathcal{X},k)$, which we will always fix before we use $R_v$.

Let $\mathcal{P}$ be a solution to an instance $(G,\mathcal{X},k)$ of \EDP. Informally speaking, in the following lemma, we show that if $\mathcal{P}$ uses at most $6k^{1.5}$ edges, then for a rich vertex $v$, $R_v = \Omega(k^{0.75})$ (i.e., there are ``many reachable'' vertices in $A_v$ from $v$ using paths that are edge-disjoint from every path in $\mathcal{P}$). In particular, we have the following lemma.

\begin{lemma}\label{L:auxillary}
Consider an execution of \textsc{Marking Procedure} on an instance $(G,\mathcal{X},k)$ of \EDP~where $G$ is a split graph. Moreover, let $\mathcal{P}$ be a solution of $(G,\mathcal{X},k)$ (not necessarily minimum) such that the total number of edges used in $\mathcal{P}$ is at most $6k^{1.5}$. Then, for any rich vertex $v\in C$, $|R_v| \geq 94k^{0.75}$.
\end{lemma}
\begin{proof}
Note that in \textsc{Marking Procedure}, for each vertex  $x \in A_v$, we mark exactly $k^{0.75}$ vertices (whose set is denoted by $M_{v,x}$). Since $v$ is a rich vertex, due to Observation~\ref{O:rich}, $|M_v| = 100k^{1.5}$. Since the total number of edges used by $\mathcal{P}$ are at most $6k^{1.5}$, the total number of vertices used in $I$ by $\mathcal{P}$ can be at most $6k^{1.5}$ as well. As $M_v \subseteq I$, there are at least $94 k^{1.5}$ vertices in $M_v$ that are not used by any path in $\mathcal{P}$.  

Targeting contradiction, assume that $|R_v| <94k^{0.75}$. By definition of $R_v$ (Definition~\ref{D:reachable}), for every vertex $y \in A_v \setminus{R_v}$, each vertex in $M_{v,y}$ is used by some path in $\mathcal{P}$. Since for each $x\in A_v$, $|M_{v,x}| = k^{0.75}$, where the number of vertices in $M_v$ that are not used by any path in $\mathcal{P}$ is at most $|R_v|\times k^{0.75} < 94k^{0.75}\times k^{0.75} = 94k^{1.5}$, a contradiction.
\end{proof}

Next, we provide the following reduction rule.

\begin{RR} [RR\ref{RRSE1}] \label{RRSE1}
Let $(G,\mathcal{X},k)$ be an instance of \EDP~where $G$ is a split graph. Let $U$ be the set of unmarked vertices we get after an execution of \textsc{Marking Procedure} on $(G,\mathcal{X},k)$. Moreover, let $U'\subseteq U$ be the set of vertices in $U$ that do not have a poor neighbor. If $U' \neq \emptyset$, then $G' \Leftarrow G- U'$ and $\mathcal{X}' \Leftarrow \mathcal{X}$.
\end{RR}

The following two lemmas (Lemmas \ref{L:auxillary1} and \ref{L:splitEdgeMain}) are essential to prove the safeness of RR\ref{RRSE1}.

\begin{lemma}\label{L:auxillary1}
Let $(G,\mathcal{X},k)$ be an instance of \EDP~where $G$ is a split graph. Consider an execution of \textsc{Marking Procedure} on $(G,\mathcal{X},k)$. Moreover, let $\mathcal{P}$ be a solution of $(G,\mathcal{X},k)$ (not necessarily minimum) such that the total number of edges used in $\mathcal{P}$ is $\ell$, where  $\ell \leq 6k^{1.5}$. Furthermore, let $u\in U$ be an unmarked vertex such that $u$ does not have any poor neighbor. If there is a path $P \in \mathcal{P}$ such that $u \in V(P)$, then there exists a solution $\mathcal{P}' = (\mathcal{P} \setminus \{P\}) \cup \{P'\}$ of $(G,\mathcal{X},k)$ such that $u\notin V(P')$, and the total number of edges in $\mathcal{P}'$ is at most $\ell+3$.
\end{lemma}

\begin{proof}
First, note that $u$ is not an endpoint of $P$ since $u \in I_N$. Let the immediate neighbors of $u$ in $P$ be $v$ and $w$ (i.e., $P$ is of the form $(s, \ldots, v,u,w,$ $\ldots, t)$). Since $u$ does not have any poor neighbor and $u\in I$, note that both $v$ and $w$ are rich vertices. Therefore, due to Lemma~\ref{L:auxillary} (and the fact that $\ell \leq 6k^{1.5}$), $|R_v| \geq 94k^{0.75}$ and $|R_w| \geq 94k^{0.75}$. See Figure~\ref{fig:splitRule} for an abstract illustration.

\begin{figure}
    \centering
    \includegraphics[scale= 0.8]{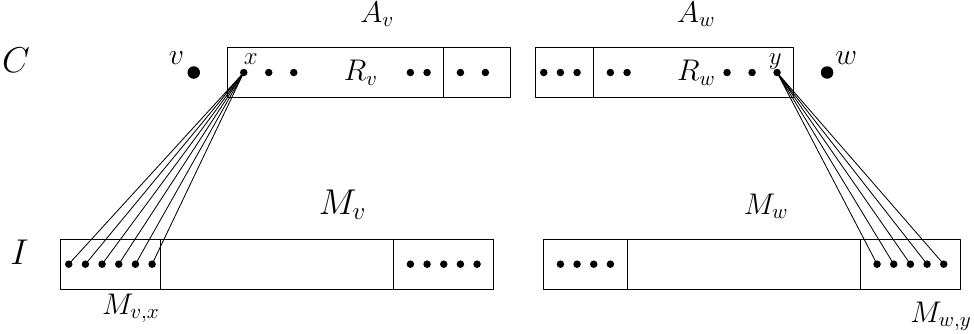}
    \caption{An abstract illustration of the notations used in the proof of Lemma~\ref{L:auxillary1}.}
    \label{fig:splitRule}
\end{figure}

Observe that both $R_v$ and $R_w$ are subsets of $C$. Now, we have one of the following cases. 
\smallskip

    \noindent \textbf{Case 1:} \bm{$R_v \cap R_w \neq \emptyset.$} Let $x \in R_v\cap R_w$. Now, due to definition of $R_v$ and $R_w$ (Definition~\ref{D:reachable}), there are vertices $x' \in M_{v,x}$ and $y'\in M_{w,x}$ such that none of $x'$ and $y'$ is used by any path in $\mathcal{P}$. So, the edges used in the path $\widehat{P} = (v,x',x,y',w)$ are not used by any path in $\mathcal{P}$ (as either $x'$ or $y'$ is an endpoint of each edge in $\widehat{P}$). Therefore, we obtain $P'$ be replacing the $(v,w)$-subpath of $P$ with the path $\widehat{P}$. Observe that $\mathcal{P}' = (\mathcal{P} \setminus \{P\}) \cup \{P'\}$ is a solution of $(G,\mathcal{X},k)$, $u\notin V(P')$, and the total number of edges in $\mathcal{P}'$ is at most $\ell+2$ (since the $(u,v)$-subpath of $P$ has two edges and $\widehat{P}$ has four edges).
    \smallskip
    
\noindent \textbf{Case 2:} \bm{$R_v \cap R_w = \emptyset.$} Recall that, due to Lemma~\ref{L:auxillary}, $|R_v| \geq 94k^{0.75}$ and $|R_w| \geq 94k^{0.75}$.  Since $R_v \subseteq C$, $R_w\subseteq C$, and $R_v \cap R_w = \emptyset$, there are at least $8836k^{1.5}$ edges of the form $xy$ such that $x\in R_v$ and $y \in R_w$. Since $\mathcal{P}$ uses at most $6k^{1.5}$ edges, there is at least one edge $xy$ that is not used by any path in $\mathcal{P}$ such that $x\in R_v$ and $y \in R_w$.  Moreover, due to definition of $R_v$ and $R_w$ (Definition~\ref{D:reachable}), there are vertices $x'\in M_{v,x}$ and $y'\in M_{w,y}$ such that none of $x'$ and $y'$ is used by any path in $\mathcal{P}$. So, the edges used in the path $\widehat{P} = (v,x',x,y,y',w)$ are not used by any path in $\mathcal{P}$. Therefore, we obtain $P'$ by replacing the $(v,w)$-subpath of $P$ with the path $\widehat{P}$. Observe that $\mathcal{P}' = (\mathcal{P} \setminus \{P\}) \cup \{P'\}$ is a solution of $(G,\mathcal{X},k)$, $u\notin V(P')$, and the total number of edges in $\mathcal{P}'$ is at most $\ell+3$ (since the $(u,v)$-subpath of $P$ has two edges and $\widehat{P}$ has five edges).

This completes our proof.
\end{proof}

\begin{lemma}\label{L:splitEdgeMain}
Let $(G,\mathcal{X},k)$ be an instance of \EDP~where $G$ is a split graph. Let $U$ be the set of unmarked vertices we get after an execution of \textsc{Marking Procedure} on $(G,\mathcal{X},k)$, and let $u\in U$ be a vertex such that $u$ does not have any poor neighbor. Let $G' = G- \{u\}$. Then, $(G,\mathcal{X},k)$ is a Yes-instance if and only if $(G',\mathcal{X},k)$ is a Yes-instance. 
\end{lemma}
\begin{proof}
We claim that $(G,\mathcal{X},k)$ and $(G',\mathcal{X},k)$ are equivalent instances of \textsc{EDP}.

In one direction, if $(G',\mathcal{X},k)$ is a Yes-instance of \textsc{EDP},
then, because $G'$ is a subgraph of $G$, $(G,\mathcal{X},k)$ is also a Yes-instance of \textsc{EDP}.
    
In the other direction, suppose $(G,\mathcal{X},k)$ is a Yes-instance of \textsc{EDP}. Furthermore, let $\mathcal{P}$ be a minimum solution of $(G,\mathcal{X},k)$. If $u$ does not participate in any path in $\mathcal{P}$, then note that $\mathcal{P}$ is a solution to $(G',\mathcal{X},k)$ as well.
    
So, let $u$ be a vertex in paths $P_1, \ldots, P_r$ in $\mathcal{P}$. Observe that $u$ is an internal vertex of all these paths (since $u \in I_N$) and $u$ does not have any poor neighbor (due to the definition of $u$). Now, for each path $P_i$ ($i\in [r]$), we will create a path $P'_i$ such that $(\mathcal{P} \setminus \{P_1, \ldots, P_r\}) \cup \{P'_{1}, \ldots, P'_{r}\}$ is a solution to $(G',\mathcal{X},k)$, and the endpoints of $P'_i$ are the same as those of $P_{i}$.

Since $\mathcal{P}$ is a minimum solution, due to Corollary~\ref{C:SplitEdge}, $\mathcal{P}$ uses at most $5k^{1.5}$ edges. Therefore, due to Lemma~\ref{L:auxillary1}, we can get a solution $\mathcal{P}_1 = (\mathcal{P}\setminus \{P_1\}) \cup \{P'_1\}$ of $(G,\mathcal{X},k)$ such that $\mathcal{P}_1$ uses at most $5k^{1.5}+3$ edges. Similarly, for $2 \leq i \leq r$, we can get a solution $(\mathcal{P}_{i-1} \setminus \{P_i\}) \cup \{ P'_i\}$ of $(G,\mathcal{X},k)$ such that $\mathcal{P}_i$ uses at most $5k^{1.5}+3i$ edges (which is always less than $6k^{1.5}$ since $i\leq r \leq k$ and $k\geq 9$) due to Lemma~\ref{L:auxillary1}. 
Observe that no path in $\mathcal{P}_r$ contains the vertex $u$. Therefore, $\mathcal{P}_r$ is a solution to $(G',\mathcal{X},k)$ as well (as $\mathcal{P}_r$ is a solution of $(G,\mathcal{X},k)$). Hence, $(G',\mathcal{X},k)$ is a Yes-instance.
\end{proof}
    
We have the following lemma to prove the safeness of RR\ref{RRSE1}.
\begin{lemma}\label{L:RRSE1}
RR\ref{RRSE1} is safe.
\end{lemma}
\begin{proof}
By Remark \ref{obsdel} (in Section~\ref{sec:prelim}), note that $G'$ is a split graph. The rest of the proof follows from Lemma~\ref{L:splitEdgeMain}.
\end{proof}

Next, we show that after an exhaustive application of RR\ref{RRSE1}, we get a subcubic vertex kernel. 

\begin{lemma}\label{L:splitBound}
Let $(G,\mathcal{X},k)$ be an instance of \EDP~where $G$ is a split graph. If we cannot apply RR\ref{RRSE1} on $(G,\mathcal{X},k)$, then $|V(G)| = \mathcal{O}(k^{2.75})$. 
\end{lemma}
\begin{proof}
Consider an execution of \textsc{Marking Procedure} on $(G,\mathcal{X},k)$ to get $U,M$, and rich and poor vertices. 

First, we count the number of marked vertices in $I_N$, being $|M|$. Note that we mark vertices in $I_N$ only for rich vertices. Consider a rich vertex $v$ (in $C$). For a rich vertex $v$, $M_v = 100k^{1.5}$ (Observation~\ref{O:rich}).  Therefore, the total number of marked vertices in $I_N$ is at most $100|C|k^{1.5}$ (if each vertex in $C$ is rich). Since $|C| \leq k$ (due to \textsc{Preprocessing Step}), the total number of marked vertices in $I_N$ are at most $100 k^{2.5} = \mathcal{O}(k^{2.5})$ .


Second, we count the cardinality of unmarked vertices $U$ in $I_N$. Since we cannot apply RR\ref{RRSE1}, each vertex $w \in U$ has at least one poor vertex as a neighbor. Therefore, $|U| \leq \sum_{v: v \text{ is poor}} |N(v)\cap U|$. So, consider a poor vertex $v_i$ in $C$. We claim that $|N(v_i) \cap U| = \mathcal{O}(k^{1.75})$. Targeting contradiction, assume that $|N(v_i)\cap U| > 400k^{1.75}+100k^{0.75}$.  Since any two vertices in $C$ can have at most $4k+1$ common neighbors in $I_N$ (due to \textsc{Preprocessing Step}), there are at least $100k^{0.75}$ vertices in $A_{v_i}$ (since \textsc{Preprocessing Step} ensures that there are no pendant vertices in $I_N$), a contradiction to the fact that $v_i$ is a poor vertex. Therefore, for each poor vertex $v_i$ (in $C$),  $|N(v_i)\cap U| = \mathcal{O}(k^{1.75})$. Since there are at most $k$ poor vertices (as $|C| \leq k$), we have that $|U| = \mathcal{O} (k^{2.75})$.

Since $|V(G)| = |C|+|I_T|+|U|+|M|$, we have that $|V(G)| = \mathcal{O}(k^{2.75})$ (as $|C| \leq k$ and $|I_T| \leq 2k$).
\end{proof}

Now, we have the following observation.
\begin{observation} \label{O:splitEdge}
RR\ref{RRSE1}, \textsc{Marking Procedure}, and \textsc{Preprocessing Step} can be executed in polynomial time. Moreover, our initial parameter $k$ does not increase during the application of RR\ref{RRSE1}.
\end{observation}

Finally, due to \textsc{Preprocessing Step}, RR\ref{RRSE1}, Observation~\ref{O:splitEdge}, and Lemmas~\ref{L:RRSE1} and~\ref{L:splitBound}, we have the following theorem.

\splitedp*

\subsection{A Linear Vertex Kernel for Threshold Graphs} \label{secthreshold}

Let $(G,\mathcal{X},k)$ be an instance of \textsc{EDP} where $G$ is a threshold graph. Note that given a threshold graph $G$, we can compute (in linear time) a partition $(C, I)$ of $V(G)$ such that $C$ is a clique, and $I$ is an independent set~\cite{hammer}. If $|C|\geq k+1$, then by Corollary \ref{prop2}, $(G,\mathcal{X},k)$ is a Yes-instance. So, without loss of generality, we assume that $|C|\leq k$. Furthermore, let us partition the set $I$ into two sets, say, $I_{T}$ and $I_{N}$, where $I_{T}$ contains the terminal vertices and $I_{N}$ contains the non-terminal vertices.
Since there are at most $2k$ terminal vertices in $G$, we have $|I_{T}|\leq 2k$. If $|I_{N}|\leq 4k+1$, then we have a kernel of at most $7k+1$ vertices, and we are done. So, assume otherwise (i.e., $|I_{N}|\geq 4k+2$). By the definition of a threshold graph, there exists an ordering, say, $(v_{1},v_{2},\ldots,v_{|I_{N}|})$, of $I_{N}$ such that $N(v_{|I_{N}|})\subseteq N(v_{|I_{N}|-1}) \subseteq \cdots \subseteq N(v_{1})$, which can be computed in linear time~\cite{heggernes2007linear}. Let $R=\{v_{4k+2},\ldots,v_{|I_N|}\}$. Note that, since $|I_{N}|\geq 4k+2$, $R\neq \emptyset$.

Now, consider the following reduction rule.

\begin{RR}[RR\ref{RRT1}]\label{RRT1}
If $R\neq \emptyset$, then $G'\Leftarrow G-R$, $\mathcal{X'} \Leftarrow \mathcal{X}$, $k' \Leftarrow k$.
\end{RR}

We have the following lemma to establish that RR\ref{RRT1} is safe.

\begin{lemma} \label{lemma8}
RR\ref{RRT1} is safe.
\end{lemma}
\begin{proof}
First, note by Remark \ref{obsdel} that $G'$ is also a threshold graph. Now, we claim that $(G,\mathcal{X},k)$ and $(G',\mathcal{X'},k')$ are equivalent instances of \textsc{EDP} on threshold graphs. 

In one direction, if $(G',\mathcal{X'},k')$ is a Yes-instance of \textsc{EDP},
then, because $G'$ is a
    subgraph of $G$, $(G,\mathcal{X},k)$ is also a Yes-instance of \textsc{EDP}.
    
    In the other direction, suppose $(G,\mathcal{X},k)$ is a Yes-instance of \textsc{EDP}. Furthermore, let $\mathcal{P}$ be a solution for which $\sum_{P\in \mathcal{P}}|V(P)\cap R|$ (i.e., the total number of visits by all the paths in $\mathcal{P}$ to vertices in $R$) is minimum. We claim that no path in $\mathcal{P}$ visits a vertex in $R$. Targeting a contradiction, suppose that there exists a path, say, $P\in \mathcal{P}$, such that $P$ visits some vertex $r\in R$.
Since $r$ is not a terminal vertex, there are two vertices, say, $v,w \in C$, such that the edges $vr$ and $wr$ appear
consecutively on the path $P$. Since $N_{G}(r)\subseteq N_{G}(x)$ for every $x\in I_{N}\setminus R$, it is clear that every vertex in $I_{N}\setminus R$ is adjacent to both $v$ and $w$. Now, we claim the following.

\begin{claim} \label{clmlabel}
For every $x\in I_{N}\setminus R$, at least one of $vx$ and $wx$ belongs to some path in $\mathcal{P}$. In particular, at least $4k+1$ edges in the paths in $\mathcal{P}$ are incident with $v$ or $w$.
\end{claim}

\begin{proof}
Let us assume that there exists a vertex, say, $y \in I_{N}\setminus R$, such that both the edges $vy$ and $wy$ are not used by any path in $\mathcal{P}$. In such a situation, one can replace the vertex $r$ with vertex $y$ in the path $P$. It leads to a contradiction to the choice of $\mathcal{P}$. Thus, at least $4k+1$ edges in the paths in $\mathcal{P}$ are incident with $v$ or $w$.
\end{proof} 

Note that any path in $\mathcal{P}$ can use at most two edges incident on $v$. The same is true for $w$. This implies that in total at most $4k$ edges incident on $v$ or $w$ (combined) can belong to the paths in $\mathcal{P}$, which is a contradiction to  Claim \ref{clmlabel}. Therefore, no path in $\mathcal{P}$ visits a vertex in $R$, and thus $\mathcal{P}$ is a solution of $(G,\mathcal{X},k)$ as well.
\end{proof}

\begin{observation} \label{obs8}
RR\ref{RRT1} can be applied in polynomial
time.
\end{observation}

Note that by applying RR\ref{RRT1}, we ensure that $|I_{N}|\leq 4k+1$, which gives us a kernel with at most $7k+1$ vertices. By Lemma \ref{lemma8} and Observation \ref{obs8}, we have the following theorem.
\threshold*

\subsection{A Quadratic Vertex Kernel for Block Graphs} \label{quadratic:edp:block}

In this section, we show that \EDP~on block graphs admits a kernel with at most $4k^2-2k$ vertices. Let us first discuss the overall idea leading us to this result.
\medskip

\noindent \textbf{Overview.}
Let $(G,\mathcal{X},k)$ be an instance of \textsc{EDP}~where $G$ is a block graph. First, we aim to construct a reduced instance where the number of blocks can be at most $4k-2$. We begin by showing that if there is an end block that does not contain any terminal, then we can delete this block from the graph (in RR\ref{RRB1}), while preserving all solutions. Next, we argue that if there is a block, say, $B$, with at most two cut vertices that do not contain any terminal, then we can either \textit{contract} (defined in Definition~\ref{D:contract}) $B$ to a single vertex, or answer negatively (in RR\ref{RRB3}). Thus, each block with at most two cut vertices in the (reduced) graph contains at least one terminal. This bounds the number of blocks with at most two cut vertices to be at most $2k$ (as $k$ terminal pairs yield at most $2k$ terminals). Next, we use the following folklore property of block graphs (For the sake of completeness, we give the proof also).
\begin{observation}\label{P:block}
Let $\ell$ be the number of end blocks in a block graph $G$. Then, the number of blocks with at least three cut vertices is at most $\ell-2$.
\end{observation}
\begin{proof}
Let $\{B_1, \ldots, B_k\}$ and $\{f_1, \ldots, f_{p}\}$
be the set of blocks and the set of cut-vertices, respectively, of the block graph $G$. Then, the
\emph{cut-tree} of $G$ is the tree $T_G=(V', E')$, where $V' = \{B_1,\ldots, B_k, f_1, \ldots, f_p\}$
and $E' = \{(B_i, f_j) \mid f_j \in B_i, 1 \leq i \leq k, 1 \leq j \leq p\}.$  Note that every pendant vertex in $T_{G}$ corresponds to an end block of $G$. Thus, the proof follows from the following property of trees~\cite{diestel}. In any tree, if $D_{1}$ denotes the number of pendant vertices and $D_{\geq 3}$ denotes the number of vertices of degree at least 3, then $ D_{\geq 3} \leq D_{1} -2$.
\end{proof}

Observation~\ref{P:block}, along with the fact that the number of end blocks is at most $2k$, establishes that the number of blocks with at least three cut vertices in the (reduced) graph is at most $2k-2$.  Therefore, we have at most $4k-2$ blocks in the (reduced) graph. Finally, due to  Lemma~\ref{L:cliqueBetter} and the properties of block graphs, we show that if a block, say, $B$, is big enough (i.e., $|V(B)| >k$), then we can contract $B$ to a single vertex while preserving the solutions (in RR\ref{RRB2}). Hence, each block contains at most $k$ vertices, and thus, the total number of vertices in the (reduced) graph is at most $4k^2-2k$.

\medskip

Now, we discuss our kernelization algorithm, which is based on the application of three reduction rules (RR\ref{RRB1}-RR\ref{RRB3}), discussed below, to an instance $(G,\mathcal{X},k)$ of \EDP~where $G$ is a block graph. 

\begin{RR}[RR\ref{RRB1}]\label{RRB1}
If $B$ is an end block of $G$ with cut vertex $v$ such that $B$ does not contain any terminal, then $G' \Leftarrow G[V(G) \setminus (V(B)\setminus \{v\})]$ and $\mathcal{X}'\Leftarrow \mathcal{X}$.
\end{RR}

We have the following lemma to establish that RR\ref{RRB1} is safe.

\begin{lemma}\label{L:RRB1}
RR\ref{RRB1} is safe.
\end{lemma}
\begin{proof}
First, observe that $G'$ is a block graph. Now, we claim that $(G,\mathcal{X},k)$ and $(G',\mathcal{X'},k)$ are equivalent instances of \textsc{EDP} on block graphs.

In one direction, if $(G',\mathcal{X'},k)$ is a Yes-instance of \textsc{EDP},
then, because $G'$ is a subgraph of $G$, $(G,\mathcal{X},k)$ is also a Yes-instance of \textsc{EDP}.
    
In the other direction, suppose $(G,\mathcal{X},k)$ is a Yes-instance of \textsc{EDP}, and let $\mathcal{P}$ be a solution of $(G,\mathcal{X},k)$. Observe that no path between two vertices $u_1,u_2 \in V(G)\setminus V(B)$ passes through a vertex $u \in V(B)\setminus\{v\}$, as otherwise, the vertex $v$ appears at least twice in this path, which contradicts the definition of a path. Therefore, since $B$ does not contain any terminal, no path in $\mathcal{P}$ passes through a vertex of $V(B)\setminus\{v\}$. As $\mathcal{X'} = \mathcal{X}$ and each path in $\mathcal{P}$ is restricted to $V(G')$, observe that $\mathcal{P}$ is also a solution of $(G',\mathcal{X'},k)$.
\end{proof}

\begin{observation}
RR\ref{RRB1} can be applied in polynomial
time.
\end{observation}

The following definitions (Definitions~\ref{D:contract} and \ref{D:virtual}) are crucial to proceed further in this section. Informally speaking, we contract a block $B$ by replacing it with a (new) vertex $v$ such that ``edge relations'' are preserved. We have the following formal definition.

\begin{definition}[Contraction of a Block]\label{D:contract}
Let $(G,\mathcal{X},k)$ be an instance of \EDP~where $G$ is a block graph, and let $B$ be a block of $G$.  The {\em contraction of $B$} in $(G,\mathcal{X},k)$ yields another instance $(G',\mathcal{X}',k')$ of \EDP~as follows. First,  $V(G') = (V(G)\setminus V(B)) \cup \{v\}$ (i.e., delete $V(B)$ and add a new vertex $v$). Moreover, define $f:V(G) \rightarrow V(G')$ such that $f(x) = x$ if $x\in V(G)\setminus V(B)$, and $f(x) = v$ if $x\in V(B)$.  Second, $E(G') = \{f(x)f(y) : xy \in E(G), \ f(x)\neq f(y) \}$. Similarly, $\mathcal{X}' = \{(f(s),f(t)) : (s,t) \in \mathcal{X}, \ f(s)\neq f(t) \}$.  Finally, let $k' = |\mathcal{X}'|$. Note that $k'$ might be smaller than $k$ (in case $B$ contains a terminal pair). Moreover, $\bigcup_{u \in V(B)} (N_G(u) \setminus B) \subseteq N_{G'}(v)$. See Figure~\ref{fig:contract} for an illustration.
\end{definition}

\begin{figure}
\centering
\begin{subfigure}{.5\textwidth}
  \centering
  \includegraphics[width=1\linewidth]{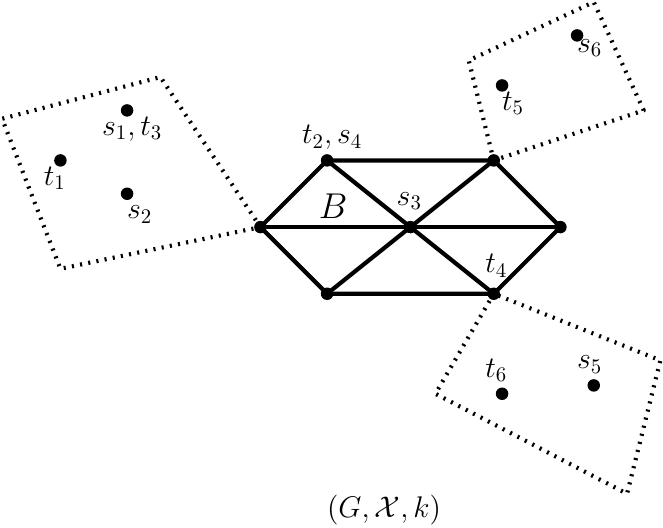}
  \label{fig:sub1}
\end{subfigure}%
\begin{subfigure}{.5\textwidth}
  \centering
  \includegraphics[width=0.75\linewidth]{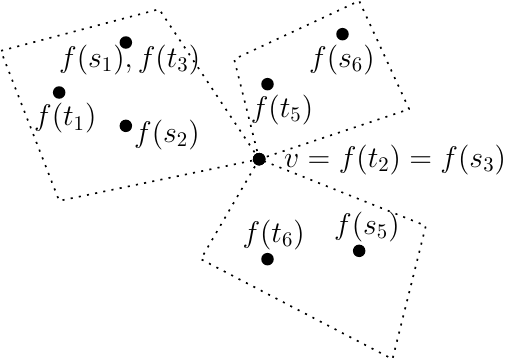}
  \label{fig:sub2}
\end{subfigure}
\caption{Here, we contract $B$ in $(G,\mathcal{X},k)$ to obtain $(G',\mathcal{X}',k')$. Note that $(f(s_4),f(s_4)) \notin \mathcal{X}'$ since $f(s_4)=f(t_4) = v$.}
\label{fig:contract}
\end{figure}

Now, we will exploit the properties of block graphs to show that if a block $B$ has at least $k+1$ vertices, then we can contract $B$ to a single vertex ``safely''. For this purpose, we define the instance $(G,\mathcal{X},k)$ of \EDP~restricted to a block $B$ of the (connected) block graph $G$ as follows.

\begin{definition}[Restriction of an Instance $(G,\mathcal{X},k)$ to a Block]\label{D:virtual}
Consider a  block $B$ whose set of cut vertices is  $U$. For each cut vertex $u \in U$, let $C_{B,u}$ denote the (connected)  component of $G[V(G) \setminus (V(B) \setminus \{u\})]$ containing $u$. Now, define $h:V(G) \rightarrow V(B)$ such that $h(x) = x$ if $x \in V(B)$, and $h(x) = u$ if $x\in V(C_{B,u})$. Then, the {\em restriction of $(G,\mathcal{X},k)$ to $B$}, denoted by $(B,\mathcal{X}_B, |\mathcal{X}_B|)$, is defined as follows: $\mathcal{X}_B = \{(h(s),h(t)) : (s,t) \in \mathcal{X}, \ h(s)\neq h(t) \}$. See Figure~\ref{Fig:restrictBlock} for an illustration.
\end{definition}

\begin{figure}
    \centering
    \includegraphics[scale=0.70]{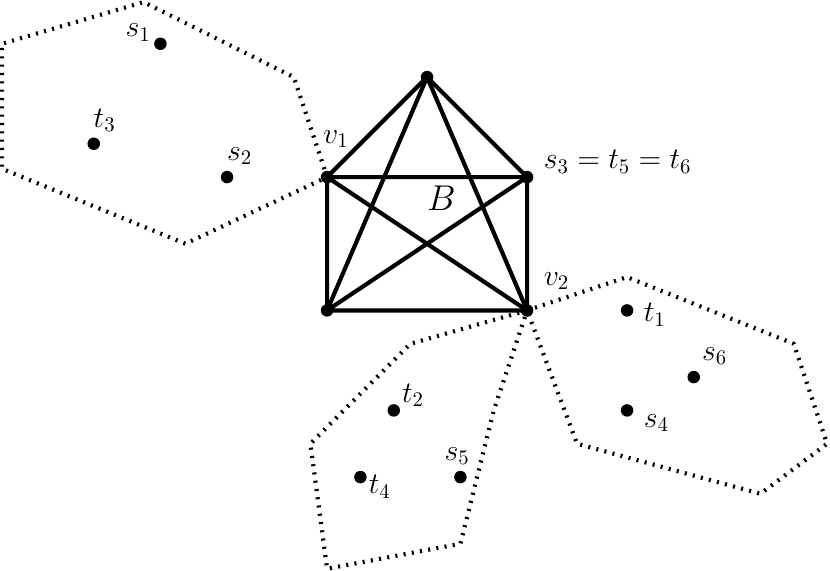}
    \caption{Here, $h(s_1) = h(s_2)=h(t_3)= v_1$. Similarly, $h(t_2) = h(t_4)=h(s_5)= h(t_1)= h(s_6) = h(s_4) =  v_2$. Also, $h(s_3) = s_3$, $h(t_5) = t_5$, and $h(t_6)=t_6$. Moreover, here $\mathcal{X}_B = \{ (h(s_1),h(t_1)), (h(s_2),h(t_2)), (h(s_3),h(t_3)), (h(s_5),h(t_5)), (h(s_6),h(t_6)) \}$. Note that $(h(s_4),h(t_4)) \notin \mathcal{X}_B$ since $h(s_4) = h(t_4) = v_2$. }
    \label{Fig:restrictBlock}
\end{figure}

We have the following trivial observation that we will use for our proofs.
\begin{observation}\label{O:BlockSubpath}
Let $\mathcal{P}$ be a set of edge-disjoint paths. Consider a set of paths $\mathcal{P}^*$ constructed in the following manner. For each path $P\in \mathcal{P}$, add a path $P^*$ to $\mathcal{P}^*$ such that $E(P^*) \subseteq E(P)$. Then, $\mathcal{P}^*$ is also a set of edge-disjoint paths.
\end{observation}

We have the following lemma.

\begin{lemma}\label{L:BlockRestrict}
Let $(G,\mathcal{X},k)$ be an instance of \EDP~on a block graph $G$, and let $B$ be a block of $G$. If $(G, \mathcal{X}, k)$ is a Yes-instance, then $(B,\mathcal{X}_B, |\mathcal{X}_B|)$ is a Yes-instance.
\end{lemma}
\begin{proof}
Let $\mathcal{P}$ be a minimum solution of $(G,\mathcal{X},k)$. Due to the definition of $(B,\mathcal{X}_B, |\mathcal{X}_B|)$ (Definition~\ref{D:virtual}) and Observation~\ref{O:BlockSubpath}, it is sufficient to show that, for each occurrence of a terminal pair $(s,t) \in \mathcal{X}$ such that $P$ is the path in $\mathcal{P}$ corresponding to it, if $h(s)\neq h(t)$, then $P$ contains an $(h(s),h(t))$-subpath, say, $P^*$. Recall that $U$ is the set of cut vertices of $B$. Since $h(s)\neq h(t)$, we have the following cases:
\medskip

    \noindent \textbf{Case 1:} $\bm{s,t \in V(B).}$ Note that in this case $h(s) = s$ and $h(t) = t$. Therefore, $P^* = P$.
    \smallskip
    
    \noindent \textbf{Case 2:} $\bm{s \in V(C_{B,u})}$ \textbf{and} $\bm{t \in V(C_{B,w})}$ \textbf{for distinct} $\bm{u,w \in U.}$ Note that in this case, every $(s,t)$-path has to pass through vertices $u$ and $w$ as each of them is a cut vertex for $s$ and $t$. Moreover, $h(s) = u$ and $h(t) = w$. Therefore, the subpath of $P$ with endpoints $u$ and $w$ is $P^*$.
    \smallskip
    
    \noindent \textbf{Case 3:} $\bm{s \in V(C_{B,u})}$ \textbf{and} $\bm{t \in V(B) \setminus \{ u\}}$ \textbf{for some} $\bm{u\in U.}$ Note that in this case, every $(s,t)$-path has to pass through the vertex $u$ as it is the cut vertex common to $C_{B,u}$ and $B$. Moreover, since $h(s) = u$ and $h(t) = t$, the subpath of $P$ with endpoints $u$ and $t$ is $P^*$.
    \smallskip
    
    \noindent \textbf{Case 4:} $\bm{s  \in V(B) \setminus \{ u\} }$ \textbf{and} $\bm{t \in V(C_{B,u})}$ \textbf{for some} $\bm{u\in U.}$ This case is symmetric to Case 3.

This establishes that $(B, \mathcal{X}_B, |\mathcal{X}_B|)$ is a Yes-instance, and, thus, completes our proof.
\end{proof}

Next, we will show that if for a block $B$, $(B,\mathcal{X}_B, |\mathcal{X}_B|)$ is a Yes-instance, then we can contract $B$ ``safely''. In particular, we have the following lemma.

\begin{lemma}\label{L:BlockUse}
Let $(G,\mathcal{X},k)$ be an instance of \EDP~on a block graph $G$, and let $B$ be a block of $G$ whose set of cut vertices is $U$. Moreover, let the contraction of $B$ in $(G,\mathcal{X},k)$ yield $(G',\mathcal{X'},k')$. Given that $(B, \mathcal{X}_B,|\mathcal{X}_B|)$ is a Yes-instance, $(G,\mathcal{X},k)$ is a Yes-instance if and only if $(G',\mathcal{X'},k')$ is a Yes-instance.
\end{lemma}
\begin{proof}

Let $\mathcal{P}_B$ be a solution of $(B, \mathcal{X}_B,|\mathcal{X}_B|)$. 

In one direction, let  $(G,\mathcal{X},k)$ be a Yes-instance, and let $\mathcal{P}$ be a solution of $(G,\mathcal{X},k)$. Then, for an occurrence of a terminal pair $(s,t) \in \mathcal{X}$ such that $f(s) \neq f(t)$, let $P$ be the $(s,t)$-path in $\mathcal{P}$ corresponding to it. We construct the path $P'$ in $G'$ with endpoints $f(s)$ and $f(t)$ in the following manner, and add it to $\mathcal{P'}$. If $P$ does not contain any vertex of $B$, then let $P' = P$.
Else, let $x,y \in V(B)$ be the first and the last vertices of $B$ to appear in $P$ while traversing $P$ from $s$ to $t$. Note that $x$ and $y$ might be the same vertex (when $s,t \in C_{B,u}$, for some $u\in U$). Then, we replace the $(x,y)$-subpath of $P$ with the vertex $v$ to get the path $P'$. Note that it is always possible since $N_G(U) \subseteq N_{G'}(v)$ (see Definition~\ref{D:contract}). 
Now, observe that the set of edges used in $P'$ other than the edges with $v$ as an endpoint is a subset of edges used in $P$.  Therefore, due to Observation~\ref{O:BlockSubpath}, if two paths, say, $P'_1$ and $P'_2$ in $\mathcal{P}'$ share an edge, say, $e$, then $v$ is an endpoint of $e$. Targeting contradiction, assume $P'_1$ and
$P'_2$ (obtained from $P_1$ and $P_2$ as discussed above) share edge $e = vw$. Now, note that $v$ corresponds to cut vertices, say, $x$ and $y$, in $V(P_1)$ and $V(P_2)$, respectively. Therefore, there is an edge $xw \in E(P_1)$ and edge $yw \in E(P_2)$. Clearly, if $x=y$, then $P_1$ and $P_2$ are not edge-disjoint, a contradiction. If $x\neq y$, then due to the definition of a block graph, $x$ and $y$ cannot be both simultaneously adjacent to a vertex $u \notin V(B)$, another contradiction. Thus, $\mathcal{P}'$ is a solution to \EDP~on $(G',\mathcal{X}',k')$.

In the other direction, let $(G',\mathcal{X}',k')$ be a Yes-instance, and let $\mathcal{P'}$ be a solution of $(G',\mathcal{X}',k')$. Then, for each occurrence of a terminal pair $(s,t)\in \mathcal{X}$, we provide a path $P$ in the following manner. First, if $f(s) \neq f(t)$ (respectively, $h(s)\neq h(t)$), then let $P' \in \mathcal{P}'$ (respectively, $P_B \in \mathcal{P}_B$) be the path associated with the occurrence of the terminal pair $(f(s), f(t)) \in \mathcal{X}'$ (respectively, $(h(s), h(t)) \in \mathcal{X}_B$) corresponding to $(s,t)$. Now, we have the following cases.
\medskip

   \noindent \textbf{Case 1:} $\bm{s,t \in V(B).}$ Then, let $P = P_B$. Note that  $(h(s),h(t)) \in \mathcal{X}_B$ since $h(s)\neq h(t)$.
    \smallskip
    
    \noindent \textbf{Case 2:}  $\bm{s,t \in C_{B,u}}$, \textbf{for some} $\bm{u\in U.}$ Here, note that any $(s,t)$-path cannot contain any vertex of $V(B)\setminus \{u\}$. So, let $P$ be the path $P'$ with the only change that if $v \in V(P')$, then we replace it with $u$. Also, note that $(f(s),f(t)) \in \mathcal{X}'$ since $f(s)\neq f(t)$.
  \smallskip
  
   \noindent \textbf{Case 3:} $\bm{s \in C_{B,u}}$ \textbf{and} $\bm{t \in C_{B,w}}$, \textbf{for distinct} $\bm{u,w \in U.}$ 
    In this case, observe that $v \in V(P')$. Let $P'$ be of the form $(s, \ldots, x,v,y, \ldots, t)$. We obtain $P$ by replacing $v$ with $P_B$, that is, $P = (s, \ldots, x, u, \ldots, w, y, \ldots, t)$.   
    \smallskip
    
    \noindent \textbf{Case 4:} $\bm{s \in C_{B,u}}$ \textbf{and} $\bm{t \in V(B)\setminus \{ u\}}$, \textbf{for some} $\bm{u \in U.}$ 
    Similarly to Case 3, we replace the vertex $v$ in path $P'$ with the path $P_B$ to obtain $P$. 
    \smallskip
    
    \noindent \textbf{Case 5:} $\bm{s \in V(B)\setminus \{ u\}}$ \textbf{and} $\bm{t  \in C_{B,u}}$, \textbf{for some} $\bm{u\in U.}$ This case is symmetric to Case 4.

\medskip

Finally, we obtain a path $P$ for every terminal pair $(s,t) \in \mathcal{X}$. Moreover, since for any path  $P'\in \mathcal{P}'$ and $P_B \in \mathcal{P}_B$, $E(P')\cap E(P_B) = \emptyset$, it is easy to see that the paths in $\mathcal{P}$ (where $\mathcal{P}$ is the set of paths obtained as discussed above in Cases 1-5) are edge-disjoint, and therefore $(G,\mathcal{X},k)$ is a Yes-instance.
\end{proof}


Next, we have the following reduction rule. 

\begin{RR}[RR\ref{RRB2}]\label{RRB2}
If $G$ has a block $B$ such that $|V(B)| > k$,  then contract $B$ in $(G,\mathcal{X},k)$ to get $(G',\mathcal{X}',k')$.
\end{RR}

The safeness of RR\ref{RRB2} is implied by Lemma~\ref{L:BlockUse} and Lemma~\ref{L:cliqueBetter}. In particular, we have the following corollary.

\begin{corollary}\label{C:BlockRule}
RR\ref{RRB2} is safe.
\end{corollary}

\begin{observation} 
RR\ref{RRB2} can be applied in polynomial
time.
\end{observation}

Finally, we have the following reduction rule.

\begin{RR}[RR\ref{RRB3}]\label{RRB3}
Let $B$ be a block of $G$ that has exactly two cut vertices, say, $u$ and $w$, and there is no terminal vertex in $V(B) \setminus \{u,w\}$. Consider the instance $(B, \mathcal{X}_B, |\mathcal{X}_B|)$ restricted to the block $B$. If $|V(B)| > |\mathcal{X}_B| $, then contract $B$ in $(G,\mathcal{X},k)$ to get the instance $(G',\mathcal{X}',k')$. Else, answer negatively.
\end{RR}

We have the following lemma to establish that RR\ref{RRB3} is safe.
\begin{lemma}\label{L:RRB3}
RR\ref{RRB3} is safe.
\end{lemma}
\begin{proof}
First, observe that $G'$ is a block graph. 
Lemma~\ref{L:BlockRestrict} implies that if $(B, \mathcal{X}_B, |\mathcal{X}_B|)$ is a No-instance, then $(G,\mathcal{X},k)$ is a No-instance. Moreover, Lemma~\ref{L:BlockUse} establishes that if $(B, \mathcal{X}_B, |\mathcal{X}_B|)$ is a Yes-instance, then we can contract $B$ in $(G,\mathcal{X},k)$. Therefore, to prove our lemma, it is sufficient to show that $(B,\mathcal{X}_B, |\mathcal{X}_B|)$ is a Yes-instance if and only if $|V(B)| > |\mathcal{X}_B| $. 
Let $r = |\mathcal{X}_B|$. Then, note that in $\mathcal{X}_B$, we have $r$ many terminal pairs, and each terminal pair has terminals $u$ and $w$.

In one direction, let $|V(B)| \leq r$. Then, note that the degree of both $u$ and $w$ in $B$ is at most $r-1$, and therefore, there cannot be $r$ edge-disjoint paths with terminals $u$ and $w$ in $B$. Hence, $(B_i,\mathcal{X}_B, |\mathcal{X}_B|)$ is a No-instance.

In the other direction, let $|V(B)| > r$. We provide $r$  edge-disjoint paths between $u$ and $w$ in the following manner. One path, say, $P_r$, consists only of the edge $uw$ (i.e., $P_r = P_{uw}$). To get $r-1$ additional edge-disjoint paths, we select $r-1$ vertices, say, $u_1, \ldots, u_{r-1}$, from $V(B)\setminus \{u,w\}$. Now, consider the path $P_i$, for $i\in [r-1]$, as $(u,u_i,w)$. Observe that paths $P_1, \ldots, P_r$ are indeed edge-disjoint. Hence, $(B_i,\mathcal{X}_B, |\mathcal{X}_B|)$ is a Yes-instance.
\end{proof}

\begin{observation}
RR\ref{RRB3} can be applied in polynomial
time.
\end{observation}

Next, we establish that once we cannot apply RR\ref{RRB1}-RR\ref{RRB3} anymore, the size of the reduced graph is bounded by a quadratic function of $k$. In particular, we have the following lemma.

\begin{lemma}\label{L:BlockSizeBound}
Let $(G, \mathcal{X}, k)$ be an instance of \EDP~where we cannot apply reduction rules RR\ref{RRB1}-RR\ref{RRB3} anymore. Then,  $G$ contains at most $4k-2$ blocks and at most $4k^2-2k$ vertices.
\end{lemma}
\begin{proof}
First, we prove that the number of blocks in $G$ is at most $4k-2$. Observe that each end block contains at least one terminal (due to RR\ref{RRB2}), and also, each block with two cut vertices contains at least one terminal (due to RR\ref{RRB3}).  Therefore, the number of blocks with at most two cut vertices is at most $2k$. Hence, due to Observation~\ref{P:block}, the number of blocks in $G$ with at least three cut vertices is at most $2k-2$. Thus, the total number of blocks in $G$ is at most $4k-2$.

Since each block contains at most $k$ vertices (due to RR\ref{RRB2}), the total number of vertices in $G$ is at most $4k^2-2k$.
\end{proof}
 
Note that throughout this section, our initial parameter $k$ does not increase during the application of reduction rules RR\ref{RRB1}-RR\ref{RRB3}. Moreover, observe that RR\ref{RRB1}-RR\ref{RRB3} can be implemented in polynomial time. Therefore, due to Lemmas~\ref{L:RRB1},~\ref{L:RRB3},~\ref{L:BlockSizeBound}, and Corollary~\ref{C:BlockRule}, we have the following theorem.

\block*

\subsection{A Linear Vertex Kernel for Clique Paths} \label{linear:edp:cliquepath}
A \textit{clique path} is a block graph where each block has at most two cut vertices. (In an informal manner, we can think that the blocks are arranged in the form of a path.) In this section, we present a linear vertex kernel for \EDP~on clique paths. First, we present an overview of the overall idea leading to this result.

\medskip
\noindent \textbf{Overview.} Let $(G,\mathcal{X},k)$ be an instance of \EDP~where $G$ is a clique path.  Our kernelization algorithm is based on the application of  RR\ref{RRB1}-RR\ref{RRB3} (defined in Section~\ref{quadratic:edp:block}) along with three new reduction rules (RR\ref{RRC3}-RR\ref{RRC2}). In our kernel for block graph in Section~\ref{quadratic:edp:block}, we established that for a block $B$, if $|V(B)| \geq k+1$, then we can contract $B$ (see RR\ref{RRB2}). Moreover, we showed that the total number of blocks in a reduced instance can be at most $4k-2$, thus giving us an $\mathcal{O}(k^2)$ vertex kernel.

Here, we use the property of clique paths that each block can have at most two cut vertices to improve the kernel size. Since there is no block with more than two cut vertices, each block must contain a terminal after an exhaustive application of RR\ref{RRB1}-RR\ref{RRB3}. Let $B$ be a block of $G$ with cut vertices $u$ and $w$. Consider the instance $(B,\mathcal{X}_B,|\mathcal{X}_B|)$, that is, the instance $(G,\mathcal{X},k)$ restricted to the block $B$ (see Definition~\ref{D:virtual}). Any terminal pair $(s,t)\in \mathcal{X}_B$ is of one of the following types:

\begin{itemize}
\item Type-A: $s,t \in \{u,w\}$. 
\item Type-B: $s = u$ and $t \in V(B) \setminus \{u,w\}$, or vice versa. 
\item Type-C: $s = w$ and $t \in V(B) \setminus \{u,w\}$, or vice versa. 

\item Type-D: $s,t \in V(B) \setminus \{u,w\}$. 
\end{itemize}

Let $a,b,c$, and $d$ denote the cardinality of Type-A, Type-B, Type-C, and Type-D occurrences of terminal pairs in $\mathcal{X}_B$, respectively.  We show that if $|V(B)| > d+ 2 + \max\{b+c-1,0\}$, then we can either contract $B$ to a single vertex ``safely'' (when $|V(B)| > \max \{a+b,a+c\}$) or report a No-instance.
Summing these numbers over all blocks yields an upper bound on the total size of the reduced instance.
The Type-A pairs are irrelevant now,
each pair of Type-B or Type-C contributes to two blocks, while
each Type-D pair appears in only a single block.
By grouping the summands in an appropriate way, we are able to eventually attain a~bound of $2k+1$.

We have the following reduction rules.
\begin{RR}[RR\ref{RRC3}]\label{RRC3}
Let $(G,\mathcal{X},k)$ be an instance of \EDP~where $G$ is a clique path. Moreover, let  $B$ be a block of $G$ such that $B$ has two cut vertices, say, $u$ and $w$. Consider the instance $(B,\mathcal{X}_B,|\mathcal{X}_B|)$. If $ |V(B)| \leq \max \{a+b, a+c \}$, then report a No-instance.
\end{RR}

We have the following lemma to prove that RR\ref{RRC3} is safe.
\begin{lemma}\label{L:RRC3}
RR\ref{RRC3} is safe.
\end{lemma}
\begin{proof}
Due to Lemma~\ref{L:BlockRestrict}, we know that if $(B,\mathcal{X}_B,|\mathcal{X}_B|)$ is a No-instance, then $(G,\mathcal{X},k)$ is a No-instance. Therefore, it is sufficient to show that if $|V(B)| \leq \max \{ a+b, a+c\}$, then $(B,\mathcal{X}_B,|\mathcal{X}_B|)$ is a No-instance. Without loss of generality, assume that $b \geq c$. Observe that the vertex $u$  appears in $a+b$ occurrences of terminal pairs (in $\mathcal{X_B}$). If $|V(B)| \leq a+b$, then $d_B(u) \leq a+b-1$, and therefore, $u$ cannot be a part of $a+b$ edge-disjoint paths. Thus, $(B,\mathcal{X}_B,|\mathcal{X}_B|)$ is a No-instance, and so is $(G,\mathcal{X},k)$.
\end{proof}

Next, we have the following reduction rule. 
\begin{RR}[RR\ref{RRC1}]\label{RRC1}
Let $(G,\mathcal{X},k)$ be an instance of \EDP~where $G$ is a clique path. Moreover, let  $B$ be a block of $G$ that has two cut vertices, say, $u$ and $w$. Consider the instance $(B,\mathcal{X}_B,|\mathcal{X}_B|)$. If $|V(B)| > d+2+\max \{ b+c-1,0\}$, then contract $B$ in $(G,\mathcal{X},k)$ to obtain the instance $(G',\mathcal{X}',k')$.
\end{RR}

We have the following lemma to establish that RR\ref{RRC1} is safe.

\begin{lemma}\label{L:CliquePath}
RR\ref{RRC1} is safe.
\end{lemma}
\begin{proof}
We first remark that $G'$ is also a clique path. Next, we show that if $|V(B)| > d+2+\max \{ b+c-1,0\}$, then $(B,\mathcal{X}_B, |\mathcal{X}_B|)$ is a Yes-instance, and therefore, due to Lemma~\ref{L:BlockUse}, we can safely contract $B$ in $(G,\mathcal{X},k)$ to get $(G',\mathcal{X}',k')$. 

If $a=0$, then the claim follows directly due to Lemma~\ref{L:cliqueBetter}. Hence, we assume that $a>0$. Due to RR\ref{RRC3},
we know that $|V(B)| > \max \{a+b,a+c\}$. To prove that $(B,\mathcal{X}_B, |\mathcal{X}_B|)$ is a Yes-instance, we consider the following two cases. 

\noindent \textbf{Case 1:}    $\bm{b+c = 0.}$ Let $Y$ denote the set of vertices in $V(B) \setminus \{u,w \}$. Note that $|Y| \geq d+1$ (since $|V(B)| \geq d+3$). Moreover, recall that $|V(B)| \geq a+1$ (since $\max \{b,c\} =0$ and RR~\ref{RRC3} is not applicable) and hence, $|Y| \geq a-1$. Now, we assign edge-disjoint paths to the terminal pairs in the following manner. 
  \begin{enumerate}
       \item  First, we assign edge-disjoint paths to each occurrence of Type-A terminal pairs. Note that there are $a$ many such occurrences of terminal pairs, and each pair has terminals $u$ and $w$. Let $P^A_a = P_{uw}$ (i.e., $P^A_a$ consists only of a single edge $uw$). Next, let $y_1,\ldots, y_{|Y|}$ be an ordering of vertices in $Y$. Recall that $|Y|\geq a-1$. Now, for $1\leq i \leq a-1$, consider the path $P^A_{i} = (u,y_i,w)$. Observe that $P^A_1, \ldots, P^A_a$ are indeed edge-disjoint. 
        \smallskip

       \item Finally, we assign edge-disjoint paths to each occurrence of Type-D terminal pairs. Note that till now, no path has used any edge $xy$ such that $\{x,y\} \subseteq Y$. Moreover, each terminal of a Type-D terminal pair belongs to $Y$, the number of Type-D terminal pairs is $d$, $Y$ induces a clique, and $|Y| \geq d+1$. Therefore, due to Lemma~\ref{L:cliqueBetter}, all Type-D terminal pairs can be assigned paths only using the edges with both endpoints in $Y$. Since we have not used any such edge in a path of the form $P^A_j$ (for $j\in [a]$), $(B,\mathcal{X}_B, |\mathcal{X}_B|)$ is a Yes-instance.
    \end{enumerate}
    
    Therefore, $(B,\mathcal{X}_B, |\mathcal{X}_B|)$ is a Yes-instance.

\medskip
\noindent \textbf{Case 2:}   $\bm{b+c >0.}$  Without loss of generality, we assume that $b\geq c$. Since $|V(B)| > a+b$ (and hence, $|V(B)\setminus \{u,w\}| \geq a+b-1$), let the vertex set $V(B)\setminus \{u,w\}$  be partitioned into two subsets $X$ and $Y$ such that (1) $|X| = b$ and $X$ contains at least one terminal vertex of Type-B, (2) $|Y| \geq a-1$. Also, note that by the assumption, we have $|V(B)|\geq b+c+d+2$. Now, we assign edge-disjoint paths to terminal pairs in the following manner.

    
    First, we assign edge-disjoint paths to each occurrence of Type-A terminal pairs. Recall that each Type-A terminal pair has terminals $u$ and $w$. Let $P^A_a = P_{uw}$ (i.e., $P^A_a$ consists only of a single edge $uw$). Next, let $y_1,\ldots, y_{|Y|}$ be an ordering of vertices in $Y$. Recall that $|Y|\geq a-1$. Now, for $1\leq i \leq a-1$, consider the path $P^A_{i} = (u,y_i,w)$. Observe that $P^A_1, \ldots, P^A_a$ are indeed edge-disjoint.
    
    Next, we assign a path to one occurrence of a Type-B terminal pair. Recall that $X$ contains at least one terminal vertex of a Type-B terminal pair, say $\chi \in \mathcal{X_B}$, and let $x_1\in X$ be that vertex. Let $P^B_{1} = P_{ux_1}$ (i.e., $P^B_1$ consists only of a single edge $ux_1$). Note that $P^B_1$ is indeed edge-disjoint with every path in $\{P^A_1,\ldots, P^A_a\}$ and till this point, we have not used any edge $xy$ such that $\{x,y\} \subseteq X\cup Y$. 
    
    Now, we assign the rest of the paths in the following manner. Consider the graph $H$ obtained after removing all the edges used by the paths $P^A_1,\ldots, P^A_{a}$ and $P^B_1$. Notice that $H$ is a split graph where $X\cup Y$ induces a clique and $\{u,w\}$ induces an independent set.  Moreover, observe that $d_H(u)\geq b-1$ and $d_H(w) \geq c$. Now, let $\mathcal{X}'_B$ be the multiset of terminal pairs obtained after removing all Type-A terminal pairs along with the terminal pair $\chi$ from $\mathcal{X}_B$. Note that these are exactly the terminal pairs in $\mathcal{X}_B$ that are not provided a path yet. Since $|\mathcal{X}'_B| \leq b+c+d-1$ and $|X\cup Y| \geq b+c+d$, due to \cref{prop2}, we have that $(H,\mathcal{X}'_B,|\mathcal{X}'_B|)$ is a Yes-instance, and thus, there exists an edge-disjoint path between each terminal pair $\mathcal{X}'_B$. Since these paths only use edges present in $H$, notice that all these paths are indeed edge-disjoint from paths $P^A_1, \ldots, P^A_{a}$ and $P^B_1$ (due to construction of $H$). Hence $(B,\mathcal{X}_B, |\mathcal{X}_B|)$ is a Yes-instance. This completes our proof.
\end{proof}

\begin{observation}
RR\ref{RRC1} can be applied in polynomial
time.
\end{observation}

In RR\ref{RRC1}, we showed that in a reduced instance, the size of each block with two cut vertices is bounded by a linear function of the number of times its vertices appear in some occurrences of terminal pairs. In the next reduction rule (RR\ref{RRC2}), we consider the end blocks (i.e., blocks with exactly one cut vertex). So, consider an end block $B$ of $G$ with cut vertex $u$, and let $(B,\mathcal{X}_B,|\mathcal{X}_B|)$ be the restriction of $(G,\mathcal{X},k)$ to $B$. Any terminal pair $(s,t)\in \mathcal{X}_B$ is one of the following types:
\begin{itemize}
    \item Type-B$'$: $s=u$ and $t\in V(B)\setminus \{u\}$, or vice versa.
    \item Type-D$'$: $s,t \in V(B)\setminus \{ u\}$.
\end{itemize}

Let $b'$ and $d'$ denote the number of occurrences of Type-B$'$ and Type-D$'$ terminal pairs in $\mathcal{X}_B$, respectively.

We have the following reduction rule.
\begin{RR}[RR\ref{RRC2}] \label{RRC2}
Let $(G,\mathcal{X},k)$ be an instance of \EDP~where $G$ is a clique path. Let $B$ be an end block of $G$ with cut vertex $u$. Consider the instance $(B,\mathcal{X}_B,|\mathcal{X}_B|)$. If $|V(B)| > b'+d'$, then contract $B$ in $(G,\mathcal{X},k)$ to obtain $(G',\mathcal{X}',k')$.
\end{RR}

We have the following lemma to prove that RR\ref{RRC2} is safe. 

\begin{lemma}\label{L:RRC2}
RR\ref{RRC2} is safe.
\end{lemma}
\begin{proof}
We first remark that $G'$ is a clique path. Since $B$ is a clique with at least $|\mathcal{X}_B|+1$ (here $|\mathcal{X}_B| = b'+d'$) vertices, due to Lemma~\ref{L:cliqueBetter}, we have that $(B,\mathcal{X}_B,|\mathcal{X}_B|)$ is a Yes-instance. Therefore, due to Lemma~\ref{L:BlockUse}, we can safely contract $B$ in $(G,\mathcal{X},k)$. \end{proof}



Next, we establish that once we cannot apply RR\ref{RRB1}-RR\ref{RRC2}, the number of vertices of the reduced graph is bounded by a linear function of $k$. In particular, we have the following lemma.

\begin{lemma}\label{L:CliqueSizeBound}
Let $(G,\mathcal{X},k)$ be an instance of \EDP~where $G$ is a clique path. If we cannot apply RR\ref{RRB1}-RR\ref{RRC2} on this instance, then $G$ contains at most $2k+1$ vertices.
\end{lemma}
\begin{proof}
Let $B_1,\ldots, B_t$ be the blocks in $G$. We say that a block $B$ is a $D$-\textit{block} if the instance $(B,\mathcal{X}_B, |\mathcal{X}_B|)$ contains neither Type-B nor Type-C terminal pairs (i.e., $b+c = 0$), and otherwise, we say that $B$ is a $C$-\textit{block}. Let $\mathcal{D}$ and $\mathcal{C}$ denote the set of $D$-blocks and $C$-blocks of $G$. Note that $|V(G)| = \sum_{B\in \mathcal{D}} |V(B)| + \sum_{B\in \mathcal{C}} |V(B)| - (t-1)$ because each cut vertex is counted in two blocks and there are $t-1$ cut vertices.  

For the ease of presentation of calculations, for a block $B$, let $B^b,B^c,B^d$ denote the cardinality of Type-B, Type-C, and Type-D terminal pairs in $(B,\mathcal{X}_B, |\mathcal{X}_B|)$, respectively. Due to RR~\ref{RRC1} and RR~\ref{RRC2}, each block $B\in \mathcal{D}$ has cardinality at most $B^d+2$ and each block $B\in \mathcal{C}$ has cardinality at most $B^d+B^b+B^c+1$. Therefore, $$ |V(G)| \leq \sum_{B\in \mathcal{D}}(B^d+2) +\sum_{B\in\mathcal{C}} (B^b+B^c+B^d+1) - t+1,$$

$$ |V(G)| \leq \sum_{B\in \mathcal{D}}(B^d+1) + |\mathcal{D}| +\sum_{B\in\mathcal{C}} (B^b+B^c+B^d+1) - t+1,$$

$$ |V(G)| \leq |\mathcal{D}| -t+1  +\sum_{B \in \mathcal{C}\cup \mathcal{D}} (B^b+B^c+B^d+1),$$

$$ |V(G)| \leq |\mathcal{D}| -t+1 +t +\sum_{B\in \mathcal{C}\cup \mathcal{D}} (B^b+B^c+B^d).$$

Now, let $k_1$ be the number of terminal pairs in $G$ such that both terminals of the terminal pair lie in the same block of $G$, possibly on the cut vertices of the block. Then, observe that $\sum_{B\in \mathcal{C}\cup \mathcal{D}} B^d \leq k_1$. Thus, the total number of terminal pairs such that the endpoints lie in different blocks are $k-k_1$. Now, consider a terminal pair whose endpoints lie in different blocks. Then, this terminal pair will behave as a Type-B terminal pair for at most one unique block in $G$ and as a Type-C terminal pair for at most one another unique block in $G$. Thus, $\sum_{B\in \mathcal{C}\cup \mathcal{D}} (B^b+B^c) \leq 2(k-k_1)$. Therefore, $|V(G)| \leq |\mathcal{D}|+1 +k_1+2(k-k_1) = |\mathcal{D}|+1+2k-k_1$. Finally, observe that $|\mathcal{D}|\leq k_1$ (since each block contains at least one terminal due to RR~\ref{RRB3}). Therefore, we have that $|V(G)| \leq 2k+1$.
\end{proof}

Now, we have the following observation.
\begin{observation}\label{O:CliquePath}
RR\ref{RRC3}-RR\ref{RRC2} can be applied in polynomial time. Moreover, during the application of RR\ref{RRC3}-RR\ref{RRC2}, we never increase the initial $k$.
\end{observation}

Finally, due to RR\ref{RRB1}-RR\ref{RRC2}, along with Lemmas~\ref{L:RRB1},~\ref{L:RRB3},~\ref{L:CliquePath},~\ref{L:RRC2}, and~\ref{L:CliqueSizeBound}, Corollary~\ref{C:BlockRule}, and Observation~\ref{O:CliquePath}, we have the following theorem.

\cliquepath*

\section{\textsf{NP}-hardness for Complete Graphs} \label{npc:complete}
In this section, we prove that \textsc{EDP} is $\mathsf{NP}$-hard on complete graphs by giving a polynomial-time reduction from \textsc{EDP} on general graphs, which is known to be $\mathsf{NP}$-hard~\cite{kramer}. Our reduction is based on the standard technique of adding missing edges and placing a terminal pair on the endpoints of the added edge.  This technique was also used to prove the \textsf{NP}-hardness of \textsc{EDP} for split graphs~\cite{hegger}.

\complete*
 \begin{proof}
Let $(G,\mathcal{X},k)$ be an instance of \textsc{EDP}, where $\mathcal{X}=\{(s_1, t_1),\ldots,$ $(s_k, t_k)\}$. Define the graph $G'$ as follows: Let $V(G')=V(G)$ and $E(G')=E(G)\cup \{uv: uv\notin E(G)\}$.  Furthermore, let $\mathcal{T} = \{(s_{uv}, t_{uv}) : uv \in E(G')\setminus E(G), s_{uv}=u, t_{uv}=v\}$. Note that for ease of notation, we denote the terminal pairs in $\mathcal{T}$ by $(s_{uv}, t_{uv})$ rather than $(u,v)$. Also, for an edge $uv\in E(G')\setminus E(G)$, we introduce either $(s_{uv}, t_{uv})$ or $(s_{vu}, t_{vu})$ in $\mathcal{T}$, not both. Let $\mathcal{X'} = \mathcal{X} \cup \mathcal{T}$.  Now, we claim that $(G,\mathcal{X},k)$ and $(G',\mathcal{X'},k')$ are equivalent instances of \textsc{EDP}, where $k'=k+|\mathcal{T}|$. Let $\mathcal{P}_{\mathcal{T}}=\{P_{uv}: (s_{uv},t_{uv})\in \mathcal{T}\}$.

In one direction, let $\mathcal{P}=\{P_{1},\ldots,P_{k}\}$ be a solution of $(G,\mathcal{X},k)$. Since for every $(s_{uv},t_{uv})\in \mathcal{T}$, the edge $uv$ does not belong to any path in $\mathcal{P}$, $\mathcal{P}\cup \mathcal{P}_{\mathcal{T}}$ is a set of edge-disjoint paths in $G'$. As $\mathcal{X'} = \mathcal{X} \cup \mathcal{T}$, $\mathcal{P}\cup \mathcal{P}_{\mathcal{T}}$ is a solution of $(G',\mathcal{X'},k')$. 

In the other direction, let $\mathcal{P'}$ be a solution of $(G',\mathcal{X'},k')$ that contains as many paths from $\mathcal{P}_{\mathcal{T}}$ as possible. Next, we claim that $\mathcal{P}_{\mathcal{T}}\subseteq \mathcal{P'}.$ Targeting a contradiction, suppose that there
exists a terminal pair, say, $(s_{uv},t_{uv}) \in \mathcal{T}$, such that $P_{uv}\notin \mathcal{P'}$. Let $Q$ denote the path in $\mathcal{P'}$
connecting $s_{uv}$ and $t_{uv}$. If none of the paths in $\mathcal{P'}$ uses the edge $uv$,
then the set $(\mathcal{P'} \setminus Q)\cup \{P_{uv}\}$ is a solution of $(G',\mathcal{X'},k')$ containing more paths from $\mathcal{P}_{\mathcal{T}}$ than $\mathcal{P'}$, contradicting the choice of $\mathcal{P'}$. Hence, there must be a unique path $P^{*}\in \mathcal{P'}$ that uses
the edge $uv$. Let $s^{*}$ and $t^{*}$ be the two terminals that are connected by the path $P^{*}$. Let $W$ denote the walk between $s^{*}$ and $t^{*}$ obtained from  $P^{*}$ by replacing the edge $uv$ with the path $Q$ (there may be some vertices that are repeated in $W$). Since $E(W)=(E(P^{*})\cup E(Q))\setminus \{uv\}$, $W$ is edge-disjoint from every path in $\mathcal{P'}\setminus\{P^{*},Q\}$ (as $\mathcal{P'}$ is a set of edge-disjoint paths). Let $Q^*$ be a path between $s^*$ and $t^*$ that uses a subset of edges of $W$, which again is edge-disjoint from every path in $\mathcal{P'}\setminus\{P^{*},Q\}$.
Hence, $\mathcal{P}=(\mathcal{P'}\setminus\{P^{*},Q\}) \cup \{{P_{uv},Q^{*}\}}$ is a solution of $(G',\mathcal{X'},k')$ that contains one more path from $\mathcal{P}_{\mathcal{T}}$ than $\mathcal{P'}$. This contradicts the choice of $\mathcal{P'}$, and thus implies that $\mathcal{P}_{\mathcal{T}}\subseteq \mathcal{P'}.$ Since $\mathcal{X'} = \mathcal{X} \cup \mathcal{T}$ and $\mathcal{P}_{\mathcal{T}}$ contains the edge-disjoint paths between the terminal pairs present in $\mathcal{T}$, $\mathcal{P'}\setminus \mathcal{P}_{\mathcal{T}}$ must contain the edge-disjoint paths between the terminal pairs present in $\mathcal{X}$. Thus, $\mathcal{P'}\setminus \mathcal{P}_{\mathcal{T}}$ is a solution of $(G,\mathcal{X},k)$.
\end{proof}

\section{Kernelization Results on \textsc{VDP}}\label{sec:vdp}

\subsection{A Linear Vertex Kernel for Split Graphs} \label{linearvdpsplit}

Let $(G,\mathcal{X},k)$ be an instance of \textsc{VDP} where $G$ is a split graph. Note that given a split graph $G$, we can compute (in linear time) a partition $(C, I)$ of $V(G)$ such that $C$ is a clique and $I$ is an independent set~\cite{hammer}. We partition the set $I$ into two sets, say, $I_T$ and $I_N$, where $I_T$ denotes the set of terminal vertices in $I$ and $I_N$ denotes the set of non-terminal vertices in $I$. Observe that a terminal pair in a split graph can only be of one of the following types:

\begin{itemize}
\item Type-I: One of the terminal vertices belongs to $C$, and the other belongs to $I$.
\item Type-II: Both terminal vertices belong to $C$.
\item Type-III: Both terminal vertices belong to $I$.
\end{itemize}

 Our kernelization algorithm is based on a preprocessing step and the application of three reduction rules (RR\ref{RR1}-RR\ref{RR3}). Before proceeding further, let us first discuss the idea behind the proof of Theorem \ref{labelsplitvdp} in the following overview. \medskip

\noindent \textbf{Overview.}
Given an instance $(G,\mathcal{X},k)$ of \textsc{VDP} where $G$ is a split graph, our main goal (in this section) is to convert the given instance of \textsc{VDP} to an instance of \textsc{VDP-Unique} (see \cref{sec:contri}) on split graphs in polynomial time. For this purpose, we focus on the vertices that belong to more than one terminal pair (including different occurrences of the same pair). The intuition suggests that one should make copies of the vertex that belongs to multiple terminal pairs. However, in certain cases (in particular for heavy terminal pairs), this approach fails. In simple words, if $(s,t)$ is a heavy terminal pair (note that $st\in E(G)$, in this case), then creating a copy of either $s$ or $t$ creates an illusion that many 1-length paths exist between $s$ and $t$, which is not true (see Figure \ref{fig1}). 

For Type-I heavy terminal pairs, deleting one copy of $(s,t)$ from $\mathcal{X}$ along with the deletion of edge $st$ from $G$ works. However, for Type-II heavy terminal pairs, the situation is more subtle than it appears, as a similar operation cannot be applied here (as the graph no longer remains split after deleting edge $st$, in this case). Therefore, the idea is to tackle these ``problematic'' heavy terminal pairs with the help of some reduction rules (RR\ref{RR1} and RR\ref{RR2} in our case).

  \begin{figure}[t]
 \centering
    \includegraphics[scale=0.95]{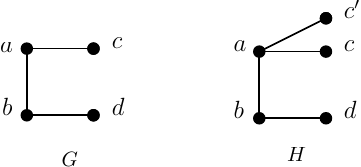}
    \caption{Let $(G,\mathcal{X},2)$ be an instance of \textsc{VDP} where $G$ is a split graph with partition $(C,I)$ such that $C=\{a,b\}$ and $I=\{c,d\}$. Let $\mathcal{X}=\{(a,c),(a,c)\}$. Furthermore, let $H$ be the graph obtained from $G$ by making a copy $c'$ of $c$, and let $\mathcal{X'}=\{(a,c),(a,c')\}$. Note that $(H,\mathcal{X'},2)$ is a Yes-instance (of \textsc{VDP}) while $(G,\mathcal{X},2)$ is a No-instance (of \textsc{VDP}).}
    \label{fig1}
\end{figure}

Another interesting thing to note here is that since we are dealing with a restricted graph class (i.e. split graphs), some vertices are more crucial than others. To turn this property in our favor, we carry out a preprocessing step, which we call \textsc{Clean-Up}. The intuition behind applying the \textsc{Clean-Up} operation is the following. Note that the vertices in the set $I_{N}$ can serve only one purpose, and that is to connect any two vertices from the clique. However, informally speaking, as the non-terminal vertices present in $C$ serve more purpose than the vertices present in $I_{N}$, it is important to save them until their use becomes necessary.  So, to tackle Type-II heavy terminal pairs (one of the ``problematic'' terminal pairs as discussed in the previous paragraph), it is desirable for us to use as many vertices from the set $I_{N}$ as possible because those terminal pairs that do not have a path of length 2 via a vertex from $I_{N}$ have to necessarily use a vertex from the clique as an internal vertex.  
 \medskip
 
 Now, let us define everything discussed above in a formal manner. For this purpose, we begin by defining the \textsc{Clean-Up} operation. Note that the \textsc{Clean-Up} operation is a preprocessing step of our kernelization algorithm. Therefore, given an instance $(G,\mathcal{X},k)$ of \textsc{VDP} where $G$ is a split graph, we apply this operation before applying any of the reduction rules RR\ref{RR1}-RR\ref{RR3} (discussed later in this section). 
\medskip

\noindent \textsc{\textbf{Clean-Up.}}  First, consider the following construction (Construction $\mathcal{A}$), which is crucial to define the \textsc{Clean-Up} operation. This construction was also used by Heggernes et al.~\cite{hegger} who used it to remove a subset of $I_N$ (safely). However, in the \textsc{Clean-Up} operation, we will remove the entire set $I_N$ (safely). In order to do so, we need a few more technical arguments than the ones present in~\cite{hegger}. 

\medskip

\noindent \textbf{Construction $\mathcal{A}$:} 
Given an instance $(G,\mathcal{X}, k)$ of \textsc{VDP} where $G$ is a split graph, we construct a bipartite graph, say, $H$, with bipartition $(A,B)$ as follows:  For every terminal pair $(s,t)$ of Type-II such that $st$ is a heavy edge with weight $w \geq 2$, we introduce $w-1$ vertices, say, $v_{st}^{1},\ldots,v_{st}^{w-1}$, to $A$. The set $B$ consists of all the vertices from set $I_{N}.$ For each $v\in I_{N}$, introduce an edge from $v$ to vertices $v_{st}^{1},\ldots,v_{st}^{w-1}$ if and only if $v$ is adjacent to both $s$ and $t$ in $G$. See Figure \ref{bpgfig} for an illustration of the construction of $H$ from $(G,\mathcal{X},k)$.

Due to Proposition \ref{propmatching}, we compute a maximum matching, say, $M$, in $H$ in polynomial time. Let $\widehat{\mathcal{X}}\subseteq \mathcal{X}$ be the multiset of all terminal pairs whose corresponding vertices in $H$ are saturated by $M$. For example, in Figure \ref{bpgfig}, if $M=\{v_{ab}^{1}f,v_{ab}^{2}g,v_{de}^{1}h\}$ is a maximum matching in $H$, then $\mathcal{\widehat{X}}=\{(a,b),(a,b),(d,e)\}$. This ends Construction $\mathcal{A}$.

\begin{remark} \label{rmrk6}
For intuition, we remark that in Construction $\mathcal{A}$, we introduce only $w-1$ copies of a heavy edge of weight $w$ because one copy will be ``taken care'' of by the edge between them (see Observation \ref{obsminimum}).
\end{remark}

Next, consider the following definition.

\begin{definition} [\textsc{Clean-Up}] \label{CU}
 Given an instance $(G,\mathcal{X},k)$ of \textsc{VDP} where $G$ is a split graph, construct the bipartite graph $H$ and find a maximum matching $M$ in $H$, as described in Construction $\mathcal{A}$. If $I_{N}\neq \emptyset$ or $\widehat{\mathcal{X}}\neq \emptyset$, then $G' \Leftarrow G-I_{N}, \mathcal{X'}\Leftarrow \mathcal{X}\setminus \widehat{\mathcal{X}},$ $k' \Leftarrow k-|\widehat{\mathcal{X}}|$. 
\end{definition}

 \begin{figure}[t]
 \centering
    \includegraphics[scale=0.75]{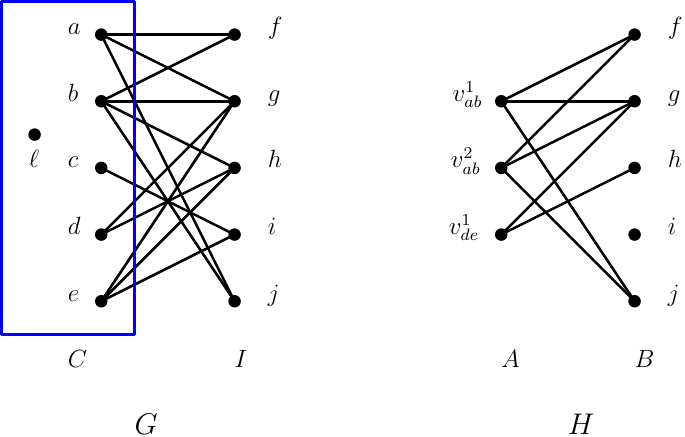}
    \caption{Let $(G,\mathcal{X},k)$ be an instance of \textsc{VDP} where $G$ is a split graph with a partition $(C,I)$ and $\mathcal{X}=\{(a,b),(a,b),$ $(a,b),(d,e),(d,e),(c,d),(b,d)\}$.  The vertices in $C$ form a clique; to keep the picture clean, we do not show the edges in $C$. Observe that the terminal pairs $(a,b)$ and $(d,e)$ are heavy.}
    \label{bpgfig}
\end{figure}

\begin{observation} \label{cleanupobs}
 \textsc{Clean-Up} can be done in polynomial
time.
\end{observation}

Note that due to Remark \ref{obsdel}, $G'$ is a split graph. Now, with the help of Definition \ref{Mtry}, Observation \ref{new:obs}, Proposition \ref{firstclaim}, and Lemma \ref{secondclaim}, we will establish that $(G,\mathcal{X},k)$ and $(G',\mathcal{X'},k')$, as described in Definition \ref{CU}, are equivalent (see Lemmas \ref{newlemma1}-\ref{newlemma}).

Before that, consider the next observation which follows trivially from Proposition \ref{prop1}.

\begin{observation} \label{newobs}
Let $(G,\mathcal{X},k)$ be a Yes-instance of \textsc{VDP} where $G$ is a split graph. Moreover, let $\mathcal{P}$ be a minimum solution of $(G,\mathcal{X},k)$. If there exists a path $P \in \mathcal{P}$ that
visits a vertex $x\in I_{N}$, then $P$ must be of the form $(u',x,v')$, where $u',v' \in C$ are terminal vertices such that $P_{u'v'}\in \mathcal{P}$.
\end{observation}

\begin{definition} \label{Mtry}
Let $(G,\mathcal{X},k)$ be a Yes-instance of \textsc{VDP} where $G$ is a split graph. Moreover, let $M$ and $H$ be as described
in Construction $\mathcal{A}$. Let $\mathcal{P}$ be a minimum solution of $(G,\mathcal{X},k)$. Then, $M'\subseteq E(H)$ is said to be \emph{induced} by $\mathcal{P}$ in $H$ if it is constructed as follows:

\begin{enumerate}
    \item Initialize $M'\Leftarrow \emptyset$.
    \item For every path $P \in \mathcal{P}$ that
visits a vertex $x\in I_{N}$: By Observation \ref{newobs}, $P$ must be of the form $(u',x,v')$, where $u',v' \in C$ are terminal vertices such that $P_{u'v'}\in \mathcal{P}$. This further implies that $u'v'$ is a heavy edge.
Let $w \geq 2$ be the weight of $u'v'$, and consider the vertices $v_{u'v'}^{1}, \ldots, v_{u'v'}^{w-1}$ in graph $H$. By Construction $\mathcal{A}$, $x$ is adjacent to every vertex in the set $\{v_{u'v'}^{1}, \ldots, v_{u'v'}^{w-1}\}$. So, we can arbitrarily choose an edge $xv_{u'v'}^{j}$ of $H$, for some $j\in [w-1]$ such that $xv_{u',v'}^{j}$ does not appear in $M'$\footnote{The existence of such a $j$ follows from the choice of $w$, and since $P_{u'v'}\in \mathcal{P}$}, and add it to $M'$.

\end{enumerate}
\end{definition}
\begin{observation} \label{new:obs}
Let $(G,\mathcal{X},k)$ be a Yes-instance of \textsc{VDP} where $G$ is a split graph. Moreover, let $M$, $A$, $B$, and $H$ be as described
in Construction $\mathcal{A}$. Let $\mathcal{P}$ be a minimum solution of $(G,\mathcal{X},k)$. Let $M'$ be induced by $\mathcal{P}$ in $H$ as described in Definition \ref{Mtry}. Then, $M'$ is a matching.
\end{observation}
\begin{proof}
 Note that for every path in $\mathcal{P}$ that has an internal vertex that belongs to $I_{N}$, we choose a distinct vertex in set $A$ to be an endpoint of an edge in $M'$ (due to Definition \ref{Mtry}). Furthermore, since the paths in $\mathcal{P}$ are internally vertex-disjoint, for every edge in $M'$, its endpoint in set $B$ is also unique. Thus, $M'$ is a matching in $H$. \end{proof}

For an illustrative example, consider $G$, $H$, and $\mathcal{X}$ given in Figure \ref{bpgfig}. Let $\mathcal{P}=\{(a,f,b),(a,\ell,b),$ $ P_{ab},$ $(d,g,e), P_{de}$ $, P_{cd}, P_{bd}\}$ be a solution of $(G,\mathcal{X},7)$. Then,  $\{v_{ab}^{1}f,v_{de}^{1}g\}$ and $\{v_{ab}^{2}f,v_{de}^{1}g\}$ are two possible choices for $M'$.
\begin{proposition} [\cite{hegger}]\label{firstclaim}
Let $(G,\mathcal{X},k)$ be a Yes-instance of \textsc{VDP} where $G$ is a split graph. Moreover, let $M$, $B$, and $H$ be as described
in Construction $\mathcal{A}$. Let $R$ $(\subseteq B)$ be the set of vertices in $I_{N}$ that are not saturated by $M$ in graph $H$. Let $\mathcal{P}$ be a minimum solution of $(G,\mathcal{X},k)$ for which the number of vertices in $R$ visited by the paths in $\mathcal{P}$ is minimum. Then, none of the paths in $\mathcal{P}$ visits a vertex in $R$. 
\end{proposition}

\begin{lemma} \label{secondclaim}
Let $(G,\mathcal{X},k)$ be a Yes-instance of \textsc{VDP} where $G$ is a split graph. Moreover, let $M$, $B$, and $H$ be as described in Construction $\mathcal{A}$. Let $R$ $(\subseteq B)$ be the set of vertices in $I_{N}$ that are not saturated by $M$ in graph $H$, and let $F=I_{N}\setminus{R}$. Let $\mathcal{P}$ be a minimum solution of $(G,\mathcal{X},k)$ for which the number of vertices in $R$ and $F$ visited by the paths in $\mathcal{P}$ is minimum and maximum, respectively (with priority to the minimization). Then, every vertex in $F$ is visited by some path in $\mathcal{P}$. 
\end{lemma}

\begin{proof}

Let $v\in F$ such that no path in $\mathcal{P}$ uses $v$ as an internal vertex. First, by Observation \ref{new:obs}, any set induced by $\mathcal{P}$ in $H$ is a matching in $H$. Let $M'$ be one such matching. Note that $M'$ does not saturate any vertex in $R$ (due to Proposition \ref{firstclaim}), and furthermore, $M'$ does not saturate $v$ in $H$ (by the definition of $M'$). This implies that $|M'|<|F|$. Since $|M|=|F|$, note that $M'$ is not a maximum matching in $H$. Thus, by Proposition \ref{maxprop} and the fact that none of $M$ or $M'$ saturates a vertex from $R$, there exists an $M'$-augmenting path, say, $Q$, in $H-R$. For an example, consider $G$, $H$, and $\mathcal{X}$ given in Figure \ref{bpgfig}. If $M=\{v_{ab}^{1}f,v_{ab}^{2}g,v_{de}^{1}h\}$ and $M'=\{v_{ab}^{1}f,v_{de}^{1}g\}$, then $(v_{ab}^{2},f,v_{ab}^{1},$ $g,v_{de}^{1},h)$ is an $M'$-augmenting path in $H-R$. 

Next, we discuss how to obtain a solution $\mathcal{\widehat{P}}$ of $(G,\mathcal{X},k)$ from $\mathcal{P}$ by using $Q$, such that $\mathcal{\widehat{P}}$ visits one more vertex from $F$ than $\mathcal{P}$ and no vertex from $R$ (contradicting the choice of $\mathcal{P}$). Let $M''$ be the matching obtained from $M'$ by replacing the saturated edges with unsaturated edges in $Q$ and vice versa (by augmenting $M'$ with respect to $Q$). For example, consider $G$, $H$, and $\mathcal{X}$ given in Figure \ref{bpgfig}. If $M=\{v_{ab}^{1}f,v_{ab}^{2}g,v_{de}^{1}h\}$ and $M'=\{v_{ab}^{1}f,v_{de}^{1}g\}$, then $M''=\{v_{de}^{1}h,v_{ab}^{1}g,v_{ab}^{2}f\}$. Since the length of an augmenting path is always odd, one of the endpoints of $Q$ must belong to $B$. This implies that $M''$ saturates one more vertex in set $F$ (as $Q$ is an augmenting path in $H-R$) than $M'$. Initialize $\mathcal{\widehat{P}}$ as $\mathcal{P}$. Remove all those paths from $\mathcal{\widehat{P}}$ that contain vertices from set $V(Q)\cap F$ as internal vertices. Furthermore, if $v_{xy}^{j}$  for some $j\in [w-1]$ is the end point of $Q$ in $A$, then also remove the path in $\mathcal{P}$ between the $j^{th}$ occurrence of terminal pair $(x,y)\in \mathcal{X}$. Next, for every edge $v_{st}^{j}b\in M''\cap Q$, where $v_{st}^{j}\in A$ for some $j\in [w-1]$ and $b\in F$, introduce the path $(s,b,t)$ in $\mathcal{\widehat{P}}$. Next, note that $\mathcal{\widehat{P}}$ is a minimum solution of $(G,\mathcal{X},k)$. Furthermore, the paths in $\mathcal{\widehat{P}}$ visit one more vertex of $F$ (and no vertex from $R$) than the paths in $\mathcal{P}$, which contradicts the choice of $\mathcal{P}$. Thus, every vertex in $F$ is visited by some path in $\mathcal{P}$. 
\end{proof}

We are now ready to establish the promised equivalence.

\begin{lemma} \label{newlemma1}
Let $(G,\mathcal{X},k)$ be a Yes-instance of \textsc{VDP} where $G$ is a split graph. Moreover, let $M$, $B$, $\mathcal{\widehat{X}}$, and $H$ be as described
in Construction $\mathcal{A}$. Then, the output $(G',\mathcal{X'},k')$ of \textsc{Clean-Up} on $(G,\mathcal{X},k)$ is also a Yes-instance of \textsc{VDP}. 
\end{lemma}
\begin{proof}

Let $(G,\mathcal{X},k)$ be a Yes-instance, and let $R$ $(\subseteq B)$ be the set of vertices in $I_{N}$ that are not saturated by $M$ in graph $H$, and let $F=I_{N}\setminus{R}$. Let $\mathcal{P}$ be a solution of $(G,\mathcal{X},k)$ for which the number of vertices in $R$ visited by the paths in $\mathcal{P}$ is minimum, and the number of vertices in $F$ visited by the paths in $\mathcal{P}$ is maximum. Now, by Proposition \ref{firstclaim}, none of the paths in $\mathcal{P}$ visits a vertex in $R$. Furthermore, by Lemma \ref{secondclaim}, for every $v\in F$, there exists a path, say, $P_v$, in $\mathcal{P}$ that visits $v$.  Therefore, if $\mathcal{\widehat{P}}$ denotes the set of internally vertex-disjoint paths between the terminals in $\mathcal{\widehat{X}}$, then note that none of the paths in $\mathcal{P}\setminus \mathcal{\widehat{P}}$ visits a  vertex in the set $I_{N}$. Thus, $\mathcal{P}\setminus \mathcal{\widehat{P}}$ is a solution of $(G',\mathcal{X'},k')$. 
\end{proof}

\begin{lemma}  \label{newlemma}
Let $(G,\mathcal{X},k)$ be an instance of \textsc{VDP} where $G$ is a split graph. Moreover, let $\mathcal{\widehat{X}}$ be as described
in Construction $\mathcal{A}$. Furthermore, let $(G',\mathcal{X'},k')$ be the output of \textsc{Clean-Up} on $(G,\mathcal{X},k)$. Then, if $(G',\mathcal{X'},k')$ is a Yes-instance of \textsc{VDP}, then $(G,\mathcal{X},k)$ is also a Yes-instance of \textsc{VDP}.
\end{lemma}

\begin{proof}

Let $(G',\mathcal{X'},k')$ be a Yes-instance, and let $\mathcal{P'}$ be one of its solutions. Note that, in order to prove that $(G,\mathcal{X},k)$ is a Yes-instance, we only need to find additional internally vertex-disjoint paths between the terminal pairs in $\widehat{\mathcal{X}}$. By the definition of $\widehat{\mathcal{X}}$, we know that for every occurrence of a terminal pair, say, $(s_{i},t_{i})$, in $\widehat{\mathcal{X}}$, there exists a vertex, say, $u_{i}$, in $I_{N}$ such that $u_{i}$ is adjacent to both $s_{i}$ and $t_{i}$ in $G$. Also, for distinct $i\in[k]$ such that $(s_{i},t_{i})\in \mathcal{\widehat{X}}$, we have distinct $u_{i}$ (i.e., $u_{i}\neq u_{j}$ for $i\neq j\in [k]$). Furthermore, since $G'=G-I_{N}$, it is clear that $u_i$ does not belong to any path in $\mathcal{P'}$. Therefore, for every $(s_{i},t_{i})\in \widehat{\mathcal{X}} $, define $P_{i}=(s_{i},u_{i},t_{i})$. Note that $\{P_{i}:{(s_{i},t_{i})\in \widehat{\mathcal{X}}}\}$ is a set of internally vertex-disjoint paths in $G$. Thus, $\mathcal{P}=\mathcal{P'}\cup \{P_{i}:{(s_{i},t_{i})\in \widehat{\mathcal{X}}}\}$ is a solution of $(G,\mathcal{X},k)$.
\end{proof}

Now, consider the following lemma.

\begin{lemma} \label{rmrk14}
  Let $(G,\mathcal{X},k)$ be an instance of \textsc{VDP} obtained by applying \textsc{Clean-Up}. Let  $(s,t)\in \mathcal{X}$ be a heavy terminal pair of Type-II in $G$ of weight $w\geq2$. Then, any minimum solution of $(G,\mathcal{X},k)$ must contain the following internally vertex-disjoint paths: $\{P_{st}\}\cup \{(s,u^{1},t),\ldots,(s,u^{w-1},t)\}$, where $\{u^{1},\ldots,u^{w-1}\}$ is a set of non-terminal vertices in $C$. 
\end{lemma} 
\begin{proof}
First, note that $P_{st}$ must belong to every minimum solution of $(G,\mathcal{X},k)$ (due to Observation \ref{obsminimum}). Moreover, by Proposition \ref{prop1}, note that the internally vertex-disjoint paths between the remaining $w-1$ terminal pairs $(s,t)$ must have length exactly 2. Next, observe that after applying the \textsc{Clean-Up} operation, $I_{N}=\emptyset$. By the definition of internally vertex-disjoint paths, vertices in $I_{T}$ cannot be used as internal vertices in any path of the solution; so, the only choice is that the internal vertices are from $C$.
\end{proof}

Now, we are ready to define the reduction rules.
\medskip

\noindent \textbf{Reduction Rules.} Let us start by defining our first reduction rule (RR\ref{RR1}).

\begin{RR}[RR\ref{RR1}]\label{RR1} If there is a terminal pair $(s,t)\in \mathcal{X}$ of Type-I such that $st\in E(G)$, then $V(G')\Leftarrow V(G)$, $E(G') \Leftarrow E(G)\setminus \{st\}$, $\mathcal{X'} \Leftarrow \mathcal{X}\setminus\{(s,t)\}$, $k' \Leftarrow k-1$. Furthermore, for every $x\in \{s,t\}$ that does not appear as a terminal in any terminal pair in $\mathcal{X'}$, update $V(G')\Leftarrow V(G')\setminus \{x\}$.  
\end{RR}

We have the following lemma to establish that RR\ref{RR1} is safe.

\begin{lemma} \label{rr1}
RR\ref{RR1} is safe. 
\end{lemma}
\begin{proof}

 First, note that since we are only removing an edge between a vertex of $C$ and a vertex of $I$, $G'$ is a split graph. Next, we claim that $(G,\mathcal{X},k)$ and $(G',\mathcal{X'},k')$ are equivalent instances of \textsc{VDP} on split graphs.
 
In one direction, let $\mathcal{P}$ be a minimum solution of $(G,\mathcal{X},k)$. By Observation \ref{obsminimum}, $P_{st}\in \mathcal{P}$. Since terminal vertices cannot be used as internal vertices in any path of the solution, note that $x\in \{s,t\}$ can only appear as an endpoint of the paths in $\mathcal{P}\setminus \{P_{st}\}$ (if $x$ is a terminal in some terminal pair in $\mathcal{X'}$). Thus, $\mathcal{P}\setminus \{P_{st}\}$ is a solution of $(G',\mathcal{X'},k')$.
 
In the other direction, let $(G',\mathcal{X'},k')$ be a Yes-instance of \textsc{VDP}. This implies that there exists a set $\mathcal{P'}$ of $k-1$ internally vertex-disjoint paths in $G'$ joining the terminal pairs in $\mathcal{X'}$. Since $st\notin E(G')$, none of the paths in $\mathcal{P'}$ uses the edge $st$. Furthermore, note that if $x\in\{s,t\}$ does not appear as a terminal in any terminal pair in $\mathcal{X'}$, then it is necessary to remove $x$ from $G'$ as otherwise; there may exist a path in $\mathcal{P'}$ that uses $x$ as an internal vertex, which risks us an invalid solution. Thus, after taking care of $x\in\{s,t\}$ accordingly, note that $\mathcal{P'}\cup \{P_{st}\}$ is a set of $k$ internally vertex-disjoint paths in $G$ joining the terminal pairs in $\mathcal{X}$. 
\end{proof}

\begin{observation} \label{rmrk4}
 After applying RR\ref{RR1} exhaustively on $G$, no Type-I terminal pair in $G$ has an edge between its terminals. Moreover, RR\ref{RR1} can be applied in polynomial
time.
\end{observation}

To define the next reduction rule (RR\ref{RR2}), we use the following notation: Let $\{(s,t)\times (w)\}$ denote $w$ copies of $(s,t)$.

\begin{RR}[RR\ref{RR2}]\label{RR2}  If there is a heavy terminal pair $(s,t)\in \mathcal{X}$ of Type-II in $G$ of weight $w\geq 2$, then $V(G') \Leftarrow V(G) \cup \{ s^{1},\ldots,s^{w-1},t^{1},\ldots,t^{w-1}\}$, $E(G')\Leftarrow E(G)\cup \{s^{i}v,t^{i}v :v\in C, i\in [w-1]\}$, $\mathcal{X'}\Leftarrow (\mathcal{X} \setminus \mathcal{\overline{X}})\cup \{(s^{1},t^{1}),\ldots,(s^{w-1},t^{w-1})\}$, where $\mathcal{\overline{X}}=\{(s,t)\times (w-1)\}\subseteq \mathcal{X}$.
\end{RR}

We have the following lemma to establish that RR\ref{RR2} is safe.

\begin{lemma} \label{rr2}
After \textsc{Clean-Up} and the exhaustive application of RR\ref{RR1}, RR\ref{RR2} is safe. 
\end{lemma}
\begin{proof}
First, note that since we are adding an independent set of vertices in $I$, $G'$ is a split graph. Next, we claim that $(G,\mathcal{X},k)$ and $(G',\mathcal{X'},k)$ are equivalent instances of \textsc{VDP} on split graphs.

In one direction, let $\mathcal{P}$ be a minimum solution of $(G,\mathcal{X},k)$. By Observation \ref{obsminimum}, we note that $P_{st}\in \mathcal{P}$. Let $\mathcal{\widehat{P}}\subseteq \mathcal{P}$ such that $\mathcal{\widehat{P}}$ contains internally vertex-disjoint paths between the terminal pairs in $\overline{\mathcal{X}}$. Observe that the paths in $\mathcal{\widehat{P}}$ must be of the form $(s,u,t)$, where $u$ is some non-terminal vertex in $C$ (due to Lemma \ref{rmrk14}). In other words, exactly $w-1$ non-terminal vertices in $C$ are reserved to be used by the paths in $\mathcal{\widehat{P}}$. Let $\{u^{1},\ldots,u^{w-1}\}$ denote such a set. Next, we define a solution $\mathcal{P'}$ of $(G',\mathcal{X'},k)$ as follows: Let $\mathcal{P'}=(\mathcal{P}\setminus \mathcal{\widehat{P}})\cup \{(s^{1},u^{1},t^{1}),\ldots,(s^{w-1},u^{w-1},t^{w-1})\}$. Since none of $\{u^{1},\ldots,u^{w-1}\}$ belongs to any path in $\mathcal{P}\setminus \mathcal{\widehat{P}}$, $\mathcal{P'}$ is a solution of $(G',\mathcal{X'},k)$.

In the other direction, let $\mathcal{P'}$ be a minimum solution of $(G',\mathcal{X'},k)$. Let $\mathcal{\widehat{P}}\subseteq \mathcal{P'}$ such that $\mathcal{\widehat{P}}$ contains internally vertex-disjoint paths between the terminal pairs $\{(s^{1},t^{1}),\ldots,(s^{w-1},t^{w-1})\}$. Observe that the paths in $\mathcal{\widehat{P}}$ must be of the form $(s^{i},u,t^{i})$, where $u$ is some non-terminal vertex in $C$ (due to Proposition \ref{prop1}). Next, we define a solution $\mathcal{P}$ of $(G,\mathcal{X},k)$ as follows: Let $\mathcal{P}=(\mathcal{P'}\setminus \mathcal{\widehat{P}})\cup \{(s,u^{1},t),\ldots,(s,u^{w-1},t)\}$, where $\{u^{1},\ldots,u^{w-1}\}$ denotes the set of non-terminal vertices in $C$ that belong to the paths in $\mathcal{\widehat{P}}$. Since none of $\{u^{1},\ldots,u^{w-1}\}$ belongs to any path in $\mathcal{P}\setminus \mathcal{\widehat{P}}$, $\mathcal{P}$ is a solution of $(G,\mathcal{X},k)$.
\end{proof}

\begin{observation} \label{rmrk3}
 After applying RR\ref{RR2} exhaustively, there do not exist any heavy Type-II terminal pairs. Moreover, RR\ref{RR2} can be applied in polynomial
time.
\end{observation}

We apply the next reduction rule (RR\ref{RR3}) for every terminal participating in more than one terminal pair.

\begin{RR}[RR\ref{RR3}]\label{RR3}  If $v\in V(G)$ belongs to $x\geq 2$ terminal pairs $(v,a_1),\ldots,(v,a_x)$, then $V(G')\Leftarrow (V(G)\setminus \{v\}) \cup \{v_{1},\ldots,v_{x}\}$, $E(G')\Leftarrow E(G)\cup \{v_{i}u: u\in N(v), i\in [x]\}$, $\mathcal{X'}\Leftarrow (\mathcal{X} \setminus \{(v,a_1),\ldots,(v,a_x)\})\cup \{(v_1,a_1),\ldots,(v_x,a_x)\}$. Moreover, if $v\in C$, then $E(G')=E(G')\cup \{v_{i}v_{j}:i\neq j\in [x]\}$.
 \end{RR} 

We have the following lemma to establish that RR\ref{RR3} is safe.

\begin{lemma} \label{rr3}
 After \textsc{Clean-Up} and the exhaustive application of  RR\ref{RR1} and RR\ref{RR2}, RR\ref{RR3} is safe. 
\end{lemma}\begin{proof}
First, note that $G'$ is a split graph because if $v\in I$, then we are adding an independent set of vertices in $I$, and if $v\in C$, then we are adding a clique in $C$. Next, we claim that $(G,\mathcal{X},k)$ and $(G',\mathcal{X'},k)$ are equivalent instances of \textsc{VDP} on split graphs.

In one direction, let $\mathcal{P}$ be a minimum solution of $(G, \mathcal{X},k)$. Let  $\mathcal{\widehat{P}}=\{P_{1},\ldots,P_{x}\}\subseteq \mathcal{P}$ be the subset of internally vertex-disjoint paths of $\mathcal{P}$ such that $P_{i}$ is a path between $v$ and $a_{i}$ for each $i\in [x]$. 
Note that there are two possibilities: one where $v\in I_{T}$, and the other where $v\in C$. First, assume that $v\in I_{T}$. On the one hand, if $(v,a_{i})$ is of Type-I (i.e., $a_{i}\in C$) for some $i\in [x]$, then $va_{i}\notin E(G)$ (by Observation \ref{rmrk4}). In this case, due to Proposition \ref{prop1}, $P_{i}$ must be of the form $(v,u_{i},a_{i})$, where $u_{i}$ is some non-terminal vertex in $C$. On the other hand, if $(v,a_{i})$ is of Type-III (i.e., $a_{i}\in I_{T}$), then due to Proposition \ref{prop1} $(i)$, $P_{i}$ must be either of the form $(v,u_{i},a_{i})$ or $(v,u_{i},u_{j},a_{i})$, where $u_{i}, u_{j}$ are some non-terminal vertices in $C$. Let $\mathcal{P^{*}}$ be the set of internally vertex-disjoint paths obtained from $\mathcal{\widehat{P}}$ by replacing $v$ with $v_{i}$ in each $P_{i}, i\in [x]$. We claim that $\mathcal{P'}=(\mathcal{P}\setminus \mathcal{\widehat{P}})\cup \mathcal{P^*}$ is a solution of $(G',\mathcal{X'},k)$. Since the paths in $\mathcal{\widehat{P}}$ are internally vertex-disjoint, the collection of all internal vertices of the paths in $\mathcal{\widehat{P}}$ are distinct. This implies that all the internal vertices of the paths in $\mathcal{P^{*}}$ are distinct, and thus the paths in $\mathcal{P'}$ are internally vertex-disjoint.

Now, assume that $v\in C$. On the one hand, if $(v,a_i)$ is of Type-I for some $i\in [x]$, then $P_{i}$ must be of the form $(v,u_{i},a_{i})$, where $u_{i}$ is some non-terminal vertex in $C$ (due to Observation \ref{rmrk4} and Proposition \ref{prop1}). On the other hand, if $(v,a_i)$ is of Type-II, then $P_{i}$ must be of the form $va_{i}$ (due to Observation \ref{rmrk3} and Proposition \ref{prop1}). Let $\mathcal{P^{*}}$ be the set of internally vertex-disjoint paths obtained from $\mathcal{\widehat{P}}$ by replacing $v$ with $v_{i}$ in each $P_{i}, i\in [x]$. We claim that $\mathcal{P'}=(\mathcal{P}\setminus \mathcal{\widehat{P}})\cup \mathcal{P^*}$ is a solution of $(G',\mathcal{X'},k)$. Since the paths in $\mathcal{\widehat{P}}$ are internally vertex-disjoint, the collection of all internal vertices of the paths in $\mathcal{\widehat{P}}$ are distinct. This implies that all the internal vertices of the paths in $\mathcal{P^{*}}$ are distinct, and thus the paths in $\mathcal{P'}$ are internally vertex-disjoint.

In the other direction, let $\mathcal{P'}$ be a solution of $(G', \mathcal{X'},k)$. Let  $\mathcal{\widehat{P}}=\{P_{1},\ldots,P_{x}\}\subseteq \mathcal{P'}$ be the subset of internally vertex-disjoint paths of $\mathcal{P'}$ such that $P_{i}$ is a path between $v_{i}$ and $a_{i}$ for each $i\in [x]$. Due to Observations \ref{rmrk4} and \ref{rmrk3}, note that either $a_{i}\neq a_{j}$ for every distinct $i,j\in [x]$ or if $a_{i}=a_{j}$ for distinct $i,j\in [x]$, then $va_{i}\notin E(G)$. Let $\mathcal{P^{*}}$ be the set of internally vertex-disjoint paths obtained from $\mathcal{\widehat{P}}$ by replacing $v_{i}$ with $v$ in each $P_{i}, i\in [x]$. We claim that $\mathcal{P}=(\mathcal{P'}\setminus \mathcal{\widehat{P}})\cup \mathcal{P^*}$ is a solution of $(G,\mathcal{X},k)$. Since the paths in $\mathcal{\widehat{P}}$ are internally vertex-disjoint, the collection of all internal vertices of the paths in $\mathcal{\widehat{P}}$ are distinct. This implies that all the internal vertices in the paths in $\mathcal{P^{*}}$ are distinct, and thus the paths in $\mathcal{P}$ are internally vertex-disjoint.
\end{proof}

By Lemma \ref{rmrk14} and Observations \ref{rmrk4} and \ref{rmrk3}, we have the following lemma.

\begin{observation} \label{rmrk1}
After applying reduction rules RR\ref{RR1}-RR\ref{RR3} exhaustively, no terminal participates in more than one terminal pair. Moreover, the reduction rules, RR\ref{RR1}-RR\ref{RR3}, can be applied in polynomial time.
\end{observation}

Before concluding this section, we need the following result by Yang et al.~\cite{yang}.
 \begin{proposition} [\cite{yang}] \label{labelyangprop}
\textsc{VDP-Unique} on split graphs admits a kernel with at most $4k$ vertices, where $k$ is the number of occurrences of terminal pairs. 
\end{proposition}

By Observations \ref{cleanupobs} and \ref{rmrk1}, we note that every instance $(G, \mathcal{X},k)$ of \textsc{VDP} where $G$ is a split graph, can be converted to an instance $(G', \mathcal{X'},k')$ of \textsc{VDP-Unique} in polynomial time. Due to Observations \ref{rmrk4} and \ref{rmrk3} and Lemmas \ref{newlemma1}, \ref{newlemma}, \ref{rr1}, \ref{rr2}, and \ref{rr3}, $(G,\mathcal{X},k)$ and $(G',\mathcal{X'},k)$ are equivalent. Furthermore, note that throughout this section, our initial parameter $k$ does not increase during the application of the \textsc{Clean-Up} operation and rules RR\ref{RR1}-RR\ref{RR3}. Therefore, using Proposition \ref{labelyangprop}, we have the following theorem.

\splitvdp*

\subsection{A Quadratic Vertex Kernel for Well-partitioned Chordal Graphs} \label{quadraticvdpwpc}

In this section, we show that $\textsc{VDP}$ on well-partitioned chordal graphs admits a kernel with $\mathcal{O}(k^{2})$ vertices. Let $(G,\mathcal{X},k)$ be an instance of \textsc{VDP} where $G$ is a well-partitioned chordal graph. A brief overview of our kernelization algorithm is given below.
\medskip

\noindent \textbf{Overview.} Ahn et al.~\cite{ahn1} showed that \textsc{VDP} on well-partitioned chordal graphs admits a kernel with $\mathcal{O}(k^{3})$ vertices. Based on their algorithm, we design our kernelization algorithm for \textsc{VDP} on well-partitioned chordal graphs. Since most of our rules (in this section) are borrowed from~\cite{ahn1}, our contribution can be viewed as an improved analysis of their algorithm.  

Note that given a well-partitioned chordal graph $G$, by Proposition \ref{ahn:prop}, we can compute (in polynomial time) a partition tree $\mathcal{T}$ of $G$. By Proposition \ref{prop1} and utilizing the properties of the partition tree $\mathcal{T}$, first, we define a marking procedure (called \textsc{Marking Procedure} by us) that marks at most $\mathcal{O}(k^{2})$ vertices in $V(G)$. Our marking procedure uses the following classification of terminal pairs. For a terminal pair $(s,t)\in \mathcal{X}$, either $st\in E(G)$ or $st\notin E(G)$. Furthermore, if $st\in E(G)$, then $(s,t)$ is either a heavy terminal pair or a light terminal pair (see Definition \ref{def:heavy}). Accordingly, we have two different sets of rules to mark the vertices in $V(G)$: one for heavy terminal pairs and one for (every occurrence of) non-heavy terminal pairs (defined below). After marking the desired vertices, we show (in Lemmas \ref{lm2} and \ref{lmheavy}) that if our input instance $(G,\mathcal{X},k)$ is a Yes-instance, then there exists a solution of $(G,\mathcal{X},k)$ that uses the vertices from the marked vertices only (as internal vertices). 

\medskip
We start the technical part of this section with the following definitions.

\begin{definition} [Non-heavy Terminal Pairs]
Let $(G,\mathcal{X},k)$ be an instance of \textsc{VDP} where $G$ is a well-partitioned chordal graph. A terminal pair $(s,t)\in \mathcal{X}$ is \emph{non-heavy} if either $st\notin E(G)$ or $(s,t)$ is light.
\end{definition}
\begin{definition} [Valid Path] \label{def:valid}
Let $(G,\mathcal{X},k)$ be an instance of \textsc{VDP} where $G$ is a well-partitioned chordal graph with a partition tree $\mathcal{T}$.
Also, for a terminal pair $(s,t)\in \mathcal{X}$, let $s\in V(B)$ and $t\in V(B')$, where $B,B'\in V(\mathcal{T})$. Then, the unique path $(B,\ldots,B')$ in $\mathcal{T}$ is the \emph{valid path} corresponding to $(s,t)$. 
\end{definition} 

\begin{remark}
In Definition \ref{def:valid}, it is possible that $B=B'$. In this case, the valid path consists of the single bag $B$.
\end{remark}

\begin{definition} [Active Boundary] \label{def:active}
Let $(G,\mathcal{X},k)$ be an instance of \textsc{VDP} where $G$ is a well-partitioned chordal graph with a partition tree $\mathcal{T}$. For a bag $B\in V(\mathcal{T})$ and a bag $B'\in N_{\mathcal{T}}(B)$, the boundary $\mathsf{bd}(B,B')$ is \emph{active} if both $B$ and $B'$ belong to the valid path corresponding to some non-heavy terminal pair $(s,t)\in \mathcal{X}$. Furthermore, we say that $(s,t)$ \emph{activates} $\mathsf{bd}(B,B')$.
\end{definition}

\begin{remark}
In Definition \ref{def:active}, note that a terminal pair $(s,t)$ such that $s,t\in V(B)$ for some $B\in V(\mathcal{T})$ does not activate any boundary. Also, note that the notion of active boundary is not defined for any heavy terminal pair. 
\end{remark}

For an illustration of valid paths and active boundaries, see Figure \ref{Fig:wpc}.

\begin{figure}
    \centering
    \includegraphics[scale=0.65]{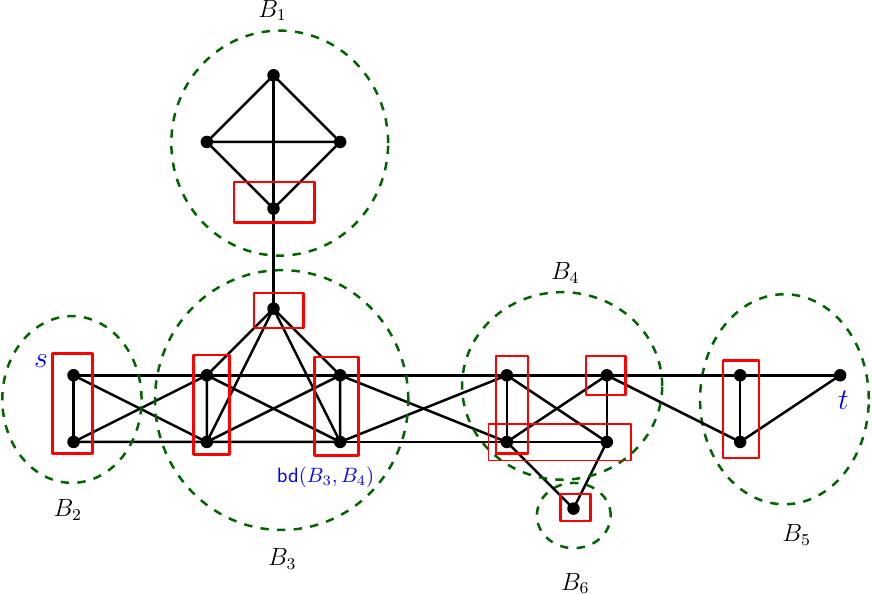}
    \caption{ For the terminal pair $(s,t)$, the valid path is $(B_{2},B_{3},B_{4},B_{5})$. Bag $B_{3}$ has three boundaries, viz. $\mathsf{bd(B_{3},B_{2})}$, $\mathsf{bd(B_{3},B_{1})}$, and $\mathsf{bd(B_{3},B_{4})}$. Assuming that $\mathcal{X}=\{(s,t)\}$, only two of them ($\mathsf{bd(B_{3},B_{2})}$, $\mathsf{bd(B_{3},B_{4})}$) are active.}
    \label{Fig:wpc}
\end{figure}

Additionally, we have the following straightforward observation that follows from Definitions \ref{def:valid} and \ref{def:active}. In particular, note that a single (non-heavy) terminal pair can activate at most two boundaries in every bag.

\begin{observation} \label{obsk}
Let $(G,\mathcal{X},k)$ be an instance of \textsc{VDP} where $G$ is a well-partitioned chordal graph with a partition tree $\mathcal{T}$. For a bag $B$ in $\mathcal{T}$, let $\mathsf{bd}(B,B_{1}),\ldots,\mathsf{bd}(B,B_{q})$ denote the active boundaries in $B$. Furthermore, for each $i\in [q]$, let $k_{B_i}$ denote the number of occurrences of terminal pairs that activate $\mathsf{bd}(B,B_{i})$. Then, $\sum_{i=1}^{q}{k_{B_i}}\leq 2k$. In other words, any bag can have at most $2k$ active boundaries.
\end{observation}

Consider the following reduction rule.

\begin{RR}[RR\ref{RR4}]\label{RR4} Let $(s,t)\in \mathcal{X}$ be a light terminal pair. Then, $G'\Leftarrow G$, $\mathcal{X'} \Leftarrow \mathcal{X} \setminus \{(s,t)\}$, $k'\Leftarrow k-1$.  Furthermore, for every $x\in \{s,t\}$ that does not appear as a terminal in any terminal pair in $\mathcal{X'}$, update $V(G')\Leftarrow V(G')\setminus \{x\}$.   
\end{RR}

\begin{proposition} [\cite{ahn1}] \label{rr4}
RR\ref{RR4} is safe. 
\end{proposition}

\begin{observation} \label{poly:rr4}
RR\ref{RR4} can be applied in polynomial
time.
\end{observation}

Using the fact that every bag in $\mathcal{T}$ is a clique, we have the following proposition. 

\begin{proposition} [\cite{ahn1}] \label{obs3}
Let $(G,\mathcal{X},k)$ be an instance of \textsc{VDP} where $G$ is a well-partitioned chordal graph with a partition tree $\mathcal{T}$, obtained after the exhaustive application of RR\ref{RR4}. Furthermore, let $\mathcal{\widehat{X}}\subseteq \mathcal{X}$ denote the multiset of non-heavy terminal pairs.  Let $\mathcal{P}$ be a minimum solution of $(G,\mathcal{X},k)$, and let $\mathcal{\widehat{P}}\subseteq \mathcal{P}$ such that $\mathcal{\widehat{P}}$ contains the internally vertex-disjoint paths between the terminal pairs in  $\mathcal{\widehat{X}}$. Then, for every $P\in \mathcal{\widehat{P}}$ and every bag $B\in V(\mathcal{T})$, $|V (P) \cap V(B)| \leq 2$. \end{proposition} 

Note that Proposition \ref{obs3} is not true (in general) for heavy terminal pairs as in a minimum solution, the vertices in the path corresponding to an occurrence of a heavy terminal pair can have $3$ vertices in common with the vertices in a bag (when the terminals themselves belong to that bag).

The following proposition is used to characterize the structure of minimum solutions (for non-heavy terminal pairs) in terms of valid paths.

\begin{proposition} [\cite{ahn1}] \label{lm1} 
 Let $(G,\mathcal{X},k)$ be a Yes-instance of \textsc{VDP} where $G$ is a well-partitioned chordal graph, obtained after the exhaustive application of RR\ref{RR4}. Furthermore, let $\mathcal{T}$ be a partition tree of $G$, and let $\mathcal{\widehat{X}}\subseteq \mathcal{X}$ denote the multiset of non-heavy terminal pairs. Let $\mathcal{P}$ be a minimum solution of $(G,\mathcal{\widehat{X}},|\mathcal{\widehat{X}}|)$. Let $(B_{1},\ldots,B_{\ell})$ be the valid path corresponding to some $(s,t)\in \mathcal{\widehat{X}}$, and let $P\in \mathcal{P}$ be the path between $s$ and $t$. Then, $V(P)\cap V(B)\neq \emptyset$ if and only if $B\in \{B_{1},\ldots,B_{\ell}\}$.
\end{proposition}

\begin{remark}
Let $V(\mathcal{X})=\{s,t: (s,t)\in \mathcal{X}\}$ throughout this section.
\end{remark}

Next, we describe our marking procedure, which is crucial to proceed further. To this end, let $(G,\mathcal{X},k)$ be an instance of \textsc{VDP} where $G$ is a well-partitioned chordal graph with a partition tree $\mathcal{T}$ of $G$, obtained after the exhaustive application of RR\ref{RR4}. Let $\mathcal{P}$ be a minimum solution of $(G,\mathcal{X},k)$. Before formally defining the marking procedure, let us give some intuition as to how to choose which vertices to mark. Note that the light terminal pairs have already been eliminated by RR\ref{RR4}. To deal with the remaining non-heavy terminal pairs, we make use of Propositions \ref{obs3} and \ref{lm1} to characterize the paths in a minimum solution between every occurrence of a non-heavy terminal pair, and, accordingly, choose which vertices to mark. Next, let $(s,t)\in \mathcal{X}$ be a heavy terminal pair. By Proposition \ref{prop1} and Observation \ref{obsminimum}, for every occurrence of $(s,t)$, the $(s,t)$-path in $\mathcal{P}$ must be of the form of either $P_{st}$ or $(s,v,t)$, where $v\in N(s)\cap N(t)$.  Since $st\in E(G)$, there are two possible cases (based on the positions of $s$ and $t$ in $\mathcal{T}$). First, let there exist bags $B, B' \in V(\mathcal{T})$ such that $s\in \mathsf{bd}(B, B')$ and $t \in \mathsf{bd}(B',B)$ (note that $B\in N_\mathcal{T}(B')$). In this case, the only vertices that can be used as internal vertices to form length-2 $(s,t)$-paths have to be contained in $\mathsf{bd}(B, B') \cup \mathsf{bd}(B', B)$. Second, let there exist a bag $B\in V(\mathcal{T})$ such that $\{s,t\} \subseteq V(B)$. In this case, the only vertices that can be used as internal vertices to form length-2 $(s,t)$-paths have to be contained in $B$ or in a bag $B'\in N_{\mathcal{T}}(B)$ such that $s,t\in \mathsf{bd}(B,B')$.
\smallskip

Now, we are ready to formally define the \textsc{Marking Procedure}. Note that it is necessary to first perform the \textsc{Marking Procedure} for every occurrence of a non-heavy terminal pair and afterward for every heavy terminal pair. 
\smallskip

\noindent \textbf{\textsc{Marking Procedure}:} 
Let $(G,\mathcal{X},k)$ be an instance of \textsc{VDP} where $G$ is a well-partitioned chordal graph with a partition tree $\mathcal{T}$, obtained after the exhaustive application of RR\ref{RR4}. For two adjacent bags, say, $B$ and $B'$, in $\mathcal{T}$, let $k_{BB'}$ denote the number of occurrences of terminal pairs that activate $\mathsf{bd}(B,B')$. Further, let $\mathsf{\widetilde{bd}}(B,B')\subseteq \mathsf{bd}(B,B')$ that do not belong to any other active boundary of $B$.

\begin{enumerate}
    \item For every occurrence of a non-heavy terminal pair $(s,t)\in \mathcal{X}$ and for every heavy terminal pair $(s',t')\in \mathcal{X}$: Initialize $M_{(s,t)} \Leftarrow \emptyset$ and $M_{(s',t')} \Leftarrow \emptyset$.
 
    \item For every bag $B\in \mathcal{T}$ that has at least one active boundary: 
    \begin{enumerate}
        \item [2.1] Let $\mathsf{bd}(B,B_{1}),\ldots, \mathsf{bd}(B,B_{q})$ be active in $B$. 
        \item [2.2] For each $i \in [q]$ and for every $(s,t)\in \mathcal{X}$ that activates $\mathsf{bd}(B,B_i)$:
         \begin{enumerate}
        \item [2.2.1] If $|\mathsf{\widetilde{bd}}(B,B_i)|\geq 2k_{BB_i}$, then add to $M_{(s,t)}$ a maximal\footnote{Here, when we say that we pick a maximal subset of some set $A$ of size at most $a$, we mean that: (i) if $|A|\geq a$, then, we pick some subset of size $a$ of $A$; (ii) otherwise, we simply pick $A$.} subset of $\mathsf{\widetilde{bd}}(B,B_i)\setminus V(\mathcal{X})$ of size at most $2k_{BB_i}$. 
        
       \item [2.2.2] Else, add to $M_{(s,t)}$ the set $\mathsf{bd}(B,B_i)$.
       
    \end{enumerate}
    \end{enumerate}
   
 \item  For every heavy terminal pair $(s,t)\in \mathcal{X}$:
 \begin{enumerate}
     \item [3.1] If there exist bags $B, B' \in V(\mathcal{T})$ such that $s\in \mathsf{bd}(B, B')$ and $t \in \mathsf{bd}(B',B)$, then add to $M_{(s,t)}$ a maximal subset of $(\mathsf{bd}(B, B') \cup \mathsf{bd}(B', B))\setminus V(\mathcal{X})$ of size at most $2k$, with preference to unmarked vertices.
     \item [3.2] If there exists a bag $B\in V(\mathcal{T})$ such that $\{s,t\} \subseteq V(B)$, then add to $M_{(s,t)}$ a maximal subset of $(B\cup \bigcup_{B' \in N_{\mathcal{T}}(B), \{s,t\}\subseteq \mathsf{bd}(B,B')} \mathsf{bd}(B',B)) \setminus V(\mathcal{X})$ of size at most $2k$, with preference to unmarked vertices.
 \end{enumerate}
 \item For every bag $B\in \mathcal{T}$ that has at least one active boundary:
 
\begin{enumerate}
        \item [4.1] Let $\mathsf{bd}(B,B_{1}),\ldots, \mathsf{bd}(B,B_{q})$ be active in $B$.
 \item [4.2] Let $F_{i}=\mathsf{bd(B,B_{i})} \setminus \mathsf{\widetilde{bd}(B,B_{i})}$ for all $i\in [q]$.
 \item [4.3] For every fixed $i\in [q]$:
 
 \begin{enumerate}
     \item [4.3.1] Let $\widehat{B}_i^{1},\ldots,\widehat{B}_i^{p}$ be neighboring bags of $B$ such that $F_{i}\cap\mathsf{bd}(B,\widehat{B}_{i}^{j})\neq \emptyset$ and $\mathsf{bd}(\widehat{B}_{i}^{j},B)$ are active for all $j\in [p]$.
 \item [4.3.2] If $|F_{i}|\geq 2k_{B\widehat{B}_{i}^{1}}+\ldots+2k_{B\widehat{B}_{i}^{p}}$, then arbitrarily keep only $2k_{B\widehat{B}_{i}^{j}}$ distinct vertices in the marked sets of terminal pairs that activate $\mathsf{bd}(B,\widehat{B}_{i}^{j})$ for all $j\in[p]$, and unmark all remaining vertices in $F_{i}$.
  \end{enumerate}
  \end{enumerate}
\end{enumerate}

This completes our \textsc{Marking Procedure}. 

\begin{observation} \label{obsmark11}
\textsc{Marking Procedure} can be executed in polynomial time.
\end{observation}

Next, we have the following lemma.

\begin{lemma} \label{mark:lemma}
Let $(G,\mathcal{X},k)$ be an instance of \textsc{VDP} where $G$ is a well-partitioned chordal graph obtained after the exhaustive application of RR\ref{RR4}. Furthermore, let $\mathcal{T}$ be a partition tree of $G$, and let $\mathcal{\widehat{X}}\subseteq \mathcal{X}$ denote the multiset of non-heavy terminal pairs. Let $I=|\mathcal{\widehat{X}}|$. Let $M_1, \ldots, M_I \subseteq V(G)$ be sets of marked vertices obtained by applying \textsc{Marking Procedure} to $(G,\mathcal{\widehat{X}},I)$. Then, for each bag $B \in V (\mathcal{T})$, $|V(B) \cap \bigcup_{i\in [I]} M_i| \in \mathcal{O}(k)$.
\end{lemma}

\begin{proof}

Consider a bag $B$ in $\mathcal{T}$. Let $\mathsf{bd}(B,B_{1}),\ldots,\mathsf{bd}(B,B_{q})$ denote the active boundaries in $B$. For each $i\in [q]$, let $p_{i}$ denote the number of marked vertices in $\mathsf{bd}(B,B_{i})$. Note that if for every $i\in[q]$, for which $\mathsf{bd}(B,B_{i})\subseteq M_{(s,t)}$, we have $|\mathsf{bd}(B,B_{i})|\leq 2k_{BB_{i}}$ (here,  $(s,t)$ is an arbitrary terminal pair that activates $\mathsf{bd}(B,B_{i})$), then
by Observation \ref{obsk} and the description of \textsc{Marking Procedure}, the total marked vertices in $B$ are at most  $\sum_{i=1}^{q}{p_{i}}\leq 4k$.

Now consider the case when, for some arbitrary but fixed $i\in[q]$, we have $\mathsf{bd}(B,B_{i})\subseteq M_{(s,t)}$ and $|\mathsf{bd}(B,B_{i})|> 2k_{BB_{i}}$ (here,  $(s,t)$ is an arbitrary terminal pair that activates $\mathsf{bd}(B,B_{i})$). This implies that there exist vertices in $\mathsf{bd}(B,B_{i})$ that are shared by more than one active boundary. Consider the set $F_i\subseteq \mathsf{bd}(B,{B_{i}})$ that is shared by boundaries, say, $B_{1},\ldots,B_{p}$. If $|F_i|\geq  2k_{BB_{1}}+\ldots+2k_{BB_{p}}$, then note that we can uniquely assign the vertices arbitrarily to each active boundary sharing them and unmark the remaining vertices in $F_i$. Otherwise, it is clear that $|F_i|<2k_{BB_{1}}+\ldots+2k_{BB_{p}}$.
\end{proof}

The following lemma helps us to establish that if $(G,\mathcal{\widehat{X}},I)$ (the instance restricted to the multiset of non-heavy terminal pairs) is a Yes-instance, then there exists a solution of $(G,\mathcal{\widehat{X}},I)$ such that the paths among the non-heavy terminal pairs use only the marked vertices (obtained by applying \textsc{Marking Procedure}) as internal vertices. To prove the lemma,  we build on the proof given by Ahn et al.~\cite{ahn1}. 

\begin{lemma} \label{lm2}
  Let $(G,\mathcal{X},k)$ be an instance of \textsc{VDP} where $G$ is a well-partitioned chordal graph obtained after the exhaustive application of RR\ref{RR4}. Furthermore, let $\mathcal{T}$ be a partition tree of $G$, and let $\mathcal{\widehat{X}}\subseteq \mathcal{X}$ denote the multiset of non-heavy terminal pairs. Let $I=|\mathcal{\widehat{X}}|$. If
 $(G,\mathcal{\widehat{X}},I)$ is a Yes-instance, then there exists a minimum solution $\mathcal{P}$ of $(G,\mathcal{\widehat{X}},I)$ such that if $P\in \mathcal{P}$ denotes the path between one occurrence of $(s,t)\in \mathcal{\widehat{X}}$ and $M_{(s,t)} \subseteq V(G)$ denotes the set of marked vertices obtained by applying \textsc{Marking Procedure}, then $ V(P) \subseteq M_{(s,t)} \cup \{s,t\}$.
\end{lemma}

 \begin{proof}
 Let $(G,\mathcal{\widehat{X}},I)$ be a Yes-instance. For every occurrence of an arbitrary (but fixed) terminal pair $(s,t)\in \mathcal{\widehat{X}}$, let $M_{(s,t)}$ denote the set of marked vertices obtained by applying \textsc{Marking Procedure}. Furthermore, let $\mathcal{P}$ be a minimum solution such that $V(P)$ has maximum vertex intersection with $ M_{(s,t)} \cup \{s,t\}$, where $P\in \mathcal{P}$ denotes the path between one occurrence of $(s,t)$. Since $(s,t)$ is a non-heavy terminal pair, $P$ is induced (due to Proposition \ref{prop1}). If $V(P) \subseteq M_{(s,t)} \cup \{s,t\}$, then we are done. So, assume otherwise. This implies that there exists some $B\in V(\mathcal{T})$ such that $V(P)\cap V(B) \nsubseteq M_{(s,t)} \cup \{s,t\}$. Let $(B_1, \ldots, B_{\ell})$ be the valid path in $\mathcal{T}$ corresponding to this particular occurrence of the terminal pair $(s,t)$. By Proposition \ref{lm1}, note that $B$ is a bag on the valid path $(B_1, \ldots, B_{\ell})$. 
 
From now on, let $j \in [\ell]$ be such that $B=B_j$, and let $Y= (V(P)\cap V(B_j))\setminus \{s,t\}$. Here, note that by the definition of the \textsc{Marking Procedure}, it is clear that every vertex in $Y$ either does not belong to any other active boundary, or has been unmarked later. By Proposition \ref{obs3}, $|Y|\leq 2$. Furthermore, for every $j'\in [\ell]\setminus \{j\}$, let $k_{B_jB_{j'}}$ denote the number of occurrences of terminal pairs that activate $\mathsf{bd}(B_{j},B_{j'})$. Now, consider the following cases based on the position of $j$ in the path $(B_1, \ldots, B_{\ell})$.
\smallskip

\noindent \textbf{Case 1:} $\bm{j=1.}$ In this case, note that $|Y|=1$ (as $|Y|\leq 2$ and $s\notin Y$). Furthermore, note that $s \notin  \mathsf{bd}(B_1,B_2)$ (otherwise $P$ will not be an induced path). Now, let $y \in Y \setminus M_{(s,t)}$. Note that $y \in \mathsf{bd(B_1,B_2)}$. Since $y$ was not marked, it is clear that we have not marked the entire set $\mathsf{bd}(B_1, B_2)$ for $M_{(s,t)}$. This implies that there are $2k_{B_1B_2}$ marked vertices in $\mathsf{bd}(B_1, B_2)$ that are reserved only for the paths going from $B_1$ to $B_2$. Since the paths in $\mathcal{P}\setminus \{P\}$ use at most $2(k_{B_1B_2}-1)$ vertices in total from $\mathsf{bd}(B_1, B_2)$ (due to Propositions \ref{obs3} and \ref{lm1}), there is at least one vertex, say, $y'$, in $\mathsf{bd}(B_1, B_2)\cap M_{(s,t)}$ that is not used by any path in $\mathcal{P}\setminus \{P\}$. 

Define $P'$ by replacing $y$ with $y'$ in $P$. Note that $P'$ is an induced $(s,t)$-path in $G$ (as $y$ and $y'$ are twins in $G[V(B_{1})\cup V(B_{2})]$). Furthermore, since $y'\notin \bigcup_{P\in \mathcal{P}} V(P)$, $P'$ is internally vertex-disjoint from every path in $\mathcal{P}\setminus \{P\}$. Let $\mathcal{P'}=(\mathcal{P}\setminus \{P\})\cup \{P'\}$. Since $| V(P')\cap (M_{(s,t)} \cup \{s,t\})|>|V(P)\cap (M_{(s,t)} \cup \{s,t\})|$, this contradicts our choice of $\mathcal{P}$. Hence, this case is not possible.
\smallskip

\noindent \textbf{Case 2:} $\bm{2\leq j < \ell.}$
First, let $|Y \setminus M_{(s,t)}|=2$ and let $y_1,y_2\in Y \setminus M_{(s,t)}$. Without loss of generality, assume that $y_1 \in \mathsf{bd}(B_j,B_{j-1})$ and $y_2\in \mathsf{bd}(B_j,B_{j+1})$. Due to Propositions \ref{obs3} and \ref{lm1}, note that the paths in $\mathcal{P}\setminus \{P\}$ use at most $2(k_{B_jB_{j-1}}-1)$ vertices from $\mathsf{bd}(B_{j},B_{j-1})$ and at most $2(k_{B_jB_{j+1}}-1)$ vertices from $\mathsf{bd}(B_{j},B_{j+1})$ 
(since we have not marked the entire sets $\mathsf{bd}(B_j, B_{j-1})$ and $\mathsf{bd}(B_j, B_{j+1})$ for $M_{(s,t)}$, this implies that $2k_{B_jB_{j-1}}$ marked vertices in $\mathsf{bd}(B_j, B_{j-1})$ are reserved only for the paths going from $B_{j-1}$ to $B_j$ and $2k_{B_jB_{j+1}}$ marked vertices in $\mathsf{bd}(B_j, B_{j+1})$ are reserved only for the paths going from $B_{j}$ to $B_{j+1}$).
This implies that there exist $y'_1 \in \mathsf{bd}(B_j,B_{j-1})\cap M_{(s,t)}$ and $y'_{2} \in \mathsf{bd}(B_j,B_{j+1})\cap M_{(s,t)}$ such that neither $y'_1$ nor $y'_2$ is used by any path in $\mathcal{P}\setminus \{P\}$. Note that since $y_{1}$ and $y'_{1}$ are twins in $G[V(B_{j-1})\cup V(B_{j})]$ and $y_{2}$ and $y'_{2}$ are twins in $G[V(B_{j})\cup V(B_{j+1})]$, we can define a path $P'$ by replacing $y_1$ with $y'_1$ and $y_2$ with $y'_2$ in $P$ such that $P'$ is also an induced path. Since $y'_{1},y'_{2}\notin \bigcup_{P\in \mathcal{P}} V(P)$, $P'$ is internally vertex-disjoint from every path in $\mathcal{P}\setminus \{P\}$. Let $\mathcal{P'}=(\mathcal{P}\setminus \{P\})\cup \{P'\}$. Since $| V(P')\cap (M_{(s,t)} \cup \{s,t\})|>|V(P)\cap (M_{(s,t)} \cup \{s,t\})|$, this contradicts our choice of $\mathcal{P}$. Hence, this case is not possible. 

Second, let $|Y \setminus M_{(s,t)}| = 1$. This case is similar to the case where $|Y \setminus M_{(s,t)}|=2$. Note that there are two possibilities here: let $y\in Y \setminus M_{(s,t)}$, then either $y\in \mathsf{bd}(B_{j},B_{j-1})\cap \mathsf{bd}(B_{j},B_{j+1})$ or we can assume without loss of generality that $y\in \mathsf{bd}(B_{j},B_{j-1}) $ and $y\notin \mathsf{bd}(B_{j},B_{j+1})$. In either case, it is possible to find a marked vertex, say, $y'$, such that we can replace $y$ with $y'$ in $P$ to obtain another minimum solution that has one more vertex in its vertex set among the marked vertices leading to a contradiction to the choice of $\mathcal{P}$. Hence, this case is also not possible. 
\smallskip

\noindent \textbf{Case 3:} $\bm{j=\ell.}$ This case is symmetric to Case 1: Replace $B_{1}$ with $B_{\ell}$ and $s$ with $t$. Hence, this case is not possible. 
 \end{proof}

The next lemma helps us to establish that if $(G,\mathcal{X},k)$ is a Yes-instance and $(s,t)\in \mathcal{X}$ is a heavy terminal pair, then there exists a solution of $(G,\mathcal{X},k)$ that uses only the marked vertices obtained by applying the \textsc{Marking Procedure} (as internal vertices) for every occurrence of $(s,t)$.

\begin{lemma} \label{lmheavy}
  Let $(G,\mathcal{X},k)$ be an instance of \textsc{VDP} where $G$ is a well-partitioned chordal graph obtained after the exhaustive application of RR\ref{RR4}. Furthermore, let $\mathcal{T}$ be a partition tree of $G$. If 
 $(G,\mathcal{X},k)$ is a Yes-instance, then there exists a minimum solution $\mathcal{P}$ of $(G,\mathcal{X},k)$ such that if $\mathcal{P'}\subseteq \mathcal{P}$ denotes a set of internally vertex-disjoint paths for every occurrence of a heavy terminal pair $(s,t)$ and $M_{(s,t)} \subseteq V(G)$ denotes the set of marked vertices obtained by applying \textsc{Marking Procedure}, then $ \bigcup_{P\in \mathcal{P'}} V(P)  \subseteq M_{(s,t)} \cup \{s,t\}$.
\end{lemma}
\begin{proof}
 Let $(G,\mathcal{X},k)$ be a Yes-instance, and let $\mathcal{P}$ be a minimum solution of $(G,\mathcal{X},k)$. For an arbitrary (but fixed) heavy terminal pair $(s,t)\in \mathcal{X}$, let $M_{(s,t)}$ denote the set of marked vertices obtained by applying the \textsc{Marking Procedure} (note that $M_{(s,t)}$ contains marked vertices for every occurrence of $(s,t)$).  Furthermore, let $\mathcal{P'}\subseteq \mathcal{P}$ denote the set of internally vertex-disjoint paths for every occurrence of $(s,t)$ such that $\bigcup_{P\in \mathcal{P'}} V(P)$ has maximum vertex intersection with $ M_{(s,t)} \cup \{s,t\}$.   If $ \bigcup_{P\in \mathcal{P'}} V(P)  \subseteq M_{(s,t)} \cup \{s,t\}$, then we are done. So, assume otherwise. This implies that there exists some $P\in \mathcal{P'}$ such that $V(P) \nsubseteq M_{(s,t)} \cup \{s,t\}$. By Proposition \ref{prop1} and Observation \ref{obsminimum}, the paths in $\mathcal{P'}$ are of the form $P_{st}$ or $(s,v,t)$, where $v\in N(s)\cap N(t)$. Since $P$ has length two, this further implies that there exists a vertex, say, $v\in V(P)$, such that $v\notin M_{(s,t)}$. Since the paths in $\mathcal{P}\setminus \{P\}$ use at most $2(k-1)$ vertices in total from $M_{(s,t)}$, there is at least one vertex, say, $v'$, in $M_{(s,t)}$ that is not used by any path in $\mathcal{P}$. Define $P'$ by replacing $v$ with $v'$ in $P$. Furthermore, since $v'\notin \bigcup_{P\in \mathcal{P'}} V(P)$, $P'$ is internally vertex-disjoint from every path in $\mathcal{P'}\setminus \{P\}$. Let $\mathcal{P'}=(\mathcal{P}\setminus \{P\})\cup \{P'\}$. Since $|\bigcup_{P\in \mathcal{P'}} V(P)\cap (M_{(s,t)} \cup \{s,t\})|>|\bigcup_{P\in \mathcal{P}} V(P)\cap (M_{(s,t)} \cup \{s,t\})|$, this contradicts our choice of $\mathcal{P}$. Hence, $ \bigcup_{P\in \mathcal{P'}} V(P)  \subseteq M_{(s,t)} \cup \{s,t\}$.
\end{proof}

\begin{remark}
From now onwards (throughout this section), let $M_{1}$ and $M_{2}$ denote the set of marked vertices obtained by applying the \textsc{Marking Procedure} to each occurrence of every non-heavy terminal pair and each heavy terminal pair, respectively.
\end{remark}

Before proceeding further, let us discuss in brief what we are planning to do next. Due to \textsc{Marking Procedure} and Lemmas \ref{lm2} and \ref{lmheavy}, we know that the bags in $\mathcal{T}$ that do not contain any marked vertex can be removed without changing the answer to our input instance. Furthermore, note that we mark only $\mathcal{O}(k^{2})$ vertices for all heavy terminal pairs (combined). However, for non-heavy terminal pairs, we have only established that the number of marked vertices in each bag is $\mathcal{O}(k)$; we still need to bound the number of bags that have a non-empty intersection with $M_{1}$. For this purpose, consider the following definition.
\begin{definition} [Well-Partitioned Forest] \label{PF}
Let $(G,\mathcal{X},k)$ be an instance of \textsc{VDP} where $G$ is a well-partitioned chordal graph obtained after the exhaustive application of RR\ref{RR4}. Furthermore, let $\mathcal{T}$ be a partition tree of $G$. After applying \textsc{Marking Procedure}, the subgraph of  $\mathcal{T}$ induced by all bags with a nonempty intersection with $M_{1}$ is called a \emph{well-partitioned forest} of $(G,\mathcal{X},k)$.
\end{definition}

The next observation follows from the fact that a well-partitioned forest consists of the union of at most $k$ paths.
\begin{observation} \label{forest}
Let $(G,\mathcal{X},k)$ be an instance of \textsc{VDP} where $G$ is a well-partitioned chordal graph obtained after the exhaustive application of RR\ref{RR4}. Furthermore, let $\mathcal{T'}$ be a well-partitioned forest of $(G,\mathcal{X},k)$ as described in Definition \ref{PF}. Then, $\mathcal{T'}$ has at most $2k$ bags of degree one.
\end{observation}

By Observation \ref{forest} and the fact that the number of vertices of degree at least three in a forest is always bounded by the number of vertices of degree 1, it is left to bound the number of degree 2 bags in $\mathcal{T'}$. The next reduction rules (RR\ref{RR5} and RR\ref{RR6}) help us to remove bags of degree 2 in $\mathcal{T'}$ with some additional properties without changing the answer of the given instance. The reduction rules RR\ref{RR5} and RR\ref{RR6} were also used by Ahn et
al.~\cite{ahn1}. However, the condition that $V(B)\cap M_{2}=\emptyset$ is necessary in RR\ref{RR6}, which was not there. Furthermore, note that the arguments given in Lemma \ref{ok:lemma} are crucial.
 
\begin{RR}[RR\ref{RR5}]\label{RR5}  Let $(G,\mathcal{X},k)$ be an instance of \textsc{VDP} where $G$ is a well-partitioned chordal graph obtained after the exhaustive application of RR\ref{RR4}. Furthermore, let $\mathcal{T'}$ be a well-partitioned forest of $G$. Let $B \in V(\mathcal{T'})$ such that $V(B) \cap V (\mathcal{X}) = \emptyset$, $d_{\mathcal{T'}}(B) = 2$, and $V(B)\cap M_{2}=\emptyset$. Let $A$ and $C$ be the two neighbors of $B$ in $\mathcal{T'}$. Let $k_{BA}$ and $k_{BC}$ denote the number of occurrences of terminal pairs that activate $\mathsf{bd}(B,A)$ and $\mathsf{bd}(B,C)$, respectively. If $|\mathsf{bd}(B,A)| < k_{BA}$ or
$|\mathsf{bd}(B,C)| <k_{BC}$, then answer negatively.
\end{RR}

\begin{RR}[RR\ref{RR6}]\label{RR6}  Let $(G,\mathcal{X},k)$ be an instance of \textsc{VDP} where $G$ is a well-partitioned chordal graph obtained after the exhaustive application of RR\ref{RR4}. Furthermore, let $\mathcal{T'}$ be a well-partitioned forest of $G$. Let $B \in V(\mathcal{T'})$ such that $V(B) \cap V (\mathcal{X}) = \emptyset$, $d_{\mathcal{T'}}(B) = 2$, and $V(B)\cap M_{2}=\emptyset$. Let $A$ and $C$ be the two neighbors of $B$ in $\mathcal{T'}$. Let $k_{BA}$ and $k_{BC}$ denote the number of occurrences of terminal pairs that activate $\mathsf{bd}(B,A)$ and $\mathsf{bd}(B,C)$, respectively. If $|\mathsf{bd}(B,A)| \geq k_{BA}$ and
$|\mathsf{bd}(B,C)| \geq k_{BC}$, then reduce $(G,\mathcal{X},k)$ to $(G',\mathcal{X},k)$, where $G'$ is obtained from $G$ by removing $B$
and making all vertices in $\mathsf{bd}(A,B)$ adjacent to all vertices in $\mathsf{bd}(C,B)$.
\end{RR}

The following proposition establishes that RR\ref{RR5} and RR\ref{RR6} are safe.

\begin{proposition} [\cite{ahn1}] \label{rr56}
RR\ref{RR5} and RR\ref{RR6} are safe. 
\end{proposition}

\begin{observation} \label{poly:rr56}
RR\ref{RR5} and RR\ref{RR6} can be applied in polynomial
time.
\end{observation}

Now, we have our final lemma.
\begin{lemma} \label{ok:lemma}
 Let $(G,\mathcal{X},k)$ be an instance of \textsc{VDP} where $G$ is a well-partitioned chordal graph obtained after the exhaustive application of RR\ref{RR4}. Furthermore, let $\mathcal{T'}$ be the well-partitioned forest of $G$ obtained after applying \textsc{Marking Procedure} followed by an exhaustive application of RR\ref{RR5} and RR\ref{RR6} on a partition tree $\mathcal{T}$ of $G$. Then, the number of degree 2 bags in $\mathcal{T'}$ is $\mathcal{O}(k)$.
\end{lemma}
\begin{proof}
Note that after applying RR\ref{RR5} and RR\ref{RR6} exhaustively, if we do not answer negatively and there exists a bag, say, $B$, with degree 2 in $\mathcal{T'}$, then either $V(B)\cap V(\mathcal{X})\neq \emptyset$ or $V(B)\cap M_{2}\neq \emptyset$. Since $|V(\mathcal{X})|\leq2k$, the number of degree 2 bags having a non-empty intersection with $V(\mathcal{X})$ is at most $2k$. Next, let $V(B)\cap V(\mathcal{X})= \emptyset$ and $V(B)\cap M_{2}\neq \emptyset$ for some bag $B\in V(\mathcal{T'})$. Then, due to the description of the \textsc{Marking Procedure}, at least one of the neighboring bags of $B$ in $\mathcal{T'}$ must contain a heavy terminal pair (as $V(B)\cap V(\mathcal{X})=\emptyset$ and $V(B)\cap M_{2}\neq \emptyset$). Since there can be at most $\big\lfloor\frac{k}{2}\big\rfloor$ heavy terminal pairs in $\mathcal{X}$, there are at most $\big\lfloor\frac{k}{2}\big\rfloor$ bags with $V(B)\cap M_{2}\neq \emptyset$ and $V(B)\cap V(\mathcal{X})= \emptyset$ in $\mathcal{T'}$.
Thus, the lemma holds.
\end{proof}

Note that throughout this section, our initial parameter $k$ does not increase during the application of reduction rules RR\ref{RR4}-RR\ref{RR6}. Therefore, using the \textsc{Marking Procedure}, Lemmas \ref{mark:lemma}-\ref{lmheavy}, \ref{ok:lemma}, Propositions \ref{rr4} and \ref{rr56}, and Observations \ref{poly:rr4}, \ref{obsmark11}, and \ref{poly:rr56}, we have the following theorem.

\wpcvdp*

\section{Polynomial-time Algorithm for Threshold Graphs}
\label{sec:threshold-poly}

Unlike the case of \textsc{EDP}, the \textsc{VDP} problem becomes easy on highly restricted graph classes.
While it remains \textsf{NP}-hard on split graphs, we will show that it becomes polynomial-time solvable on threshold graphs.
Recall that a split graph with a split partition $(C,I)$ is threshold if we can order the vertices of $I$ as $v_1, v_2, \dots, v_{|I|}$ such that $N(v_1) \subseteq N(v_2) \subseteq \dots \subseteq N(v_{|I|})$.
On the intuitive level, $v_1$ is the ``weakest'' vertex of the graph, so the paths starting at $v_1$ need to be processed first.
We present a greedy argument stating that if $v_1$ is a terminal, then we can freely allocate the paths starting at $v_1$ without ``worrying'' about the other paths, and then remove $v_1$ from the graph.
Throughout this section, we abbreviate $(G,\mathcal{X},k)$ as simply $(G,\mathcal{X})$ when referring to a \textsc{VDP} instance because we do not need to keep track of the parameter.

Now, let us begin with the following lemma.

\begin{lemma}\label{lem:threshold-VDP-greedy}
    Let $(G,\mathcal{X})$ be an instance of \textsc{VDP} such that $G$ is a threshold graph with a split partition $V(G) = (C,I)$.
    Let $v \in I$ be a vertex for which $N_G(v) \subseteq N_G(w)$ for every $w \in I$ and
    $\mathcal{X}_v \subseteq \mathcal{X}$ be the subset of pairs containing $v$.
    Suppose that $\mathcal{X}_v \ne \emptyset$. Then, in polynomial time, we either detect that $(G,\mathcal{X})$ is a No-instance, or
    compute an equivalent instance $(G', \mathcal{X} \setminus \mathcal{X}_v)$ such that $G'$ is an induced subgraph of $G-v$. 
\end{lemma}
\begin{proof}
    Let $Y$ be the set of vertices, different than $v$, that appear in at least one pair in $\mathcal{X}_v$ and $T$ be the set of vertices appearing in $\mathcal{X} \setminus \mathcal{X}_v$.
    
    Let $\widehat{\mathcal{X}}_v \subseteq \mathcal{X}_v$ be defined as the {\em set}
    $\{\{v,y\} \mid y \in N_G(v) \cap Y\}$
    (if there are several identical elements in $\mathcal{X}_v$, then we choose one of them).
    We set $\ell = |\mathcal{X}_v \setminus \widehat{\mathcal{X}}_v|$ where $\mathcal{X}_v \setminus \widehat{\mathcal{X}}_v$ is considered a multiset.
    
    If $|N_G(v) \setminus (Y \cup T)| < \ell$,
    then $N_G(v)$ cannot accommodate all paths starting at $v$, so we can report that $(G,\mathcal{X})$ is a No-instance.
    Otherwise, there exists an injective mapping 
    $\tau \colon [\ell] \to N_G(v) \setminus (Y \cup T)$; we fix an arbitrary one.
    
    We construct a solution $\mathcal{P}_v$ to the instance $(G,\mathcal{X}_v)$ as follows. 
    For each $u \in N_G(v) \cap Y$, add the path $P_{vu} = (v,u)$ to $\mathcal{P}_v$;
    this resolves the pairs in $\widehat{\mathcal{X}}_v$.
    Next, for each 
    $i \in [\ell]$ we consider the terminal pair $\{s_i,t_i\} \in  \mathcal{X}_v \setminus \widehat{\mathcal{X}}_v$ (where $s_i = v$)
    and insert the path $(v, \tau(i), t_i)$ to $\mathcal{P}_v$.
    Observe that due to the choice of $v$, the set $N_G(v) \setminus \{t_i\}$ is contained in $N_G(t_i)$ so the edge between $\tau(i)$ and $t_i$ is present in $G$.

    Let $G'$ be obtained from $G$ by removing the vertex set $(\bigcup_{P \in \mathcal{P}_v} V(P)) \setminus T$.
    Note that this set contains $v$.
    We will show that $(G', \mathcal{X} \setminus \mathcal{X}_v)$ is equivalent to $(G, \mathcal{X})$.
    The first implication is straightforward: when $\mathcal{P}'$ is a solution to $(G', \mathcal{X} \setminus \mathcal{X}_v)$ then $\mathcal{P}' \cup \mathcal{P}_v$
    is a family of internally vertex-disjoint paths that forms a solution to $(G, \mathcal{X})$.
    In order to show the second implication, we argue that $\mathcal{P}_v$ is part of some solution.

    We claim that if $(G,\mathcal{X})$ is a Yes-instance, then there exists a solution $\mathcal{P}$ such that $\mathcal{P}_v \subseteq \mathcal{P}$.
    Let $\mathcal{P}$ be a minimum solution that additionally maximizes the number of used paths that belong to $\mathcal{P}_v$.
    Suppose that $\mathcal{P}_v \not\subseteq \mathcal{P}$ and let $P$ a path from $\mathcal{P}_v \setminus \mathcal{P}$.
    By \cref{obsminimum}
    we know that every single-edge path corresponding to a pair in $\widehat{\mathcal{X}}_v$ is present in $\mathcal{P}$, so $P$
    must be of the form $(v, \tau(i), t_i)$ for some $i \in [\ell]$.
    Observe that an induced path in a threshold graph can have length at most two, so by \cref{prop1} every path in $\mathcal{P}$ visits at most three vertices.
    By the choice of $\mathcal{P}$, there exists some path $Q \in \mathcal{P} \setminus \mathcal{P}_v$ which uses the vertex $u = \tau(i)$; recall that $u \not\in Y \cup T$, so $u$ is an internal vertex of $Q$.
    Then $Q$ must be of the form $(x,u,y)$ where $x,y \in Y \cup T$.
    Let $u' \in N_G(v) \setminus (Y \cup T)$ be a vertex visited by some $(v,t_i)$-path $P'$ from $\mathcal{P} \setminus \mathcal{P}_v$ (we know that $P'$ exists because some path must be a replacement for $P$). Observe that due to choice of $v$, the vertex $u'$ belongs to both $N_G(x)$ and $N_G(y)$.
    Therefore, we can replace $Q$ with $(x,u',y)$ 
    and replace $P'$ with $P$, obtaining a new valid minimum solution.
    This solution uses more paths from $\mathcal{P}_v$ than $\mathcal{P}$, yielding a contradiction.
    
    Let $\mathcal{P}$ be the solution to 
    $(G,\mathcal{X})$ satisfying $\mathcal{P}_v \subseteq \mathcal{P}$.
    Then $\mathcal{P} \setminus \mathcal{P}_v$ cannot use any internal vertex from any path in $\mathcal{P}_v$, nor any terminal vertex from $\mathcal{P}_v$ that does not belong to $T$.
    Hence  $\mathcal{P} \setminus \mathcal{P}_v$ is a solution to $(G', \mathcal{X} \setminus \mathcal{X}_v)$.
    The lemma follows.
\end{proof}

We can now utilize the greedy procedure to repeatedly remove the vertices from the independent set and reduce the graph to a clique.
It remains to ensure that the weakest vertex is a terminal, so \cref{lem:threshold-VDP-greedy} can be applied.
To this end, we take advantage of the {\sc Clean-Up} operation from \cref{linearvdpsplit}.

\thresholdPoly*
\begin{proof}
    We begin with applying the {\sc Clean-Up} operation
    (\cref{CU}) to the given instance, which does not affect the answer (Lemmas \ref{newlemma1} and \ref{newlemma}) and does not increase the graph size.
    Since the class of threshold graphs is closed under vertex deletion, we obtain a threshold graph as well.
    Let $(G,\mathcal{X})$ denote the instance after performing {\sc Clean-Up} and let $(C,I)$ be the split partition of $G$.
    Due to such preprocessing, we can assume that every vertex from $I$ appears in $\mathcal{X}$.
    Let $v \in I$ be a vertex for which $N_G(v) \subseteq N_G(w)$ for every $w \in I$ (its existence is guaranteed by the definition of a threshold graph) and
    $\mathcal{X}_v \subseteq \mathcal{X}$ be the subset of pairs containing $v$.
    We apply the reduction from \cref{lem:threshold-VDP-greedy} and obtain an equivalent instance $(G', \mathcal{X} \setminus \mathcal{X}_v)$ such that $G'$ is an induced subgraph of $G-v$.
    As a consequence, $G'$ is again a threshold graph, 
    and it has fewer vertices than $G$.

    We iterate the procedure described above as long as the independent part of the graph is non-empty.
    At every iteration, the size of the graph decreases, so at some point, the process terminates.
    We either arrive at a No-instance (and report it as specified in \cref{lem:threshold-VDP-greedy}) or
    obtain a graph being a clique.
    Such an instance can be solved easily: it suffices to check whether the number of terminal pairs that cannot be resolved by single-edge paths (due to repetitions of terminal pairs) does not exceed the number of non-terminal vertices. 
    This concludes the proof.
    \end{proof}

\section{Conclusion}\label{sec:conclusion} 
In this paper, we studied \textsc{VDP} and \textsc{EDP}, two disjoint paths problems in the realm of Parameterized Complexity. We analyzed these problems with respect to the natural parameter ``the number of (occurrences of) terminal pairs''. We gave several improved kernelization results as well as new kernelization results for these problems on subclasses of chordal graphs. 

For \VDP, we provided a $4k$ vertex kernel on split graphs and an $\mathcal{O}(k^2)$ vertex kernel on well-partitioned chordal graphs. We also show that \VDP~becomes polynomial-time solvable on threshold graphs.
For \EDP, we first proved \textsf{NP}-hardness on complete graphs. Second, we provided an $\mathcal{O}(k^{2.75})$ vertex kernel on split graphs, a $7k+1$ vertex kernel on threshold graphs, an $\mathcal{O}(k^2)$ vertex kernel for \EDP~on block graphs, and a $2k+1$ vertex kernel on clique paths.

Apart from the obvious open questions to improve the sizes of the kernels we designed, the following is a natural next step for future work.

\begin{question}\label{Q:chordal}
Do \VDP~or/and \EDP~admit polynomial kernels for chordal graphs?
\end{question}

It is worth noting here that Golovach et al.~\cite{golovach2022parameterized} proved that it is unlikely for \textsc{Set-Restricted Disjoint Paths}, a generalization of \VDP~where each terminal pair has to find its path from a predefined set of vertices, to admit a polynomial kernel even on interval graphs. However, as noted by them, their reduction is heavily dependent on the sets designed for terminal pairs and thus cannot be directly generalized to \VDP. Moreover, recently W{\l}odarczyk and Zehavi~\cite{michalEDP} established that both \VDP~and \EDP~are unlikely to admit polynomial compression even when restricted to planar graphs. They also suggested to investigate the existence of polynomial kernels for chordal graphs for \VDP~and \EDP.

Another interesting open problem is to study the kernelization complexity of \EDP~on well-partitioned chordal graphs. 
Note that \EDP~on well-partitioned chordal graphs is more subtle than \textsc{VDP} on the same. The reason is that, in \VDP, a path between a non-adjacent terminal pair must be induced due to Proposition \ref{prop1}.
This paves the way to define valid paths in the partition tree of the given well-partitioned chordal graph. However, the concept of valid paths and, thus, the \textsc{Marking Procedure} (from Section \ref{quadraticvdpwpc}) fails for \textsc{EDP}, as a path in the solution can visit the bags of the partition tree in any weird manner.
Furthermore, the approach for \textsc{EDP} on block graphs does not generalize to well-partitioned chordal graphs because the intersection of adjacent bags can be large (in well-partitioned chordal graphs),
whereas, for a block graph $G$, there always exists a partition tree such that for any two consecutive bags, the corresponding boundary of one of the bags has size one.

\bibliography{main.bib}

\end{document}